\documentclass[11pt]{article}
\usepackage{fullpage}
\usepackage[margin=1in]{geometry}

\usepackage{graphicx}
\usepackage{wrapfig}

\usepackage{xcolor}
\usepackage{hyperref}
\usepackage[linesnumbered,ruled]{algorithm2e}
\usepackage[noend]{algpseudocode}
\usepackage[normalem]{ulem}
\usepackage{multicol}
\usepackage{multirow}
\usepackage{hhline}
\usepackage{amsmath}
\usepackage{mathtools}
\usepackage[capitalize,noabbrev]{cleveref}
\usepackage{float}
\usepackage{array}
\usepackage{xspace}
\usepackage{subfigure}
\newcolumntype{L}[1]{>{\raggedright\let\newline\\\arraybackslash\hspace{0pt}}m{#1}}
\newcolumntype{C}[1]{>{\centering\let\newline\\\arraybackslash\hspace{0pt}}m{#1}}
\newcolumntype{R}[1]{>{\raggedleft\let\newline\\\arraybackslash\hspace{0pt}}m{#1}}

\usepackage{amsthm} 

\RequirePackage[T1]{fontenc} 
\RequirePackage[tt=false, type1=true]{libertine} 
\RequirePackage[varqu]{zi4} 
\RequirePackage[libertine]{newtxmath}

\usepackage{nicefrac}
\usepackage{array}
\usepackage{makecell}

\newtheorem{definition}{\textbf{Definition}}
\newtheorem{assumption}{\textbf{Assumption}}
\newtheorem{remark}{\textbf{Remark}}

\newtheorem{lemma}{\textbf{Lemma}}
\newtheorem{theorem}{\textbf{Theorem}}
\newtheorem{proposition}{\textbf{Proposition}}
\newtheorem{corollary}{\textbf{Corollary}}
\newtheorem{claim}{\textbf{Claim}}

\newcommand{\calt}{\mathcal{T}}
\newcommand{\cala}{\mathcal{A}}

\newcommand{\opt}{\texttt{OPT}\xspace}
\newcommand{\alg}{\texttt{ALG}\xspace}
\newcommand{\oja}{\textsc{flp-$\phi$-variable}\xspace}

\newcommand{\hx}{x}

\newcommand{\pr}{\mathbb{P}}

\DeclareMathOperator*{\argmax}{arg\,max}
\DeclareMathOperator*{\argmin}{arg\,min}

\newcommand{\problemkRentalVD}{\textsf{kRental-Variable}\xspace}
\newcommand{\problemkRentalFD}{\textsf{kRental-Fixed}\xspace}

\newcommand{\gocrv}{$\gamma$-\textsc{ocr-v}\xspace}

\newcommand{\ocr}{\textsc{ocr}\xspace}
\newcommand{\gocr}{$\gamma$-\ocr}
\newcommand{\locr}{1-\ocr}

\newcommand{\algkRentalFD}{\textsc{dop-$\phi$-fixed}\xspace}

\newcommand{\algorithmVD}{\textsc{dop-$\phi$-variable}\xspace}

\newcommand{\algorithmVDfrac}{\textsc{dop-$\phi$-variable-fractional}\xspace}

\newcommand{\indicator}[1]{\mathbb{I}_{\{#1\}}}

\usepackage{soul}
\sethlcolor{yellow}

\usepackage{todonotes}
\newcounter{todotan}

\newcounter{todobo}

\begin{document}
\title{Online Rounding Schemes for $ k $-Rental Problems}

\author{
    Hossein Nekouyan\thanks{University of Alberta. Email: \texttt{nekouyan@ualberta.ca}}\\
    \and
    Bo Sun\thanks{University of Ottawa. Email:
    \texttt{bo.sun@uottawa.ca}}\\
    \and 
    Raouf Boutaba\thanks{University of Waterloo. Email:
    \texttt{rboutaba@uwaterloo.ca}}\\
    \and
    Xiaoqi Tan\thanks{University of Alberta. 
    Email: \texttt{xiaoqi.tan@ualberta.ca}}
}

%
%
%
\maketitle              
\begin{abstract}
We study two online resource allocation problems with reusability in an adversarial setting, namely \problemkRentalFD and \problemkRentalVD. In both problems, a decision-maker manages $k$ identical reusable units and faces a sequence of rental requests over time. We develop theoretically grounded relax-and-round algorithms with provable competitive ratio guarantees for both settings. For \problemkRentalFD, we present an optimal randomized algorithm that achieves the best possible competitive ratio. The algorithm first computes an optimal fractional allocation using a price-based approach, and then applies a novel lossless online rounding scheme to obtain an integral solution. For \problemkRentalVD, we first establish the impossibility of achieving lossless online rounding. We then introduce a limited-correlation rounding technique that treats each unit independently while introducing controlled dependencies across allocation decisions involving the same unit. Combined with a carefully-crafted price-based method for computing the fractional allocation, this approach yields an order-optimal competitive ratio for the variable-duration setting.

\end{abstract}

\section{Introduction}

This paper investigates adversarial online allocation problems involving \textit{reusable} resources, which we refer to as \textit{$k$-rental} problems. In the $k$-rental setting, a seller (or decision-maker) manages an inventory of $k$ identical units of a resource, say, $ k $ balls. A sequence of requests arrives over time, each seeking to \textit{rent} one ball for a specified duration. Once allocated, a ball becomes unavailable until the rental period ends, after which it returns to the inventory. Each request also specifies a unique valuation for renting a ball. The objective is to maximize the total valuation of all accepted requests.

In the special case without reusability (i.e., the rental durations are infinite), numerous variants of $k$-rental problems have been extensively studied under different arrival models. These include the secretary problem in the random-order model \cite{gardner1970mathematical}, prophet inequalities in stochastic IID or non-IID settings \cite{huang2025prophet,fu2024prophet,Jiashuo2023prophet}, and different forms of adversarial online selection problems, such as $ k $-search and one-way trading \cite{el2001optimal,lorenz2009optimal,jiang2021online,tan2023threshold}, the online knapsack problem~\cite{sun2022online,tan2020mechanism,zhou2008budget}, and their applications in revenue management~\cite{ma2020algorithms,ball2009toward}.

In this paper, we study two adversarial $k$-rental problems, denoted \problemkRentalFD and \problemkRentalVD. In \problemkRentalFD, all requests have a fixed and identical rental duration, but their valuations are uncertain. The \problemkRentalVD setting generalizes this by allowing rental durations to vary across requests, with each request's valuation assumed to be linear in its duration. 

For the \problemkRentalFD problem, \cite{delong2022} considers a more general online bipartite matching model in which each matched offline node becomes available again after a fixed $d$ time units. Their algorithm achieves a competitive ratio of 0.589, though it does not address the optimality of their algorithm or lower bounds of the \problemkRentalFD problem. They also proposed a second algorithm using a correlation-based rounding scheme that achieves a 0.505-competitive ratio which is very close to the 0.5 ratio attained by the greedy algorithm in this setting. However, the primary limitation of their approach lies in the rounding step, which lacks sufficient effectiveness to deliver tight performance guarantees. A more recent work by \cite{ekbatani2025} examines an online job assignment problem that overlaps with both \problemkRentalFD and \problemkRentalVD. They develop \textit{deterministic} algorithms that achieve order-optimal competitive ratios, but only in the \textit{large-inventory} regime where $k \to \infty$. This leaves open the important question of designing \textit{randomized} algorithms with tight guarantees for both \problemkRentalFD and \problemkRentalVD across all ranges of inventory sizes, including small and moderate regimes (i.e., $ k $ is  finite).

\subsection{Our Contributions}

We address the above open question affirmatively by developing randomized algorithms for $k$-rental under general inventory settings. Our algorithms achieve the optimal competitive ratio for \problemkRentalFD and a near-optimal performance guarantee for \problemkRentalVD. Specifically, our contributions are twofold:

\paragraph{\gocr and lossless online rounding schemes.} We introduce an online rounding subroutine termed $\gamma$-Online Correlated $k$-Rental (\gocr) that captures the core challenge of rounding any feasible online fractional allocations in reusable settings with fixed rental durations. The goal is to round each fractional allocation to an integral one, allocating a unit with probability at least equal to its fractional value multiplied by $ \gamma \in [0, 1] $. As a warm-up, we present an independent rounding scheme and show that it is lossless when the inventory is large but fails in small-inventory cases. To address this, we develop a new online rounding algorithm, \locr, which introduces correlations across time steps to achieve lossless rounding. We believe this method has broader applicability and may be of independent interest.

\paragraph{Implications and insights of \gocr for pricing reusable resources.}
 Our online rounding framework \gocr yields both positive and negative results for the \problemkRentalFD and \problemkRentalVD problems, respectively. For \problemkRentalFD, we design an optimal randomized algorithm grounded in the \locr framework. This algorithm first computes an optimal online fractional solution using a price-based approach and then applies the \locr rounding scheme to obtain an integral allocation. In contrast, for \problemkRentalVD with variable rental durations, we establish a fundamental impossibility result: no online rounding scheme can convert an arbitrary fractional solution into an integral one while preserving the same competitive ratio. To address this challenge, we propose a randomized algorithm based on a limited-correlation rounding scheme that treats different units independently while introducing controlled dependencies across allocation decisions involving the same unit. Although the rounding step is inherently lossy, we prove that the resulting algorithm attains an order-optimal, best-known performance guarantee for \problemkRentalVD.

\paragraph{Techniques.}
Our overarching approach follows the \textit{relax-and-round} framework, which first computes a fractional allocation for a relaxed version of the problem using a price-based method, and then rounds this solution to obtain an integral allocation. The second step, i.e., the rounding procedure, is more nuanced and central to our contribution. From a technical standpoint, to derive a pricing function that produces the optimal fractional solution for \problemkRentalFD, we adopt an online primal-dual approach to guide the design of an appropriate pricing function. For \problemkRentalVD, we employ the LP-free certificate framework developed in~\cite{goyal2020} to construct the pricing function. In this approach, designing an $\alpha$-competitive algorithm reduces to finding a feasible solution to a system of \textit{delayed differential inequalities}, which the pricing function of \algorithmVD must satisfy in order to achieve the desired competitiveness. This system, parameterized by the achievable competitive ratio, characterizes the design of a pricing function that attains the smallest possible competitive ratio.

\subsection{Related Work}

Online resource allocation—the process of assigning limited resources to sequentially arriving requests to maximize social welfare or profit—is a central topic in computer science and operations research.   
In the following we briefly mention important papers from the literature related to this work. 

\paragraph{Online scheduling.}
The adversarial online interval scheduling problem, introduced by \cite{Lipton1994OnlineIS}, considers the task of scheduling a sequence of intervals (jobs) on a single server as they arrive in order of their start times. Assuming that the minimum and maximum job durations are unknown, their randomized algorithm based on a converging sequence of coin flip probabilities achieves an $O\bigl((\ln \Delta)^{1+\varepsilon}\bigr)$ competitive ratio for any $\varepsilon > 0$, where $\Delta$ is the ratio of the longest to the shortest job duration. Motivated by this work, several subsequent studies (e.g., \cite{Faigle96,Gupta2016,Goyal2020OIS}) have explored various extensions of the problem. For instance, \cite{Goyal2020OIS} examined a variant in which the minimum and maximum job durations are known, referred to as the \emph{online reservation problem}, and proposed a static pricing algorithm with a performance guarantee of $3 \cdot \bigl(1 + \ln(\Delta)\bigr)$. However, the static threshold approach may be ill-suited for online allocation involving reusable resources, where the algorithm must manage more complex dynamics and schedule jobs to maximize total resource utilization. In Section~\ref{sec:krental:variable}, we investigate whether a dynamic pricing scheme, one that continuously updates resource prices in response to changing market conditions, can achieve stronger performance guarantees. Additionally, we develop an online rounding scheme that, when integrated with the dynamic pricing mechanism, effectively converts fractional solutions into integral ones while maintaining robust performance guarantees.

\paragraph{Online matching with reusable resources.}
Several studies, including \cite{delong2022,sun2022online,Garofalakis2002,huo2023online,ekbatani2025}, have investigated the online matching problem and its variants in settings with resource reusability. In particular, as previously discussed, \cite{delong2022} examined an online bipartite matching problem in which resources are reusable. This setting extends the classical online matching model introduced by \cite{karp90} by allowing each matched offline node to become available again after $d$ time units. In Section \ref{sec:krental-fixed}, we focus on a special case of this model involving a single resource type and valuation uncertainty. We explore whether improved competitive ratios and stronger correlation schemes can be achieved compared to their results. \cite{huo2023online} studied adversarial reusability in the context of online assortment planning, where the \problemkRentalVD problem can be reduced to their model. Their work also centers on deterministic algorithm design. Within the \problemkRentalVD framework, their algorithm achieves a competitive ratio of $4\,\ln\bigl(\tfrac{d_{\max}}{d_{\min}}\bigr)$ as the inventory size tends to infinity, where $d_{\max}$ and $d_{\min}$ denote the maximum and minimum rental durations, respectively. More recently, \cite{ekbatani2025} considered a generalized version of the \problemkRentalVD problem within the online matching framework. Their model incorporates uncertainty in each request’s per-unit-time valuation over the requested interval and across different item types. Given the close connection between their formulation and ours, we provide a detailed comparison in Appendix~\ref{apx:comparison}, deferred due to space constraints.

\paragraph{Online rounding.}
Recent work in computer science and operations research has demonstrated the effectiveness of online rounding frameworks~\cite{fahrbach2022edge,ma2024randomizedroundingapproachesonline}, which typically adopt a \textit{relax-and-round} paradigm. Notable progress has been made in designing innovative online rounding schemes in recent years. For example, \cite{fahrbach2022edge} introduced a subroutine called Online Correlated Selection, which imposes negative correlation across selected pairs to achieve improved competitive ratios over classical greedy algorithms, which are limited to a $\frac{1}{2}$-competitive ratio. Building on this idea, \cite{delong2022} developed the Online Correlated Rental method for settings with reusable resources and fixed rental durations. Similar correlation-based techniques have also been adapted to the Adwords problem by \cite{huang2024adwords}. In the context of stochastic optimization, several studies (e.g., \cite{Feldman2021,Alaei2014}) have proposed related rounding techniques, such as online contention resolution schemes, demonstrating their effectiveness in applications including prophet inequalities and online stochastic matching.

\section{A Lossless Rounding Scheme for Online Correlated $ k $-Rental}
\label{sec:losslessCorrelation}

We first introduce Online Correlated $ k $-Rental (\ocr), which is an online rounding subroutine for $ k $-rental problems. The  \ocr subroutine captures the key challenge of rounding a fractional solution in $ k $-rental problems where resources are reusable, and 

serves as a building block for algorithm design in Section~\ref{sec:krental-fixed}.

\subsection{\gocr: Definitions, Objectives, and Challenges}
\label{sec:gocr}

\begin{definition}[\gocr]
Consider a set of $k$ identical balls, each uniquely labeled from the set $\{1, 2, \ldots, k\}$. Each ball can be rented to a player for a fixed duration of $d$ time units, after which it becomes available for reuse. A sequence of $N$ players arrives one by one, where each player $n$ is characterized by a tuple $(\hat{x}_n, a_n)$. Here, $\hat{x}_n \in [0,1]$ denotes the target probability with which the procedure should assign a ball to player $n$, and $a_n$ is the arrival time of player $n$. For a fixed $\gamma \in [0,1]$, a \gocr\ is an online rounding scheme that guarantees renting a ball to each player $n \in [N]$ with probability at least $\gamma \hat{x}_n$, for all input instances $\{(\hat{x}_n, a_n)\}_{n \in [N]}$ satisfying the following condition:
\begin{align}\label{eq:ocr_conditions}
\hat x_n \leq \min \left\{1,\, k - \sum\nolimits_{j \in [n-1]} \hat x_{j} \cdot \mathbb{I}_{\{a_{j} + d > a_n\}} \right\}, \qquad\quad \forall n \in [N].
\end{align}
\end{definition}

As a rounding scheme, 
\gocr is a randomized online algorithm that makes irrevocable decisions upon the arrival of each player, either assigning a ball to the player or rejecting it.
\gocr focuses on the inputs that satisfy the regularity conditions in Eq.~\eqref{eq:ocr_conditions}. In particular, the term $\sum_{j \in [n-1]} \hat x_{j} \cdot \mathbb{I}_{\{a_{j} + d > a_n\}}$ on the right-hand side quantifies the cumulative targeted probabilities of players who arrived prior to player~$n$ and whose rental intervals (each of length~$d$) overlap with that of player~$n$. This sum accounts for the expected occupancy of the balls at the time of player~$n$’s arrival. Furthermore, the condition in Eq.~\eqref{eq:ocr_conditions} also guarantees that $\hat x_n \leq 1$ for all $n \in [N]$, thereby ensuring that the assignment probability for any individual player does not exceed the availability of a single unit. Violation of this constraint would imply that the total targeted assignment at the arrival of player~$n$ exceeds the inventory limit of $k$ units.

The goal is to design a \gocr that maximizes $\gamma$. A \locr is referred to as a \textit{lossless} rounding scheme, as it guarantees renting a ball to each player $n \in [N]$ with probability at least $\hat{x}_n$. Designing such a scheme, however, is non-trivial. As a warm-up, we first present a simple independent rounding algorithm that incurs rounding losses. This scheme highlights the limitations of uncorrelated rounding and underscores the need to correlate the allocation decision for each newly arriving player with the decisions made for previously served players who still hold rented balls. 

\subsection{Warm Up: A Simple Independent Rounding Scheme}

\begin{algorithm}[H]
\caption{\gocr using Independent Rounding}
\label{alg:independent:rounding}
\textbf{Input}: Number of balls $k$, rental duration $d$, parameter $f_k \in [0,1]$.\\
\textbf{Output}: Assignment $z_n \in \{0, 1\}$ for each player $n \in [N]$.

\For{each player $n$ with the tuple $(a_n, \hat x_n)$}{
    Compute the number of balls currently rented at time $a_n$: $y_n = \sum_{j=[n-1]} z_j \cdot \mathbb{I}_{\{a_j + d > a_n\}}$.

    Sample a random seed {$S_n \sim U(0,1)$.} 

    \If{$S_n \leq f_k \cdot \hat x_n$ \textbf{and} $y_n < k$}{
        $z_n \gets 1$; \Comment{Allocate the $(y_n+1)$-th available ball to player $n$}
    }
    \Else{
        $z_n \gets 0$. \Comment{Reject request $n$}
    }
}
\end{algorithm}

Consider the rounding scheme in Algorithm~\ref{alg:independent:rounding}, which is motivated by rounding schemes in the approximation algorithms literature.

It samples a random seed $S_n$ independently at the arrival of each player $n$ from the uniform distribution $U(0, 1)$. A ball is then allocated to player~$n$ if $S_n \leq f_k \cdot \hat x_n$ and at least one ball is available. The parameter $f_k \in [0,1]$ serves as a downscaling factor that depends on the total number of balls, $k$, and targets to reduce the allocation probability for each player. This downshift increases the likelihood that a ball will be available for future arrivals, allowing more players to be served.

\begin{proposition}
\label{lem:ir}
Given $f_k\in [0,1]$, Algorithm~\ref{alg:independent:rounding} is $\gamma_k$-\ocr, where $\gamma_k = f_k \cdot \left(1 - \exp\left( - \frac{(k - f_k \cdot k)^2}{f_k \cdot k + k } \right) \right).$

\end{proposition}

\begin{figure}
    \centering
    \includegraphics[width=0.5\linewidth]{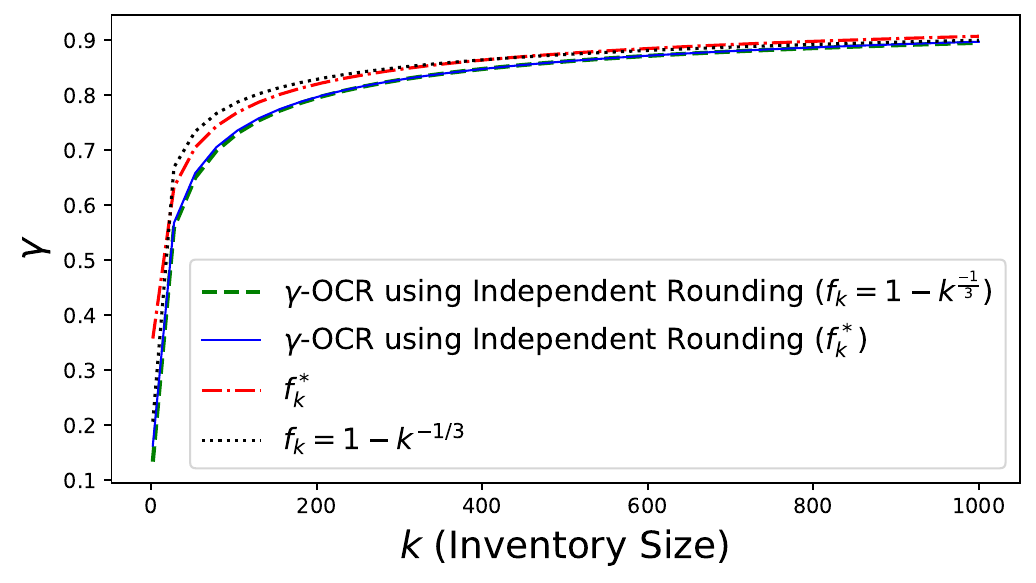}
    \caption{Comparison of the approximate parameter $f_k = 1 - \frac{1}{k^{1/3}}$ versus the optimal parameter $f^*_k$, and the resulting $\gamma$ values for Algorithm~\ref{alg:independent:rounding} using $ f_k$ and $f_k^*$. }
    \label{fig:independet}
\end{figure}

The proof of Proposition \ref{lem:ir} is based on the Chernoff bound inequality and is provided in Appendix \ref{apx:proof:indp:rounding}. This result also informs the design of the downscaling parameter $f_k$ as a function of $k$ to maximize $\gamma_k$. The optimal value $f_k^*$ cannot be obtained in closed form, since it involves solving a transcendental equation of the form $A(x) = \ln(1 - D(x))$. However, it can be efficiently computed using numerical methods such as binary search, and is illustrated in Figure~\ref{fig:independet}. Based on the numerical results, we also observe that the optimal parameter $f_k^*$ closely follows a simple expression $f_k = 1 - \frac{1}{k^{1/3}}$. 

Figure~\ref{fig:independet} shows that the performances of Algorithm~\ref{alg:independent:rounding} with $f_k = 1 - \frac{1}{k^{1/3}}$ and $f_k^*$ are also close. 
In addition, as seen in the figure, for large values of $k$ the resulting $\gamma$ approaches one; however, for small inventory sizes (small values of $k$), the algorithm is not an effective rounding scheme. This motivates our design for a lossless rounding scheme.

\subsection{A \locr Through Dependent Rounding}

In this subsection, we propose a \locr by correlating the allocation decisions in Algorithm \ref{alg:locr}.

\paragraph{Key idea of creating negative correlation.}  
At a high level, the key idea of Algorithm \ref{alg:locr} is to correlate each current allocation decision with past decisions involving players whose allocated balls are expected to return sooner. Specifically, prior to the arrival of any player, the algorithm samples a single random seed $r \sim U(0,1)$, which remains fixed throughout the entire execution and serves as the only source of randomness. This shared seed induces correlation across allocation decisions for different players. For each player $n$, the algorithm uses this seed to determine whether to allocate a ball. It maintains two pointers: $m_n \in \{1, 2, \dots, k \}$ and $p_n \in [0,1]$. The first pointer, $m_n$, indicates the ball under consideration for allocation to player $n$; this ball may be unavailable if it is currently assigned to a previous player. The second pointer, $p_n$, specifies a subinterval of $[0,1]$ used to guide the allocation decision. For each player $n$, the algorithm proceeds based on the relationship between $p_n$ and $\hat{x}_n$, as described in the following two cases.
\begin{itemize}
    \item \textbf{Case 1:} If $p_n + \hat x_n < 1$, then the algorithm allocates ball $m_n$ to player $n$ if $r \in [p_n,\,p_n+\hat x_n)$ and the ball is available in the system. Otherwise, the player is rejected and no ball is allocated.

    \item \textbf{Case 2:} If $p_n + \hat x_n \geq 1$, then the algorithm allocates ball $m_n$ if $r \in [p_n,\,1]$ (and the ball is available in the system), or allocates ball $m_n+1$ if $r \in [0,\,p_n+\hat x_n-1)$. Since $\hat x_n \leq 1$, these intervals are non-overlapping, ensuring that no more than one ball is allocated.
\end{itemize}

\begin{algorithm}[tbh]
\caption{$1$-\ocr using Dependent Rounding}
\label{alg:locr}
\textbf{Input}: Number of balls $k$, rental duration $d$.

\textbf{Output}: Assignment $z_n \in \{0, 1\}$ for each player $n \in [N]$.

\textbf{Initialize:} Set $m_1 = 1$, $p_1 = 0$.
sample a random seed $r \sim U(0,1)$.

\For{each request $n$ with the tuple $(a_n, \hat x_n)$}{

\If{$ \sum_{j \in [n]} \hat x_{j} \cdot \mathbb{I}_{\{a_{j} + d > a_n\}} > k $} {

$z_n \leftarrow 0$. \Comment{Reject player $n$}

}

    \If{$p_n + \hat x_n < 1$}{
        \If{$r \in [p_n, p_n+\hat x_n)$}  { \label{eq:kRental:OCKR:if1}

        $z_n \leftarrow 1$. \Comment{Assign ball $m_n$ to player $n$}
            
        }
        \Else{
        
        $z_n \leftarrow 0$. \Comment{Reject player $n$}

        }
         
        Update $p_{n+1} = p_n + \hat x_n$ and $m_{n+1} = m_n$.
    }
    \Else{
        \If{$r \in [p_n, 1]$} { \label{eq:kRental:OCKR:if2}

        $z_n \leftarrow 1$. \Comment{Allocate ball $m_n$ to player $n$}
            
        }
        \ElseIf{$r \in [0, \hat x_n + p_n - 1)$}{ \label{eq:kRental:OCKR:if3}

        $z_n \leftarrow 1$. \Comment{Allocate ball $m_n+1$ to player $n$}
            
        }    
    \Else{
    $z_n \leftarrow 0$. \Comment{Reject player $n$}
    
    }
    Update $p_{n+1} = \hat x_n + p_n - 1$ and $m_{n+1} = m_n + 1$ (if $m_{n+1} > k$, then set $m_{n+1} = 1$).
    }
}
\end{algorithm}

\paragraph{Key invariant of Algorithm \ref{alg:locr}.}

It is worth noting that Algorithm \ref{alg:locr} does not verify the availability of balls $m_n$ and $m_n+1$ before assigning them to players. This omission is justified by some key invariants that we establish in Proposition \ref{lemma:kRental:VD:loca:availability} below concerning the availability of these balls.

\begin{proposition}[Invariants of Algorithm~\ref{alg:locr}]
\label{lemma:kRental:VD:loca:availability}
Upon the arrival of the $n$-th player, the following holds:
\begin{itemize}
    \item If $p_n + \hat x_n < 1$ and the random seed $r \in [p_n,\, p_n+\hat x_n)$, then the ball $m_n$ is available.
    \item If $p_n + \hat x_n \ge 1$ and the random seed $r \in [p_n,\, 1]$, then the ball $m_n$ is available.
    \item If $p_n + \hat x_n \ge 1$ and the random seed $r \in [0,\, p_n+\hat x_n-1)$, then the ball $m_n +1$ is available.

\end{itemize}
\end{proposition}
The proof of the above proposition is provided in Appendix \ref{apx:proof:lemma:kRental:VD:loca:availability}. The pointer mechanism described therein plays a critical role in synchronizing the scheduling of different balls for consecutive players. It also maintains a record of subintervals within $[0,1]$ such that, if the random seed $r$ falls within these intervals, either ball $m_n$ or $m_{n+1}$ remains available. As explained in the context of Algorithm~\ref{alg:locr}, the total length of the subinterval used for decision-making for each player $n$ is precisely $\hat x_n$. Since $r$ is sampled from the uniform distribution $U(0,1)$ and a ball is assigned when $r$ lies within the corresponding subinterval, each player receives a ball with the target probability $\hat x_n$.

Leveraging the key invariants established in Proposition~\ref{lemma:kRental:VD:loca:availability}, we can conclude that Algorithm~\ref{alg:locr} is lossless.

\begin{theorem}
\label{thm:correctness-OCS-algorithm}
Algorithm \ref{alg:locr} is a \locr, namely, it is a lossless rounding scheme.

\end{theorem}

Due to the space limit, we defer the full proof of the above theorem to Appendix~\ref{apx:thm:correctness-OCS-algorithm}, and provide a detailed example to explain how Algorithm~\ref{alg:locr} works in Appendix~\ref{apx:example:locr}.

\section{\problemkRentalFD: $k$-Rental with Fixed Rental Durations}
\label{sec:krental-fixed}

In this section, we define the  \problemkRentalFD problem and present an optimal randomized algorithm based on the rounding scheme, \locr, developed in the previous section.

\subsection{Problem Formulation and Assumptions}

Let us formally define the online $k$-rental problem (\problemkRentalFD) as follows. A decision maker has $k$ units of an item and allocates them to online arriving requests. 
Each request $n \in [N]$ asks to rent one item for $d$ time units. 
The $n$-th request arrives at time $a_n$, and has a valuation $v_n$ for renting one item. In particular, if an item is allocated to request $n$, the item is rented starting from $a_n$, and becomes available again at time $a_n + d$. 

Let $x_n \in \{0,1\}$ denote the decision to accept or reject the $n$-th request. The objective of the problem is to maximize the total valuation of requests that receive an item allocation, i.e., $\sum_{n \in [N]} x_n \cdot v_n$, while respecting the constraints on available items.

Let $I = \{v_n, a_n\}_{n \in [N]}$ denote an instance of \problemkRentalFD. The maximum valuation from the optimal clairvoyant algorithm, $\opt(I)$, can be determined by solving the following optimization problem:
\begin{subequations}
\label{lp:kRentalFD:primal}
\begin{align}
 \max_{x_n} \quad & \sum\nolimits_{n \in [N]} v_n \cdot x_n, \\
 \label{eq:capacity}
    \text{s.t.} \quad & \sum\nolimits_{j\in [n]} x_{j} \cdot \indicator{d + a_{j} > a_n} \leq k,\quad \forall\, n \in [N], \\
    & x_n \in \{0,1\},\quad \forall\, n \in [N].
\end{align}    
\end{subequations}

The constraint~\eqref{eq:capacity} ensures that at any point in time throughout the horizon, no more than $k$ items are allocated. 
The left-hand-side of this constraint counts the number of units of the item allocated at that moment, including the decision made for the arriving request~$n$.
Thus, it is sufficient to enforce this constraint only at the arrival of each request~$n$, when the available inventory may decrease.

In the online setting, the decision maker must make an irrevocable decision to accept or reject each request upon its arrival. The uncertainty regarding future requests' valuations and the overall demand for the items makes this decision challenging. 
To achieve a bounded performance, we follow the literature and assume the valuations of requests are bounded.
\begin{assumption}
\label{ass:bounded-value}
All valuations are within the range $[v_{\min},v_{\max}]$,
    i.e., $ v_n \in [v_{\min},v_{\max}], \forall n\in[N]$.
\end{assumption}

Let $\mathcal{I}$ denote the set of all instances of the \problemkRentalFD that satisfy Assumption~\ref{ass:bounded-value}. Our goal is to design online algorithms whose objective is competitive with that of $\opt(I)$ for every instance $I \in \mathcal{I}$. 
Specifically, an online algorithm $\alg$ is said to be $\alpha$-competitive if, for any input instance $ I $, the following inequality holds $\alpha \ge \frac{\opt(I)}{\mathbb{E}[\alg(I)]}$, where the expectation $\mathbb{E}[\alg(I)]$ is taken over the randomness of the online algorithm.

In the following section, we present an algorithm for the \problemkRentalFD problem that achieves the minimum possible competitive ratio among all online algorithms. This result strictly improves upon the deterministic algorithm proposed by \cite{ekbatani2025}, which attains only an order-optimal competitive ratio in the asymptotic regime where the inventory size tends to infinity.

\subsection{\algkRentalFD: A Relax-and-Round Algorithm for \problemkRentalFD}
\label{sec:kRental:fixed:algorithm:explanation}

We introduce a randomized online algorithm, \algkRentalFD, presented in Algorithm~\ref{alg:kRentalFD}, which is based on a general \textit{relax-and-round} framework. In the \textit{relaxation step}, we design an online algorithm for a continuous version of problem~\eqref{lp:kRentalFD:primal} by relaxing the integrality constraint to $x_n \in [0,1]$ for all $n \in [N]$. This algorithm produces an online fractional solution $\{\hat{x}_n\}_{n \in [N]}$. In the \textit{rounding step}, the algorithm converts the fractional solution $\hat{x}_n$ into an integral decision $x_n \in \{0,1\}$.

\paragraph{Core idea for obtaining fractional allocation: Pseudo-utility maximization in Eq. \eqref{eq:kRental:FD:fraction:allocation}.} 
Upon the arrival of request $n$, the algorithm first computes an expected utilization level $y_n = \sum_{j=1}^{n-1} \hx_{j} \cdot \indicator{a_{j} + d > a_n}$, which is the expected number of resource units that are currently rented by the previous requests whose rental durations overlap with that of request $n$. Using this expected utilization level, the algorithm determines the fractional allocation by solving a \textit{pseudo-utility maximization} problem as described in Eq.~\eqref{eq:kRental:FD:fraction:allocation}. The first term $x \cdot v_n$ is the valuation from request $n$, and the second term is a \textit{pseudo-cost} of renting $x$ unit of an item when the current utilization level is $y_n$. Specifically, the pseudo-cost, $k\, \int_{y_n/k}^{(y_n+x)/k} \phi(\eta) \, d\eta$, is estimated using a carefully-designed, normalized \textit{marginal pricing} function $\phi: [0,1] \to [v_{\min}, v_{\max}]$. Thus, the integration of $ \phi $ over the resource utilization interval $ [\frac{y_n}{k}, \frac{y_n+x}{k}] $ gives the {pseudo-cost}. 

In this context, the fractional allocation $\hat x_n$ represents the optimal fraction of a resource unit to allocate to request $n$, given their valuation $v_n$ and the pricing rule defined by $\phi$.

\paragraph{Rounding subroutine: \locr.} In this step, the fractional solution $\hat x_n$ is passed to the lossless online rounding procedure \locr. This procedure generates an integral decision $x_n$ on whether to accept the request. This procedure operates as an online algorithm, with an instance of this algorithm initiated at the start of Algorithm~\ref{alg:kRentalFD}. As each request $n$ arrives at time $a_n$, the \locr procedure receives a probability value $\hat x_n \in [0, 1]$ (generated by the relaxation step) as input at time $a_n$. Then it makes an integral decision $x_n$ on whether to accept or reject the request. All rounding decisions are based on one random seed, which is fixed when Algorithm~\ref{alg:kRentalFD} initiates the \locr instance and remains the same for all requests.

\begin{algorithm}[htb]
\caption{Duration-Oblivious Price-based Algorithm with Lossless Rounding for \problemkRentalFD (\algkRentalFD); 
}
\label{alg:kRentalFD}
\textbf{Input:} Pricing function $\phi: [0,1] \rightarrow [v_{\min}, v_{\max}]$.

\textbf{Initiate} an instance of \locr procedure.

\For{each arriving request $n$ with $(v_n, a_n)$}{
    Compute the expected utilization level at time $ a_n $: $y_n = \sum_{j=1}^{n-1} \hx_{j} \cdot \indicator{a_{j} + d > a_n}$.

    \If{ $y_n < k$} 
    {
    Compute the fractional allocation: \Comment{Relax step}
    \begin{align}
    \label{eq:kRental:FD:fraction:allocation}
    \hat x_n = \argmax_{x \in [0,\,\min\{1,\,k - y_n\}]}  x \cdot v_n - k\, \int_{\eta=y_n/k}^{(y_n+x)/k}  \phi(\eta)\,d\eta.
    \end{align}
    
    Decide the integral allocation by Algorithm~\ref{alg:locr}: $x_n =$ \locr$(\hat x_n, a_n)$. \label{line:rounding} \Comment{Round step}
    } 
    \Else{$x_n = 0$. \Comment{Reject request $n$}}
}
\end{algorithm}
I 

Before leaving this subsection, we remark that the pseudo-utility maximization that yields the fractional allocation in Eq. \eqref{eq:kRental:FD:fraction:allocation} differs fundamentally from the pseudo-cost frameworks of \cite{ekbatani2025,sun2022online}. Those works adopt a \textit{forward-looking} approach that estimates the cost of accepting a request using the \emph{entire} utilization trajectory over the request’s rental interval. Thus, the pseudo-cost is a functional of the utilization curve throughout the interval, not merely the level observed at arrival. In contrast, our algorithm adopts a much simpler \textit{myopic} approach that computes its pseudo-cost just based on the \emph{current} utilization at the moment the request arrives (hence the term ``\textit{duration-oblivious}"). We next show that it is a simpler yet still competitive pricing rule.

\subsection{Theoretical Guarantee of \algkRentalFD}

As explained in Section~\ref{sec:kRental:fixed:algorithm:explanation}, \algkRentalFD employs the \locr procedure to round the fractional allocation $\hat x_n$ at each time step and to make an integral decision. For this procedure to function correctly, the fractional solution generated by the relax step must satisfy the feasibility conditions in Eq.~\eqref{eq:ocr_conditions} for \locr.

\begin{lemma}[Feasibility of Fractional Solutions]
\label{lemma:feasibility-oca}

In \algkRentalFD, the input of the round step (line \ref{line:rounding} of Algorithm \ref{alg:kRentalFD}),  $ \{(\hat x_n, a_n)\}_{n \in [N]}$, satisfies the feasibility condition in Eq. \eqref{eq:ocr_conditions}.
\end{lemma}

The above lemma naturally follows from the design of \algkRentalFD. Based on this, we prove the performance guarantee for \algkRentalFD.

\begin{theorem}
\label{them:kRentalFD}
 \algkRentalFD is $ \bigl(1+\ln \bigl(\tfrac{v_{\max}}{v_{\min}}\bigr) \bigr)$-competitive for \problemkRentalFD when  $\phi$ is given by 
\begin{align}
\label{eq:pricing-FD}
\phi(y) = v_{\min} \cdot \exp\Bigl(\bigl(1+\ln(\tfrac{v_{\max}}{v_{\min}})\bigr)\cdot y - 1\Bigr),
\quad 
y \in [0,1].
\end{align}
\end{theorem}

The complete proof of Theorem \ref{them:kRentalFD} is provided in Appendix \ref{apx:proof:them:kRentalFD}. Our  proof follows the well-established online primal-dual (OPD) framework, with a notable deviation from the conventional approach: the updates to the dual variables are deferred until after the arrival of the final request similar to the work done in \cite{huo2023online}. Specifically, the dual linear program involves variables $\{u_n, \lambda_n\}_{n \in [N]}$. In line with the interpretation in \cite{eden2021economics}, $\lambda_n$ can be viewed as the per-unit price of the item at the arrival of request~$n$, while $u_n$ represents the utility that request~$n$ derives under the pricing scheme that determines the fractional allocation. A key step in the OPD method is the update of the variable $\lambda_n$. For each request~$n$, it suffices to increase the price only at the arrival time of request $\nu^*_n$, defined as the \textit{latest-arriving request whose interval overlaps with that of request $n$}. This suffices because any future request whose allocation might be influenced by that of request~$n$ must also have an interval overlapping with $\nu^*_n$. Therefore, updating the item price at the arrival of $\nu^*_n$ ensures the the dual constraint getting satisfied.

Finally, we can show that \algkRentalFD achieves the optimal competitive ratio by showing a matching lower bound as follows.

\begin{proposition}[Lower Bound]
\label{lemma:kRental:FD:lb}
Under Assumption \ref{ass:bounded-value}, no online algorithm, deterministic or randomized, can obtain a competitive ratio better than $ 1+\ln \bigl(\tfrac{v_{\max}}{v_{\min}}\bigr) $ for the \problemkRentalFD problem.
\end{proposition}

The proof of the above proposition is presented in Appendix~\ref{apx:lemma:kRental:FD:lb}. Following the approach of \cite{sun2024static,jazi2025posted}, we construct a \emph{family of hard instances} to derive the desired lower bound. The instance begins with $k$ requests, each with valuation $v_{\min}$. These are followed by successive batches of requests with monotonically increasing valuations, eventually reaching $v_{\max}$. All requests arrive within an arbitrarily small time window that is negligible relative to the deadline $d$. By applying the representative function-based approach \cite{sun2024static,jazi2025posted} to this family of instances, we derive the optimal online algorithm and establish a matching lower bound applicable to all online algorithms.

\section{\problemkRentalVD: $k$-Rental with Variable Rental Durations}

\label{sec:krental:variable}

This section focuses on the \problemkRentalVD\ problem, which generalizes the \problemkRentalFD\ problem studied in Section~\ref{sec:krental-fixed} by allowing variable rental durations.

\subsection{Problem Statement and Assumptions}
We formally define the $k$-rental problem with variable rental durations (\problemkRentalVD) as follows.
A decision maker has $k$ identical units of a resource to allocate to $N$ requests arriving online.
Each request $n \in [N]$ arrives at time $a_n$ and asks to rent one unit of the resource for $d_n$ time units, where the rental duration $d_n$ may vary across requests.
If a unit is allocated to request $n$, it remains occupied until time $a_n + d_n$, after which it becomes available again for future requests.
Each request has a valuation equal to its rental duration $d_n$. Let $x_n \in \{0,1\}$ indicate whether request $n$ is accepted ($x_n = 1$) or rejected ($x_n = 0$). The objective is to maximize the total valuation of accepted requests, $
\sum_{n\in[N]} x_n \, d_n,$
subject to the resource constraint corresponding to the finite inventory capacity.

Let $I = \{a_n, d_n\}_{n \in [N]}$ denote an instance of the problem. The performance of the optimal clairvoyant algorithm on instance $I$, $\opt(I)$, can be computed based on following integer linear program:
\begin{subequations}
\label{lp:kRentalVD:primal}
\begin{align} 
 \max_{x_n} \quad & \sum\nolimits_{n \in [N]} x_n \cdot d_n, \\
    s.t. \quad & \sum\nolimits_{j\in [n]} x_{j} \cdot \mathbb{I}_{\{d_j + a_j > a_n\}} \leq k,\quad \forall n \in [N], \\
    & x_n \in \{0,1\},\quad \forall\, n \in [N].
\end{align}    
\end{subequations}

To achieve a bounded competitive ratio, we still impose constraints on the adversary, ensuring that the requested rental durations of requests
are bounded within a finite support.

\begin{assumption}\label{ass:bounded-duration-continuous}
All rental durations are within the range $[d_{\min},d_{\max}]$, i.e., $ d_n \in [d_{\min},d_{\max}], \forall n\in[N]$.
\end{assumption}

In the following, we first prove that obtaining a lossless rounding scheme for a variant of the \gocr problem where players have variable rental durations is impossible. Given this, we then develop a new randomized algorithm that uses a new limited-correlation scheme to round the fractional solution, achieving a tight competitive ratio for \problemkRentalVD.

\subsection{Impossibility of \locr under Variable Durations}
We show that the correlation scheme of the \locr procedure in Algorithm~\ref{alg:locr} cannot be extended to settings with variable, request-dependent rental durations; in fact, when the duration of the requested rental intervals varies, it can be proven that no lossless rounding scheme exists. 

\begin{theorem}
\label{prop:rounding:impossiblity}
There is no rounding scheme that can attain a lossless rounding, i.e.\ $\gamma = 1$, for the \gocr problem with variable rental durations.
\end{theorem}

We prove the theorem by constructing a family of hard instances on which every online algorithm necessarily incurs a non-zero loss. Fundamentally, a lossless algorithm would need to randomize its decisions such that, across all sample paths, the number of allocated items remains balanced. Specifically, upon the arrival of request~$n$, the algorithm must accept it on precisely those sample paths where the number of allocated items is lower than in any other path, or on paths where the number of allocated items is about to decrease due to the imminent expiration of rental durations for previously accepted requests. However, in the variable-duration setting, rental lengths fluctuate unpredictably with the arrival of new requests. A request allocated a unit later than others may return it earlier, introducing significant variability. This makes it infeasible to maintain the necessary correlation across sample paths. As a result, any online algorithm must incur a non-trivial competitive loss. In contrast, when rental durations are fixed, a request that receives a unit later will always return it later. This temporal consistency enables effective correlation of decisions and allows the design of a lossless rounding scheme. In Appendix~\ref{apx:proof:impos:rounding}, we formally define the \gocr problem under variable rental durations and establish the impossibility of a lossless correlation scheme in this setting.

\subsection{\algorithmVD : A Randomized Algorithm with Limited Correlation}

We present \algorithmVD in Algorithm~\ref{alg:kRentalVD-dynamic}, which also follows a relax-and-round paradigm. Upon the arrival of request $n$, the algorithm computes the \textit{probabilistic utilization level} $y_n^{(i)}$ for each unit $i \in [k]$. It then selects the unit $i^*_n$ with the lowest utilization level (breaking ties arbitrarily). A fractional allocation $\hat{x}_n \in [0, 1]$ is determined by solving the pseudo-utility maximization problem in Eq.~\eqref{eq:kRental:VD:fractional}, using a pricing function $\phi$. Note that here $\phi$ does not refer to the function in Eq.~\eqref{eq:pricing-FD}, but instead denotes a new function specifically designed for \algorithmVD. Once $\hat{x}_n$ is computed, the algorithm proceeds to the rounding phase by correlating the current allocation decision for unit $i^*_n$ with its prior allocation history. Specifically, if unit $i^*_n$ is available, a new random seed is drawn, and the unit is allocated to request $n$ with probability ${\hat{x}_n}/{(1 - y_n^{(i^*_n)})}$. This ensures that the overall allocation probability for unit $i^*_n$ is exactly $\hat{x}_n$. \textit{This correlation process is carried out independently for each unit, introducing dependency only among allocation decisions involving the same unit}.

\begin{algorithm}
    \caption{Duration-Oblivious Price-based Algorithm with Limited Correlation for \problemkRentalVD (\algorithmVD) 
    }
    \label{alg:kRentalVD-dynamic}
    
    \textbf{Input:} Pricing function $\phi: [0,1] \rightarrow [d_{\min}, d_{\max}]$.

\For{each arriving request $n$ with $(d_n, a_n)$}{
            
      Compute the probabilistic utilization level at $ a_n $: $ y^{(i)}_{n} = \sum_{j = 1}^{n-1} \hat x_{j} 
        \cdot \indicator{ a_{j} + d_{j} > a_{n} } 
        \cdot \indicator{i^*_{j} = i}, \forall i\in [k] $.

      Let $ i^*_n = \argmin_{i \in [k]}\,\bigl\{y_n^{(i)}\bigr\} $ and compute
      \begin{align}
        \label{eq:kRental:VD:fractional}
       \hat x_n = \argmax_{x\in[0,1]}\  d_n \cdot x - \int_{y_n^{(i^*_n)}}^{y_n^{(i^*_n)}+x} \phi(\eta) d\eta.
      \end{align}

      Sample a random seed $S_n \sim U(0,1)$.
      
      \If{$S_n \leq \frac{\hat x_n}{1 - y_n^{(i^*_n)}}$ and  unit $i^*_n$ available in system}{
      $x_n = 1$ \Comment{Allocate unit $i^*_n$ to request $n$ }
        }
        \Else{
        $x_n = 0$. \Comment{Reject request $n$}
        }

    }
\end{algorithm}

\paragraph{Overview of challenges and design principles of $ \phi $.} Designing a competitive pricing function $\phi$ for \algorithmVD involves two primary challenges. First, unlike the relax-and-round approach adopted for \problemkRentalFD, Algorithm~\ref{alg:kRentalVD-dynamic} integrates the relaxation and rounding steps into a unified procedure, wherein the rounding phase directly influences the structure of the pricing function $\phi$. Similar design principles have been employed in related works, such as \cite{gaoOCS,fahrbach2022edge}. Second, as with Algorithm~\ref{alg:kRentalFD} and in contrast to the approaches in \cite{ekbatani2025,sun2022online}, Algorithm~\ref{alg:kRentalVD-dynamic} adopts a myopic, duration-oblivious decision-making strategy. It bases allocation decisions solely on the current utilization levels $\{y^{(i)}_{n}\}_{i \in [k]}$ at the arrival time $a_n$ of request $n$, rather than considering the entire rental duration. This duration-oblivious pricing scheme may initially appear counterintuitive, as the pseudo-cost term in Eq.~\eqref{eq:kRental:VD:fractional} does not depend on the rental duration $d_n$. Nevertheless, as elaborated in the next section, the pricing function $\phi$ is carefully designed to implicitly capture potential fluctuations in utilization throughout the entire duration of request~$n$. In essence, our goal is to \textit{maintain a simple and cognitively lightweight pricing scheme, while embedding the complexity into the internal structure of the pricing function $\phi$ itself}.

\subsection{Theoretical Guarantees of \algorithmVD}

We employ the LP-free certificate method developed by~\cite{goyal2020} to establish a performance guarantee for \algorithmVD. This involves constructing a system of linear constraints parameterized by $\alpha$, such that the existence of a feasible solution certifies $\alpha$-competitiveness. Satisfying this system effectively imposes a distinct constraint for each request~$n$. To address these constraints, we derive a system of differential inequalities that guide the design of the pricing function $\phi$. For each possible utilization level at the arrival of request~$n$, we consider the worst-case scenario for utilization fluctuations over its rental interval. Consequently, the design constraints for $\phi$ naturally arise from this worst-case analysis, embedding such fluctuations explicitly into the pricing function.  The following theorem formalizes this result.

\begin{theorem}
\label{prop:kRentalVD-Dynamic-CR}
 \algorithmVD is $\alpha$-competitive for the \problemkRentalVD problem, provided that the pricing function $\phi$ is increasing and satisfies the following inequalities 
 for all $d_n \in \bigl[d_{\min}, d_{\max}\bigr]$ and $ n \in [N] $:
\begin{align}
\label{eq:kRental:VD:price:design:1}
\begin{cases}
\int_{\eta=y_1}^{2y_1} \frac{2\alpha}{3}\,\phi(\eta)\,d\eta 
\;+\; \frac{\alpha}{3}\,d_n\,\Bigl(\phi^*(d_n)\;-\;2\,y_1\Bigr)
&\;\ge\; d_n,
\quad \forall\,
y_1 \in \Bigl[0,\,\tfrac{\phi^*(d_n)}{2}\Bigr],\\[6pt]
\frac{\alpha}{3}\,d_n\,\phi^*(d_n)
\;-\; \frac{\alpha}{3}\,y_2\,\Bigl(d_n\;-\;\phi(y_2)\Bigr)
&\;\ge\; d_n,
\quad \forall\,
y_2 \in \Bigl[0,\,\phi^*(d_n)\Bigr],
\end{cases}
\end{align}   
where $\phi^*(d_n) = \sup_{x \in [0,1]}\{\phi(x) \leq d_n \} $.
\end{theorem}

The detailed proof of the theorem is provided in Appendix~\ref{apx:prop:kRental:VD:CR}. 
By Theorem~\ref{prop:kRentalVD-Dynamic-CR}, designing an algorithm with guaranteed performance reduces to constructing a pricing function $\phi$ that satisfies the constraints in Eqs.~\eqref{eq:kRental:VD:price:design:1},
while minimizing the competitive ratio $\alpha$. 
However, solving $\phi$ is challenging in general because it involves solving a system of \textit{delayed differential inequalities} with an inverse term.

\paragraph{Tightness of results from \algorithmVD.}
Based on Theorem \ref{prop:kRentalVD-Dynamic-CR}, we can show that \algorithmVD is 
$3\cdot (1+\ln(\frac{d_{\max}}{d_{\min}} ))$-competitive for the \problemkRentalVD problem when the pricing function $\phi$ is designed as: $
\phi(y) = d_{\min} \cdot \exp([1+\ln(\frac{d_{\max}}{d_{\min}})]\cdot y - 1), \forall y \in [0,1] $. This also matches the best-known competitive ratio for \problemkRentalVD in prior work~\cite{Goyal2020OIS}.
On the other hand, we can prove that no online algorithms can attain a competitive ratio smaller than $ 1+\ln({d_{\max}}/{d_{\min}})$ for the \problemkRentalVD problem. 
Therefore, one simple feasible design of $\phi$ based on Theorem~\ref{prop:kRentalVD-Dynamic-CR} can already attain an order-optimal competitive ratio for \problemkRentalVD. We defer all proofs related to the order-optimality to Appendix~\ref{apx:order:optimal:alg:variable}.

\paragraph{Comparison to prior work.
}

To demonstrate the significance of our results, we further develop a numerical method in Appendix \ref{apx:numerical:phi:computation} to solve Eqs. \eqref{eq:kRental:VD:price:design:1} and obtain a pricing function. In Figure \ref{fig:kRental:VD:comparison:numerical}, we show that the competitive ratio achieved by \algorithmVD (blue curve) surpasses the best known bound of $3\left(1 + \ln(d_{\max}/d_{\min})\right)$ from \cite{Goyal2020OIS}.

As a relax and round algorithm, we also investigate the performance of the fractional solution from the relaxation step. Recall that \algorithmVD adopts an integrated design of the relaxation and rounding. Thus, deriving the fractional solution without the rounding step needs to additionally modify the LP certificate conditions in Eqs.~\eqref{eq:kRental:VD:price:design:1}. We show the details in Appendix~\ref{apx:krental:frac} and call the resulting new algorithm  \algorithmVDfrac.
The competitive ratio of \algorithmVDfrac is illustrated in green curve in Figure~\ref{fig:kRental:VD:comparison:numerical}.  We observe that its performance is comparable to $4 + \ln(d_{\max}/d_{\min})$, which is the best-known competitive ratio of \problemkRentalVD in the large inventory regime (as $k\to\infty$)~\cite{ekbatani2025}. 
We highlight that it is possible to design a new algorithm for the fractional solution of \problemkRentalVD using the techniques developed in~\cite{ekbatani2025} (See Appendix~\ref{apx:comparison} for more details).
In comparison, our proposed \algorithmVDfrac uses a far simpler, \emph{duration-oblivious} pseudo-utility maximization rule: the fractional decision for each arriving request depends only on the utilization level at the arrival of the request, whereas the new algorithm based on~\cite{ekbatani2025} must additionally track the utilization fluctuation over the entire rental interval of every request.

\begin{figure}
    \centering
    \includegraphics[width=0.5\linewidth]{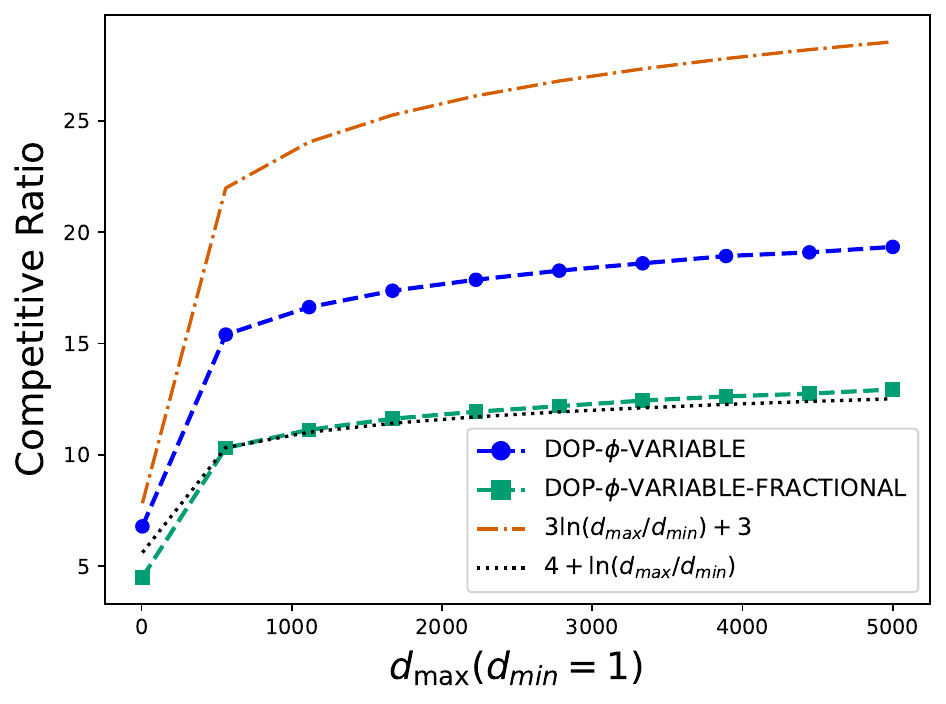}
    \caption{The blue and green curves illustrate the competitive ratios of Algorithm \ref{alg:kRentalVD-dynamic} in the integral and fractional settings (using the numerically derived pricing function~$\phi$). The brown and black curves serve as benchmarks, which correspond to the best known bound for general $k$~\cite{Goyal2020OIS} and the bound for large $k$~\cite{ekbatani2025}.}
    \label{fig:kRental:VD:comparison:numerical}
\end{figure}

Furthermore, due to the impossibility results of \gocr under variable durations, we are unable to design an online rounding scheme for the fractional solution based on~\cite{ekbatani2025}. In contrast, our proposed algorithm \algorithmVD integrates the design of deriving the fractional solution and the rounding scheme, which results in the improved competitive solutions.  
Designing a tight online rounding scheme for general online fractional solutions remains an interesting open question.

\paragraph{Insights into the hardness of \problemkRentalVD.}

To design a pricing function that satisfies the constraints in Theorem~\ref{prop:kRentalVD-Dynamic-CR} with the minimum possible value of $\alpha$, it suffices to enforce these constraints only at the critical points where the left-hand side of the inequalities is minimized. Furthermore, achieving the optimal value of $\alpha$ requires that the inequalities hold with equality. This condition leads to a system of \textit{delayed differential equations}. Notably, the fact that the pricing function attaining the tightest competitive ratio arises from such a system highlights a \textit{memory} effect: past prices, which are used to reserve a fraction of the resource for short-duration requests, affect the current pricing dynamics. As a result, the pricing function becomes sensitive to the durations of previously accepted requests, particularly those assigned fractional allocations. This sensitivity implies that the pricing function not only sets current prices but also implicitly manages request scheduling by balancing resource reservations in anticipation of returns from short-duration requests. We believe that this provides a new perspective on the fundamental hardness of \problemkRentalVD and related online allocation problems involving reusable resources, such as in scheduling and matching \cite{Huang_2024_Survey}.

\section{Conclusion and Future Directions}
We studied two adversarial online allocation problems involving reusable resources: the online $k$-rental problem with fixed rental durations (\problemkRentalFD) and a more general variant with variable rental durations (\problemkRentalVD). To address both settings, we proposed a unified relax-and-round framework that leverages a price-based approach to compute fractional solutions and novel online rounding schemes to convert them into integral decisions. For \problemkRentalFD, we integrated the price-based strategy with a new lossless online rounding scheme, \locr, achieving the optimal competitive ratio. In the more general \problemkRentalVD setting, where lossless rounding is provably unattainable, we developed a limited-correlation rounding strategy that achieves an order-optimal guarantee. We believe that the lossless online rounding technique introduced in this work is broadly applicable and has the potential to enhance algorithms for a wide range of online resource allocation problems.

This work opens several promising directions for future research. In particular, two open problems stand out as especially compelling: (i) designing more powerful rounding schemes that further narrow the gap between fractional and integral solutions in the variable-duration setting, and (ii) establishing tight upper bounds on the best achievable \gocr under variable rental durations, despite the impossibility of attaining \locr as shown by our result.

\clearpage
\begin{center}
   \textbf{\LARGE Appendix}
\end{center}
\appendix

\section{Proof of Proposition~\ref{lem:ir}}
\label{apx:proof:indp:rounding}

Let $P_n = \{i \in [n-1] \mid a_i + d > a_n \}$ denote the set of players whose rental intervals overlap with that of player~$n$. Let $E_n$ denote the event that a ball is allocated to player~$n$. Define $A_i$ as the event that the random seed $s_n \leq f_k \cdot \hat x_n$ for the $n$-th player, for each $i \in [N]$.
For convenience in notation, we occasionally abuse the notation by using $E_n$ or $A_n$ to also denote the indicator random variable $\mathbb{I}_{E_n}$ and $\mathbb{I}_{A_n}$.
Then, we have:
\begin{align*}
    \pr[E_n] &= \pr[A_n] \cdot \pr\left[\sum_{i \in P_n} E_i < k\right] \geq \pr[A_n] \cdot \pr\left[\sum_{i \in P_n} A_i < k\right] = f_k \cdot \hat x_n \cdot \pr\left[\sum_{i \in P_n} A_i < k\right].
\end{align*}
The inequality $\pr\left[\sum_{i \in P_n} E_i < k\right] \geq \pr\left[\sum_{i \in P_n} A_i < k\right]$ holds because, whenever $\sum_{i \in P_n} A_i < k$ occurs, the event $\sum_{i \in P_n} E_i < k$ must also occur.  
This is true since if fewer than $k$ players satisfy the seed condition (which means that they are eligible for allocation), then at most $k - 1$ players can have a resource unit under allocation at the arrival of player~$n$.
Thus, to prove the proposition, it suffices to lower bound the term $\pr\left[\sum_{i \in P_n} A_i < k\right]$ and show that this probability is at least $1 - \exp\left( - \frac{(k - f_k \cdot k)^2}{f_k \cdot k + k } \right)$. 

The random variables $\{ A_i \}_{i \in P_n}$ are mutually independent since the random seed $s_i$ for each player $i$ is sampled independently from the uniform distribution.

Furthermore, we have
\[
\mathbb{E}\left[\sum_{i \in P_n} A_i\right] = \sum_{i \in P_n} \mathbb{E}[A_i] = \sum_{i \in P_n} f_k \cdot \hat x_i \leq f_k \cdot k,
\]
where the last inequality follows from the condition in Eq.~\eqref{eq:ocr_conditions}.

We can restate an alternative form of the Chernoff bound as follows:  
\begin{lemma}[Chernoff Bound]
Let $\{X_i\}_{i=1}^n$ be independent random variables taking values in $[0,1]$, and let $S = \sum_{i=1}^n X_i$ with mean $\mu = \mathbb{E}[S]$. Then, for any threshold $A > \mu$, $\pr[S \ge A] \le  \exp\left(-\frac{(A - \mu)^2}{A + \mu}\right).$
  
\end{lemma}

Applying the above form of Chernoff bound gives
    \begin{align*}
        \pr\left[\sum_{i \in P_n} A_i \ge k \right] \leq \exp\left( - \frac{(k-\mathbb{E}[\sum_{i \in P_n} A_i])^2}{\mathbb{E}[\sum_{i \in P_n} A_i] + k} \cdot  \right) \leq \exp\left( - \frac{(k-f_k\cdot k)^2}{ f_k \cdot k + k}  \right), 
    \end{align*}
    where the second inequality follows from the inequality $\mathbb{E}[\sum_{i \in P_n} A_i] \leq f_k \cdot k$ and the fact that function $\exp\left( - \frac{(k-x)^2}{ x + k}   \right) $ is increasing when $x \leq k$ and $f_k < 1$ .
    Putting together everything, we have
        \begin{align*}
            \pr[E_n] = \pr[A_n] \cdot \pr[\sum_{i \in P_n} A_i < k]  \leq \hat x_n \cdot f_k \cdot (1 - \exp\left( - \frac{(k-f_k\cdot k)^2}{  f_k \cdot k + k} \right) ).
    \end{align*}
We thus complete the proof of Proposition \ref{lem:ir}.

\section{Proofs and Examples for \locr}

\subsection{Proof of Proposition~\ref{lemma:kRental:VD:loca:availability}}
\label{apx:proof:lemma:kRental:VD:loca:availability}

For each player $n \in [N]$ and ball $i \in [k]$, we define the set $R_{n}^{(i)}$ as follows:
\begin{align}
R_{n}^{(i)}  = \Bigl\{\bigl(p_{j}, \max\{1,  p_{j} + \hat x_{j}\}, j\bigr)
           \big| j \in B_{n}^{(i)}\Bigr\} \cup 
           \Bigl\{\bigl(0, p_{j} + \hat x_j - 1, j\bigr)
           \big| j \in B_{n}^{(i-1)},  p_{j} + \hat x_j > 1\Bigr\},
\end{align}
where the set of players in $B^{(i)}_{n}$ are defined as 
\begin{align*}
   B^{(i)}_{n} = \Bigl\{ j \in [n] \mid a_{j} + d > a_n, m_{j} = i, \hat x_j \neq 0 \Bigr\}.
\end{align*}

For each ball $i,$ the set $R_{n}^{(i)}$ contains all triplets that each correspond to the process of the algorithm deciding to allocate ball $i$ to player $n$. Specifically, if an element of the form
$\Bigl(p_{j}, \max\{1, p_{j} + \hat x_j\}, j\Bigr)$
appears in $R_{n}^{(i)}$ and $j \in B_{n}^{(i)},$ then upon the arrival of player $j,$ if the random seed $r$ lies in the range
$(p_{j}, \max\{1, p_{j} + \hat x_j\}]$,
Algorithm \ref{alg:locr} allocates ball $m_{j} = i$ to player $j.$

Similarly, if 
$\Bigl(0,  p_{j} + \hat x_j - 1,  j\Bigr)$ 
belongs to $R_{n}^{(i)}$ and $j \in B_{n}^{(i-1)}$, then upon the arrival of player $j$, if 
$r \in [ 0,  p_{j} + \hat x_j - 1)$,
Algorithm~\ref{alg:locr} allocates ball $m_{j} + 1 = i$ to player $j$. Thus, the element
$\Bigl(0,  p_{j} + \hat x_j - 1,  j\Bigr)$
is added to $R_{n}^{(i)}$ to reflect this allocation decision.

In what follows, we show that if Proposition~\ref{lemma:kRental:VD:loca:availability} fails, then there must exist two elements $(s_1,e_1,n_1)$ and $(s_2,e_2,n_2)$ in either $R_{n}^{(m_n)}$ or $R_{n}^{(m_{n+1})}$ whose intervals $[s_1,e_1)$ and $[s_2,e_2)$ intersect. We consider the following two cases, reflecting those discussed in Proposition~\ref{lemma:kRental:VD:loca:availability}:
\begin{itemize}
    \item \textbf{Case 1: $p_n + \hat x_n < 1$.} If Proposition~\ref{lemma:kRental:VD:loca:availability} holds, then whenever $r \in [p_n, p_n+\hat x_n)$, ball $m_n$ remains available at the arrival of player $n$. Suppose the proposition fails in this case. Then, for some $r \in [p_n, p_n+\hat x_n)$, ball $m_n$ is unavailable. Hence, there must exist a player $j < n$ such that if $r \in [s_1,e_1)$, with $[s_1,e_1)\cap [p_n, p_n+\hat x_n)\neq\varnothing$, ball $m_n$ was allocated to $j$. Thus, in this case Proposition~\ref{lemma:kRental:VD:loca:availability} does not hold and by the definitions of $R_{n}^{(m_n)}$, the set $R_{n}^{(m_n)}$ contains two triplets $(s_1,e_1,j)$ and $(p_n,p_n+\hat x_n,n)$ such that range $[s_1,e_1)$ and $[p_n,p_n+\hat x_n)$ overlap.

    \item \textbf{Case 2: $p_n + \hat x_n \ge 1$.} A similar argument shows that if Proposition~\ref{lemma:kRental:VD:loca:availability} fails in this case, there must exist two elements $(s_1,e_1,n_1)$ and $(s_2,e_2,n_2)$ in either $R_{n}^{(m_n)}$ or $R_{n}^{(m_{n+1})}$ where the intervals $[s_1,e_1)$ and $[s_2,e_2)$ intersect.
\end{itemize}

Next, we show that no set $R_{n}^{(i)}$ (for any $n \in [N]$ and $i \in [k]$) can contain two elements whose corresponding intervals overlap. Indeed, Proposition~\ref{prop:forms-R-ocs} establishes that $R_{n}^{(i)}$ can only be in one of four specific forms, none of which allows two elements $(s_1,e_1,n_1)$ and $(s_2,e_2,n_2)$ with overlapping intervals $[s_1,e_1)$ and $[s_2,e_2)$. 
Let us define $l_{n}^{(i)} 
= \sum_{(s,e,t) \in R_{n}^{(i)}} (e - s)$.
\begin{proposition}
\label{prop:forms-R-ocs}
At the arrival time of each player $n \in [N]$ and for each ball $i \in [k]$, if the set $R_{n}^{(i)}$ is non-empty, then upon sorting $R_{n}^{(i)}$ by the first element of each triplet, one of the following four forms arises:

\noindent
\textbf{(i)} If $i = m_n$ and $\hat x_n + p_n < 1$, then for some $j \in \mathbb{N}$, $R_{n}^{(i)}$ takes the form
\begin{align*}
 R_{n}^{(i)} 
= \biggl\{\bigl(s_{1}, s_{2}, t_{1}\bigr), \dots, (s_{j} = p_n, s_{j+1}=p_{n} + \hat x_n, t_{j}),  (s_{j+1}+1 - l_{n}^{(i)},  s_{j+2}, t_{j+1}), \dots, (s_{q}, s_{q+1}=1, t_{q})\biggr\},
\end{align*}
where $ t_{j+1} < t_{j+2} < \cdots < t_{q} < t_{1} < t_{2} < \cdots < t_{j} $ and $ s_{1} \le s_{2} \le \cdots \le s_{j+1} \le s_{j+2} \le \cdots \le s_{q+1}=1 $.

\noindent
\textbf{(ii)} If $i = m_n$ and $\hat x_n + p_n \ge 1$, $R_{n}^{(i)}$ takes the form
\begin{align*}
R_{n}^{(i)} = \biggl\{\bigl(s_{1}, s_{2}, t_{1}\bigr),  \bigl(s_{2}, s_{3}, t_{2}\bigr) \dots, \bigl(s_{q} = p_n, s_{q+1}=1, t_{q}\bigr)\biggr\},
\end{align*}
where $ t_{1} < t_{2} < \cdots < t_{q} $ and $
s_{1} \le s_{2} \le \cdots s_q = p_n \le s_{q+1} = 1 $.

\noindent
\textbf{(iii)} If $i = m_n + 1 $ and $\hat x_n + p_n > 1$, then $R_{n}^{(i)}$ takes the form
\begin{align*}
R_{n}^{(i)} = \biggl\{(s_1 = 0, s_2 = p_n+\hat x_n-1,t_1 = n), \bigl(s_2,  s_{3}, t_{2}\bigr),  (s_{3}, s_{4}, t_{3}), \dots, (s_{q}, s_{q+1}=1, t_{q})\biggr\},
\end{align*}
where $ t_{2} < \cdots < t_{q} < t_{1} = n $ and $
s_{1} = 0 < s_{2} < \cdots < s_{q+1} $.

\noindent
\textbf{(iv)} If $ i = m_n + 1 $ and $\hat x_n + p_n \leq 1$,  or if $i \in [k] - \{m_n,m_n+1\}$, then $R_{n}^{(i)}$ takes the form
\begin{align*}
R_{n}^{(i)} = \biggl\{\bigl(s_{1}=1 - l_{n}^{(i)},  s_{2}, t_{1}\bigr),  (s_{2}, s_{3}, t_{2}), \dots, (s_{q}, s_{q+1}=1, t_{q})\biggr\},
\end{align*}
where $ t_{1} < t_{2} < \cdots < t_{q} $ and $
1 - l_{n}^{(i)} = s_{1} \le s_{2} \le \cdots \le s_{q+1}=1 $.
\end{proposition}

\begin{proof}
We prove the proposition by induction on $N$, the number of arriving players.

\textbf{Base case ($N=1$).}
In this scenario, we have $p_1=0$ and $m_1 = 1$, leading to $R_{1}^{(1)} = \{(0, x_1, 1)\}$. If $\hat x_n + p_n < 1$ then $R_{1}^{(1)}$ takes the first form $j = 1$, otherwise it is in the second form and all the other sets $R_{1}^{(i)}$ are empty. Therefore, the proposition holds for $N=1$.

\textbf{Induction step:} Assume that for $N = M$ the proposition holds, i.e., for each $i \in [k]$, the set $R_{M}^{(i)}$ is in one of the above forms in Proposition~\ref{prop:forms-R-ocs}. We now prove that $R_{M+1}^{(i)}$ is also in one of these forms after processing the $(M+1)$-th player. After the arrival of the $(M+1)$-th player, two types of changes may occur in $R_{M+1}^{(i)}$ compared to $R_{M}^{(i)}$: existing elements may be removed, or new elements may be added.
An element in $R_{M}^{(i)}$ is removed from $R_{M+1}^{(i)}$ if the rental period for the corresponding player ends by time $a_{M+1}$. In other words, for an element corresponding to a player $j$ satisfying $a_{M} < a_{j} + d \le a_{M+1}$, that element is removed from $R_{M+1}^{(i)}$. In the following discussion, we examine the two scenarios that can arise depending on whether an element is added or removed from the set $R_{M}^{(i)}$.

\subsubsection{\textsc{Scenario 1 (element removed)}} 
Suppose an element is removed from $R_{M}^{(i)}$ upon the arrival of player $M+1$. We consider four subcases, depending on which form $R_{M}^{(i)}$ takes.

\medskip
\noindent
$\bullet$ \textbf{Subcase-1a: $R_{M}^{(i)}$ is in the first form.} Then we have 
\begin{align*}
R_{M}^{(i)} = \{(s_{1}=0,  s_{2},  t_{1}),  \dots,  (s_{j}=p_n,  s_{j+1}=p_n + \hat x_n,  t_{j}),  (s_{j+1}+1 - l_{n}^{(i)},  s_{j+2},  t_{j+1}),  \dots,  (s_{q},  s_{q+1}=1,  t_{q})\}.
\end{align*}
for some values $\{s_i\}$ and $\{t_i\}$ satisfying the conditions of the first form in Proposition~\ref{prop:forms-R-ocs}.

In the set $R_{M}^{(i)}$, the first elements to be removed—based on the arrival times of these players and the fact that they appear in $R_{M}^{(i)}$ only for a fixed duration—are those corresponding to players $j+1$ up to $q$ and $1$ to $j$. It is easy to see that the removal of the elements corresponding to players $j+1$ up to $q$ leads $R_{M+1}^{(i)}$ to remain in the first form. The removal of elements corresponding to players $1$ to $j$ causes $R_{M+1}^{(i)}$ to have either the first or second form, depending on the values of $x_{m+1}$ and $p_{m+1}$.

\medskip
\noindent
$\bullet$  \textbf{Subcase-1b: $R_{M}^{(i)}$ is in the second form.} Then we have 
\begin{align*}
R_{M}^{(i)} = \{(s_{1},  s_{2},  t_{1}),  (s_{2},  s_{3},  t_{2}),  \dots,  (s_{q}=p_{M},  s_{q+1}=p_{M+1},  t_{q})\},
\end{align*}
for some values $\{s_i\}$ and $\{t_i\}$ satisfying the conditions of the second form in Proposition~\ref{prop:forms-R-ocs}.

Since $t_{1} \le t_{i}$ for all $i \in [q]$, $(s_{1},  s_{2},  t_{1})$ is the oldest element and thus the first to be removed. Consequently, if one element is removed upon the arrival of player $M+1$, it must be $(s_{1},  s_{2},  t_{1})$. A straightforward check shows that after this removal, $R_{M+1}^{(i)}$ continues to be in the fourth form. If more than one element is removed, a similar argument implies that $R_{M+1}^{(i)}$ will be in the fourth form.

\medskip
\noindent
$\bullet$ \textbf{Subcase-1c: $R_{M}^{(i)}$ is in the third form.} If an element is removed from $R_{M}^{(i)}$ upon the arrival of player $M+1$, it must be the oldest element in that set. By an argument analogous to the subcases above, $R_{M+1}^{(i)}$ then transitions to either the first or second form depending on the values of $x_{m+1}$ and $p_{m+1}$.

\medskip
\noindent
$\bullet$ \textbf{Subcase-1d: $R_{M}^{(i)}$ is in the fourth form.} Here, removing an element from $R_{M}^{(i)}$ at the arrival of player $M+1$ causes $R_{M+1}^{(i)}$ to remain in the same form.

\subsubsection{\textsc{Scenario 2 (element is added)}} 
Suppose a new element corresponding to the decision made for player $M+1$ is added to  the set $R_{M+1}^{(i)}$.
Again, we consider the following cases:

\medskip
\noindent
$\bullet$ \textbf{Subcase-2a: $R_{M}^{(i)}$ in the first form.} In this case, we have 
\begin{align*}
R_{M}^{(i)} = \{(s_{1}=0,  s_{2},  t_{1}),  \dots,  (s_{j}=p_{M},  s_{j+1}=p_{M+1},  t_{j}),  (s_{j+1}+1-l_M^{(i)},  s_{j+2},  t_{j+1}),  \dots,  (s_{q},  s_{q+1}=1,  t_{q})\}.
\end{align*}
Suppose a new element $(s, e, M+1)$ is added at the arrival time of the $(M+1)$-th player. By the design of Algorithm \ref{alg:locr}, we have $s = p_{M+1}$ and $e = p_{M+1} + x_{M+1}$. Thus, this new element is inserted after the element $(s_{j},  s_{j+1}=p_{M+1},  t_{j})$. To ensure that $R_{M+1}^{(i)}$ remains in the proper form (i.e., either in the first or second form), we must have 
\begin{align*}
e \le s_{j+1}+1-l_{M}^{(i)}.
\end{align*}

We now formalize and prove this requirement.

\begin{lemma}
\label{lemma:kRental:locr:proof}
Consider the case where the set $R_{M}^{(i)}$ is in the first form as described above, and suppose that upon the arrival of player $M+1$, a new element $(s, e, M+1)$ is added with $s = p_{M+1}$ and $e = p_{M+1} + x_{M+1}$. Then, it holds that 
\begin{align*}
e \le s_{j+1} + 1 - l_{M}^{(i)},
\end{align*}
where $l_{M}^{(i)} = \sum_{(s,e,t) \in R_{M}^{(i)}} (e - s)$, as previously defined, represents the total sum of the range sizes corresponding to each element $(s,e,t)$ in $R_{M}^{(i)}$.
\end{lemma}

\begin{proof}
Consider the stage at which the algorithm has completed processing the $t_q$-th player and has added the element $(s_q, s_{q+1} = 1, t_q)$ to $R_{t_q}^{(i)}$. By construction, at this point the pointer $p_{t_q+1}$ reaches $1$, and the pointer $m_{t_q}$ is advanced. Between the processing of the $t_q$-th and $t_1$-th players, the pointer $m_n$ is advanced $k - 1$ times until it once again points to ball $i$, at which point a new element $(s_1, s_2, t_1)$ is added to $R_{t_1}^{(i)}$. For the pointer $m_n$ to be advanced $k - 1$ times, there must exist a sequence of players from the $t_q$-th to the $t_1$-th such that the following condition holds:
\begin{align*}
\sum_{j=t_{q}+1}^{t_1-1} \hat x_j + \Bigl(x_{t_q} + p_{t_q} - 1\Bigr) + \Bigl(1 - p_{t_1}\Bigr) \cdot \mathbf{1}_{\{m_{t_1} = i -1 \}} = k-1.
\end{align*}
Next, it follows that if 
\begin{align*}
e > s_{j+1} + 1 - l_{M}^{(i)},
\end{align*}
then 
\begin{align*}
\sum_{j=1}^{M+1} \hat x_j \cdot \mathbf{1}_{\{a_{j}+d > a_{M+1}\}} > k,
\end{align*}
which contradicts the constraint in Eq.~\eqref{eq:ocr_conditions}. Therefore, we must have 
\begin{align*}
e \le s_{j+1} + 1 - l_{M}^{(i)}.
\end{align*}
Thus, we complete the proof of Lemma \ref{lemma:kRental:locr:proof}.
\end{proof}

Thus, by Lemma \ref{lemma:kRental:locr:proof}, if a new element is added to $R_{M+1}^{(i)}$ when $R_{M}^{(i)}$ is in the first form, no two elements will have overlapping intervals, and the set remains in either the first or second form. This completes the argument.

\medskip
\noindent
$\bullet$ \textbf{Subcase-2b: $R_{M}^{(i)}$ in the second form.} \quad
Since at the arrival time of player $M$ we have $i = m_n$, if a new element is added to $R_{M+1}^{(i)}$ at the arrival time of player $M+1$, it must be that $k = 1$. In this case, the newly added element takes the form 
$(p_{M+1}, p_{M+1} + x_{M+1}, M+1)$. 
By the same reasoning, if the interval 
$(p_{M+1}, p_{M+1} + x_{M+1}]$ 
overlaps with any other element in $R_{M+1}^{(i)}$, this would contradict the constraint specified in Eq.~\eqref{eq:ocr_conditions}.

\medskip
\noindent
$\bullet$ \textbf{Subcase-2c: $R_{M}^{(i)}$ in the third form.} \quad
In this case, we have
\begin{align*}
R_{M}^{(i)}
&= \Bigl\{\bigl(s_1=0, s_2=p_M+x_M-1, t_1=M\bigr), \bigl(s_2,  s_{3},  t_{2}\bigr), \dots, \bigl(s_{q},  s_{q+1}=1, t_{q}\bigr)\Bigr\}.
\end{align*}
Suppose a new element $(s,e,t)$ is added at the arrival time of player $M+1$. Then, by the design of Algorithm \ref{alg:locr}, 
\begin{align*}
s &= p_M+x_M-1 = p_{M+1},\\[2mm]
e &= \max\{1, p_{M+1}+x_{M+1}\}.
\end{align*}
Thus, this new element is appended after the element $(s_{1}, s_{2}, M)$. Following the same reasoning as in Subcase-1a, we must have $e < s_2$, otherwise a contradiction arises with the constraint in Eq.~\eqref{eq:ocr_conditions}.

\medskip
\noindent
$\bullet$ \textbf{Subcase-2d: $R_{M}^{(i)}$ in the fourth form.} \quad
Here, we have
\begin{align*}
R_{M}^{(i)}
&= \Bigl\{\bigl(s_1=1-l_{M}^{(i)},  s_{2},  t_{1}\bigr), (s_{2},  s_{3},  t_{2}), \dots, (s_{q},  s_{q+1}=1,  t_{q})\Bigr\}.
\end{align*}
Suppose a new element $(s,e,t)$ is added at time $M+1$. By the design of Algorithm~\ref{alg:locr}, we have
\begin{align*}
s &= 0 = p_{M+1},\\[2mm]
e &= \max\{1, p_{M+1}+x_{M+1}\}.
\end{align*}
Hence, this new element is inserted before $(s_{1}, s_{2}, t_{1})$. Again, by analogous reasoning to Subcase-1a, we must have $e < s_{2}$; otherwise, we reach a contradiction with the constraint in Eq.~\eqref{eq:ocr_conditions}.

Putting everything together, we complete the proof of Proposition \ref{prop:forms-R-ocs}.
\end{proof}

Our above proof of Proposition~\ref{prop:forms-R-ocs} shows that $R_{n}^{(i)}$ can only take one of four specific forms, none of which permits two elements $(s_1,e_1,n_1)$ and $(s_2,e_2,n_2)$ with overlapping intervals $[s_1,e_1)$ and $[s_2,e_2)$. Thus, Proposition~\ref{lemma:kRental:VD:loca:availability} follows.

\subsection{Proof of Theorem~\ref{thm:correctness-OCS-algorithm}}
\label{apx:thm:correctness-OCS-algorithm}
We consider two cases based on the relation between $1-p_n$ and $\hat x_n$.

\textbf{Case I:} $1-p_n \geq \hat x_n$. In this case, the algorithm assigns ball $m_n$ to player $n$ if ball $m_n$ is available and $r \in [p_n, p_n+\hat x_n)$. Let $E_{m_n}$ denote the event that ball $m_n$ is available upon the arrival of the $n$-th player at time $a_n$, and let $E$ denote the event $r \in [p_n, p_n+\hat x_n)$. Then the probability that ball $m_n$ is allocated to player $n$ is 
\begin{align*}
\Pr[E_{m_n} \cap E].
\end{align*}

Based on the definition of the set $R^{(m_n)}_n$, it contains all the ranges that have been used to allocate ball $m_n$ to previous players, including those whose rental players overlap with player $n$. Proposition~\ref{prop:forms-R-ocs}  ensures that the range $[p_n, p_n+\hat x_n)$ does not overlap with any other range in $R^{(m_n)}_n$. Thus, if the random sample $r \in [p_n, p_n+\hat x_n)$, then ball $m_n$ is guaranteed to be available at the arrival of player $n$. So if $E$ occurs, ball $m_n$ is available. Consequently,
\begin{align*}
   & \Pr[E_{m_n} \cap E] \\
=\ & \Pr[E] \cdot \Pr[E_{m_n}\mid E] \\
=\ & \Pr[E] \\
=\ & \Pr[r \in [p_n, p_n+\hat x_n)] \\
=\ & \hat x_n.
\end{align*}
Hence, when $\hat x_n < 1-p_n$, player $n$ gets ball $m_n$ with probability $\hat x_n$.

\textbf{Case II:} $1-p_n < \hat x_n$. In this case, the algorithm first attempts to allocate ball $m_n$ if it is available and if $r \in [p_n, 1]$. If this allocation does not occur, it then attempts to allocate ball $m_n+1$ provided that it is available and $r \in [0, \hat x_n-1+p_n)$. Let $E_{m_n}$ denote the event that ball $m_n$ is available at time $a_n$, and let $E_1$ denote the event $r \in [p_n, 1]$. Similarly, let $E_{m_n+1}$ denote the event that ball $m_n+1$ is available, and let $E_2$ denote the event $r \in [0, \hat x_n-1+p_n)$. Then the total probability that a ball is allocated to player $n$ is given by
\begin{align*}
\Pr[E_{m_n} \cap E_1] + \Pr\Bigl[E_{m_n+1} \cap E_2 \cap (E_{m_n} \cap E_1)'\Bigr].
\end{align*}
Based on Proposition~\ref{prop:forms-R-ocs} and the fact that $p_n \neq 0$, the set $R_{n}^{(m_n)}$ is in the first form and the range $[p_n, 1]$ does not overlap with any previous ranges in $R_{n}^{(m_n)}$. Therefore, if $E_1$ occurs, ball $m_n$ is available, and
\begin{align*}
\Pr[E_{m_n} \cap E_1] = \Pr[E_1] = 1-p_n.
\end{align*}
Similarly, the event $E_{m_n+1} \cap E_2$ occurs with probability
\begin{align*}
\Pr[E_{m_n+1} \cap E_2] = \Pr[E_2] = \hat x_n - 1 + p_n.
\end{align*}
Since the ranges $[p_n, 1]$ and $[0, \hat x_n-1+p_n)$ do not overlap (because $\hat x_n \leq 1$), we have
\begin{align*}
   & \Pr\Bigl[E_{m_n+1} \cap E_2 \cap (E_{m_n} \cap E_1)'\Bigr] \\
=\ & \Pr\Bigl[r \in \bigl([0, \hat x_n-1+p_n) \cap [0, p_n)  \bigr)\Bigr] \\
=\ & \Pr\Bigl[r \in [0, \hat x_n-1+p_n)\Bigr] \\
=\ &  \hat x_n - 1 + p_n.
\end{align*}

Thus, the total probability that a ball is allocated to player $n$ is
\begin{align*}
   & \Pr[\text{ball allocated to player } n] \\
=\ &  \Pr[E_{m_n} \cap E_1] + \Pr\Bigl[E_{m_n+1} \cap E_2 \cap (E_{m_n} \cap E_1)'\Bigr] \\
=\ & (1-p_n) + (\hat x_n-1+p_n) \\
=\ & \hat x_n.
\end{align*}

Thus, in both cases, each player $n$ gets a ball with the desired probability $\hat x_n$, which proves that Algorithm~\ref{alg:locr} is a lossless online scheme.

\subsection{Example: How \locr Works in Algorithm~\ref{alg:locr}}

\label{apx:example:locr}

To illustrate how Algorithm \ref{alg:locr} works and demonstrate how it outperforms a rounding procedure that makes independent decisions at each time step, consider the following instance with $k = 2$, $d = 5$, and $N = 4$. The target probability values and arrival time of players are given by:
\begin{align*}
\Big \{ (a_1=1, \hat{x}_1=0.4), \quad  (a_2=2, \hat{x}_2=0.5), \quad (a_3=3,  \hat{x}_3=0.6), \quad  (a_4=6, \hat{x}_4=0.6) \Big \}.
\end{align*}
It can be verified that for each player $n \in \{1, \cdots, 4\} $, the inequality in Eq.~\eqref{eq:ocr_conditions} is satisfied for the above instance.  Initially, before any player arrives, a random seed $r$ is drawn from the uniform distribution $U(0,1)$. This random seed is used throughout the horizon to correlate the algorithm's decisions.

\begin{figure*}[htb]
\centering
\includegraphics[scale=0.3]{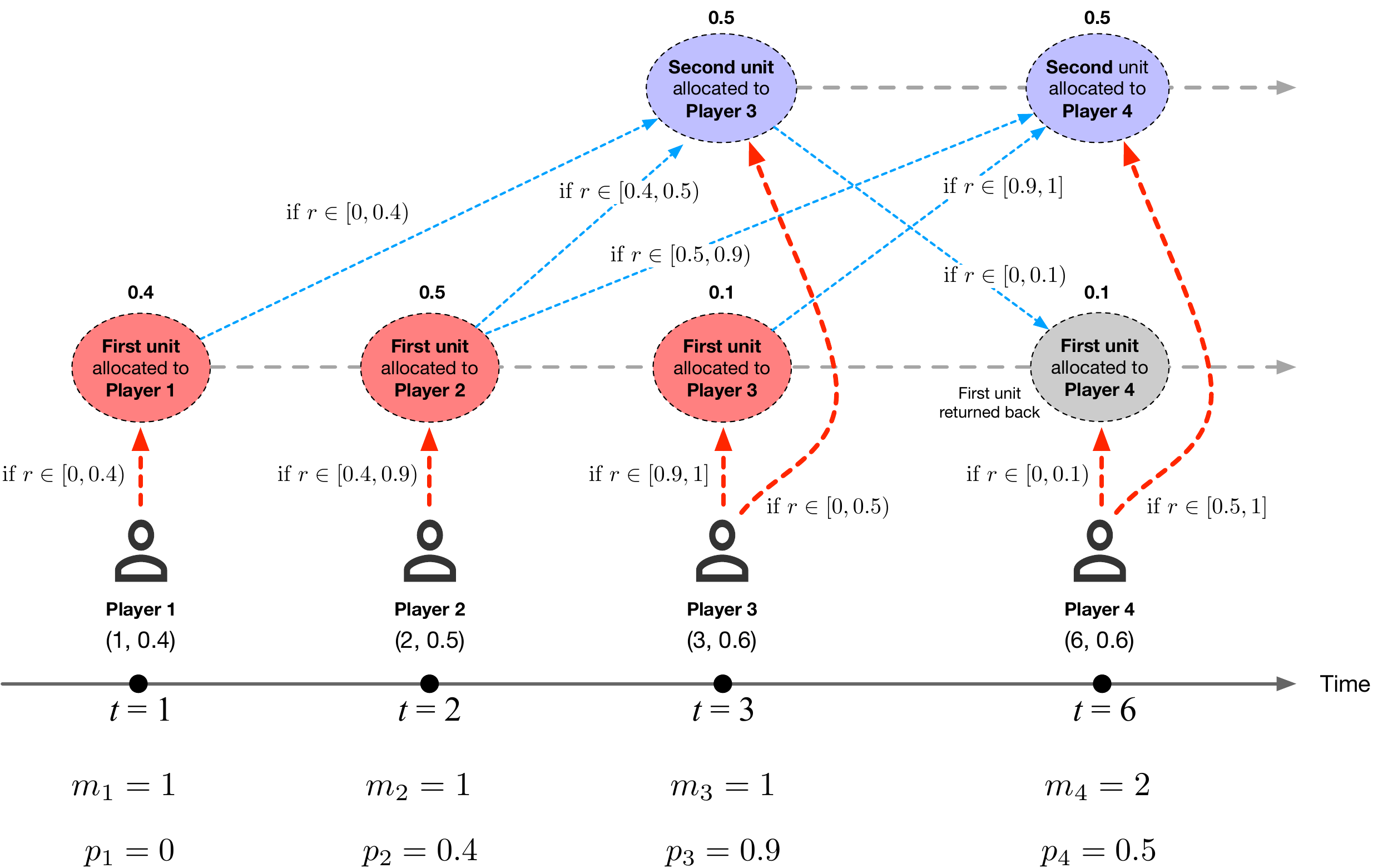}
\caption{Illustration of the online correlated assignment process in Algorithm \ref{alg:locr}. The number above each circle represents the probability of the corresponding event occurring. In this example, player 3 receives the first ball with probability 0.1 and the second ball with probability 0.5. Notably, the assignment of the second ball to player 3 (with probability 0.5) is correlated with the decisions made for players 1 and 2, as illustrated by the two blue dashed arrows on the left. For instance, if the realization of the random seed is $ r = 0.45 $, then player 1 is rejected, player 2 receives the first ball, player 3 receives the second ball, and player 4 is rejected.}
\label{figure:oca_example}
\end{figure*}

\begin{itemize}
    \item When player 1 arrives at time $t=1$ with $\hat{x}_1=0.4$, Algorithm \ref{alg:locr} assigns the first ball to player 1 if $r \in [0,\, 0.4)$. Thus, \textbf{player 1 will receive one ball with probability 0.4}.
    
    \item When player 2 arrives at time $t=2$ with $ \hat{x}_2=0.5$, Algorithm \ref{alg:locr} assigns the first ball provided that $r \in [0.4, 0.9)$. Thus, player 2 receives the first ball with probability $0.5$. Note that if $r \in [0.4, 0.9)$, the first ball is guaranteed to be available because it was rented at $t=1$ only if $r \in [0,0.4)$. Thus, \textbf{player 2 will receive one ball with probability 0.5}.
    
    \item When player 3 arrives at time $t=3$ with $\hat{x}_3=0.6$, Algorithm \ref{alg:locr} assigns the first ball if it is available and $r \in [0.9, 1]$, which happens with probability 0.1. If the first ball is not allocated, Algorithm \ref{alg:locr} assigns the second ball if $r \in [0, 0.5)$. Consequently, \textbf{the total probability of allocating one ball to player 3 is 0.6}.
    
    \item Finally, when player 4 arrives at time $t=6$ with $ \hat{x}_4=0.6$, Algorithm \ref{alg:locr} assigns the second ball if $r \in [0.5, 1]$, an event that occurs with probability 0.5; assigns the first ball if $ r\in [0, 0.1) $, an event that occurs with probability 0.1. Note that for the case when $ r\in [0, 0.1) $, the first ball will be returned by player 1 at $ t = 6 $. Thus, \textbf{player 4 will receive one ball with probability 0.6}.
\end{itemize}
We can see from the above example that Algorithm \ref{alg:locr} always maintains a \locr scheme in that it assigns one ball to every player $n \in [N]$ with  probability exactly $\hat x_n$.

\section{Proof of Theorem~\ref{them:kRentalFD}}
\label{apx:proof:them:kRentalFD}
We use an online primal-dual approach to establish the competitive ratio of Algorithm~\ref{alg:kRentalFD}. Consider the dual LP corresponding to the primal LP in Eq.~\eqref{lp:kRentalFD:primal}:
\begin{align}
\label{lp:kRentalFD:dual}
\min_{\mathbf{\lambda,u}} \quad & \sum_{n \in [N]} u_{n} + k \cdot \sum_{n \in [N]} \lambda_n,\\ 
\label{eq:kRentalFD:dual:cons}
\text{s.t.} \quad & u_{n} + \sum_{j= n}^{N} \lambda_{j}\mathbf{1}_{\{a_n + d > a_{j}\}} \ge v_{n}, \quad \forall n \in [N],\\ 
& \lambda_n \ge 0, \quad \forall n \in [N].
\end{align}

Let $\alg(\mathcal{I})$ represent the expected objective value of Algorithm~\ref{alg:kRentalFD} on an instance $\mathcal{I}$. Moving forward, we first propose a feasible solution for the primal problem given in Eq.~\eqref{lp:kRentalFD:primal}, denoted as $\{x^{\alg}_{n}\}_{n \in [N]}$. We then show that $\alg(\mathcal{I})$ equals the objective value of the primal LP solution.

In the second step, we design a feasible solution for the dual problem, denoted as $\{u^{\alg}_{n},\lambda^{\alg}_{n}\}_{n \in [N]}$, which corresponds to the dual linear program, and define the dual objective
\begin{align*}
   D^{\alg} = \sum_{n \in [N]} u^{\alg}_{n} + k \cdot \sum_{n \in [N]} \lambda^{\alg}_{n}.
\end{align*}
We then establish that the dual solution $\{u^{\alg}_{n},\lambda^{\alg}_{n}\}_{n \in [N]}$ is feasible for the dual LP in Eq.~\eqref{lp:kRentalFD:dual}.

In the final step,
we establish that 
\begin{align*}
   D^{\alg} \le \underbrace{\Bigl(1 + \ln\Bigl(\frac{v_{\max}}{v_{\min}}\Bigr)\Bigr)}_{F} \cdot P^{\alg}.
\end{align*}

After proving these steps, by weak duality we have
\begin{align*}
   \alg(\mathcal{I}) = P^{\alg} \ge \frac{1}{F} \cdot D^{\alg} \ge \frac{1}{F} \cdot \opt(\mathcal{I}),
\end{align*}
where $\opt(\mathcal{I})$ denotes the objective value of the optimal clairvoyant algorithm on instance $\mathcal{I}$. Thus, the $1 + \ln\Bigl(\frac{v_{\max}}{v_{\min}}\Bigr)$-competitiveness of Algorithm~\ref{alg:kRentalFD} follows.

\paragraph{Step-I: Design of primal solution and $\alg(\mathcal{I}) = P^{\alg}$.} 
We set the primal variables as $x^{\alg}_{n} = \hat x_{n}$ for all $n \in [N]$. Let $\Delta_{P^{\alg}}^{(n)}$ denote the increase in the primal LP objective value resulting from updating $x^{\alg}_{n}$, and let $\Delta^{\alg(\mathcal{I})}_{n}$ denote the corresponding increase in the expected objective value of Algorithm~\ref{alg:kRentalFD} when a unit of resource is allocated to request $n$. By our update, we have
\begin{align*}
   \Delta_{P^{\alg}}^{(n)} = v_n \cdot x^{\alg}_{n}.
\end{align*}
Furthermore, 
\begin{align*}
   \Delta^{\alg(\mathcal{I})}_{n} &= v_n \cdot \Pr[\text{a unit of resource is allocated to request } n] \\
   &= v_n \cdot \hat x_{n}.
\end{align*}
The final equality follows because the \locr rounding scheme allocates a unit to request $n$ with probability exactly $\hat x_{n}$. Summing over requests for all $ n \in [N]$ establishes that $\alg(\mathcal{I}) = P^{\alg}$.

\paragraph{Step-II: Design of dual solution,  $\{u^{\alg}_{n}, \lambda^{\alg}_{n}\}_{n \in [N]}$, and feasibility of the dual solution.} Let us initialize all dual variables to zero. For each rental request~$n$, let $\hat x_{n}$ be the fractional allocation chosen by Algorithm~\ref{alg:kRentalFD}. We then perform the following updates:
\begin{align}
\label{eq:kRentalFD:dual:update1}
\lambda^{\alg}_{\nu^*_n} & = \lambda^{\alg}_{\nu^*_n} + F \cdot \int_{\eta = \frac{y_n}{k}}^{\frac{y_n + \hat x_{n}}{k}} \phi(\eta) d\eta,\\
\label{eq:kRentalFD:dual:update2}
u^{\alg}_n & = F \cdot \hat x_{n} \cdot \Bigl( v_n - \phi\Bigl(\frac{y_n + \hat x_{n}}{k}\Bigr)\Bigr),
\end{align}
where $ \nu^*_n = \max\{j \ge n | a_{j} < a_n + d\} $ and $ F = 1 + \ln\Bigl(\frac{v_{\max}}{v_{\min}}\Bigr)$. To prove the feasibility of the above dual solution, we must show that the dual constraint in Eq.~\eqref{eq:kRentalFD:dual:cons} corresponding to each request $n$ is satisfied; that is,
\begin{align*}
    u^{\alg}_{n} + \sum_{j = n}^{N} \lambda^{\alg}_{j} \cdot \indicator{a_n + d > a_{j}} \ge v_n.
\end{align*}

For the $n$-th request, the dual variable $u^{\alg}_{n}$ is updated according to Eq.~\eqref{eq:kRentalFD:dual:update2}, so that
\begin{align*}
    u^{\alg}_{n} \ge F \cdot \hat x_{n} \cdot \Bigl( v_n - \phi\Bigl(\frac{y_n+\hat x_{n}}{k}\Bigr)\Bigr).
\end{align*}
Next, we aim to prove that 
\begin{align}
\label{eq:kRental:FD:proof:theorem}
\sum_{j = n}^{N} \lambda^{\alg}_{j} \cdot \indicator{a_n + d > a_{j}} \ge F \cdot \int_{0}^{\frac{y_n+\hat x_{n}}{k}} \phi(\eta)  d\eta.
\end{align}
Assuming this inequality holds, we first show that the dual constraint in Eq.~\eqref{eq:kRentalFD:dual:cons} is satisfied for each request $n$ and then get back to the proof of the above inequality. Let us consider the following two cases.

\textbf{Case 1:} $\phi\Bigl(\frac{y_n+\hat x_{n}}{k}\Bigr) < v_n$. It can be verified that following from Eq.~\eqref{eq:kRental:FD:fraction:allocation}, we have
\begin{align*}
\hat x_{n} = \max\Bigl\{0, \min\Bigl\{1, k\cdot\phi^{-1}(v_n) - y_n\Bigr\}\Bigr\}.
\end{align*}
Since we have $\phi\Bigl(\frac{y_n+\hat x_{n}}{k}\Bigr) < v_n$, then it follows that $\hat x_{n} = 1$. 
Therefore,
\begin{align*}
u^{\alg}_{n} + \sum_{j = n}^{N} \lambda^{\alg}_{j} \cdot \indicator{a_n + d > a_j}
& \ge F \cdot \Bigl( v_n - \phi\Bigl(\frac{y_n+\hat x_{n}}{k}\Bigr)\Bigr) 
+ F \cdot \int_{0}^{\frac{y_n+\hat x_{n}}{k}} \phi(\eta)  d\eta \\
& \ge F \cdot \left( \int_{\frac{y_n+\hat x_{n}}{k}}^{\phi^{-1}(v_n)} \phi(\eta)  d\eta 
+ \int_{0}^{\frac{y_n+\hat x_{n}}{k}} \phi(\eta)  d\eta \right) \\
& = F \cdot \int_{0}^{\phi^{-1}(v_n)} \phi(\eta)  d\eta \\
& \ge v_n,
\end{align*}
where the second inequality follows from the monotonicity of $\phi$, and the last inequality follows from the design of $\phi$ in Theorem~\ref{them:kRentalFD}.

\textbf{Case 2:} $\phi\Bigl(\frac{y_n+\hat x_{n}}{k}\Bigr) \ge v_n$. We can lower bound the LHS of Eq.~\eqref{eq:kRentalFD:dual:cons} as follows:
\begin{align*}
u^{\alg}_{n} + \sum_{j = n}^{N} \lambda^{\alg}_{j} \cdot \indicator{a_n + d > a_j}
& \ge F \cdot x_{n} \cdot \Bigl( v_n - \phi\Bigl(\frac{y_n+\hat x_{n}}{k}\Bigr)\Bigr)
+ F \cdot \int_{0}^{\frac{y_n+\hat x_{n}}{k}} \phi(\eta)  d\eta \\
& \ge F \cdot \int_{0}^{\frac{y_n+\hat x_{n}}{k}} \phi(\eta)  d\eta \\
& \ge F \cdot \int_{0}^{\phi^{-1}(v_n)} \phi(\eta)  d\eta \\
& \ge v_n,
\end{align*}
where the second inequality follows from the monotonicity of $\phi$ and the last inequality is as in Case 1.

Thus, assuming Eq.~\eqref{eq:kRental:FD:proof:theorem} holds, the dual constraint in Eq.~\eqref{eq:kRentalFD:dual:cons} is satisfied for each $n \in [N]$, proving the feasibility of the dual solution $\{u^{\alg}_n, \lambda^{\alg}_n\}$.

Let us now return to the proof of inequality Eq. \eqref{eq:kRental:FD:proof:theorem}. Define the set of requests $O_{n}$ as
\begin{align*}
   O_{n} = \{ n \leq j \leq N \mid a_{n} + d > a_j, \hat x_j > 0 \},
\end{align*}
i.e., the set of requests arriving after request $n$ whose rental request intervals overlap with that of request $n$. Furthermore, define the set of requests $B_n$ as
\begin{align*}
   B_n = \{ 1 \leq j \le n \mid a_{j} + d >a_n, \hat x_j > 0 \}.
\end{align*}
Note that for each request $j \in B_n$, we have $\nu^*_{j} \in O_n$. 
Then, it follows that:
\begin{align*}
     \sum_{j = n}^{N} \lambda^{\alg}_{j} \cdot \indicator{a_n + d > a_j}
     & \ge \sum_{j \in O_{n}} \lambda^{\alg}_{j} \\
     &= \sum_{j \in O_{n}} \sum_{m \in [n]} \Delta_{m}^{\lambda^{\alg}_{j}} \\
     &\ge \sum_{m \in B_n} \Delta_{m}^{\lambda^{\alg}_{\nu^*_m}},
\end{align*}
where $\Delta_{m}^{\lambda^{\alg}_{j}}$ denotes the increase in the value of $\lambda^{\alg}_{j}$ due to the dual update in Eq.~\eqref{eq:kRentalFD:dual:update2} performed for request $m$.
In the above inequalities, the first inequality follows from the definition of the set $O_n$ and the second inequality follows from the fact that the request $\nu^*_m$ is in the set $O_n$ for each request $m \in B_n$. 

Moving forward, let us sort the requests in $B_n$ in increasing order of their arrival times so that
\begin{align*}
   B_n = \{n_1, n_2, \dots, n_{|B_n|}\}.
\end{align*} 

We have:
\begin{align*}
\sum_{m \in B_n} \Delta_{m}^{\lambda^{\alg}_{\nu^*_m}} 
&\ge F  \sum_{i=1}^{|B_n|} \int_{\eta = y_{n_i}}^{y_{n_i}+ \hat x_{n_i}} \phi(\eta) d\eta \\
&\ge F \sum_{i=1}^{|B_n|} \int_{\eta = \sum_{j=1}^{i-1} x_{n_j}}^{\sum_{j=1}^{i} \hat  x_{n_j}} \phi(\eta) d\eta \\
&= F \int_{\eta = 0}^{\sum_{i=1}^{|B_n|} \hat x_{n_i}} \phi(\eta) d\eta \\
&=F  \int_{\eta = 0}^{y_{n}+ \hat x_{n}} \phi(\eta) d\eta.
\end{align*}
Here, the first inequality follows from the update for $\lambda^{\alg}_{\nu^*_m}$ in Eq.~\eqref{eq:kRentalFD:dual:update2} for each request $m$, and the second inequality holds because, by the definition of the set $B_n$, all requests in $B_n$ arriving prior to request $n_i$ have rental requests overlapping request $n_i$; hence, the variable $y_{n_i}$ is at least $\sum_{j=1}^{i-1} x_{n_{j}}$ at the arrival of request $n_i$.

Putting together the results obtained in the above two cases, the inequality in Eq.~\eqref{eq:kRental:FD:proof:theorem} follows.

\paragraph{Step-III: Proof of $D^{\alg} \leq F \cdot P^{\alg}$.}
Instead of directly proving the overall inequality $D^{\alg} \leq F \cdot P^{\alg}$, we show that for each request $n$:
\begin{align*}
\Delta^{D^{\alg}}_{n} \leq F  \Delta^{P^{\alg}}_{n},
\end{align*}
where $\Delta^{D^{\alg}}_{n}$ denotes the increase in the dual objective value after updating the dual variables for request $n$ via Eqs.~\eqref{eq:kRentalFD:dual:update1} and \eqref{eq:kRentalFD:dual:update2}, and $\Delta^{P^{\alg}}_{n}$ denotes the increase in the primal objective value after updating the primal LP solution for request $n$ by setting $x^{\alg}_{n} = \hat x_{n}$.

We now consider the following two cases in order to prove the above inequality.
\begin{itemize}
    \item \textbf{Case 1: $\hat x_{n} = 1$.}
In this case, we have
\begin{align*}
    \Delta^{D^{\alg}}_{n}  & = \Delta_n^{u^{\alg}_n} + k  \Delta_n^{\lambda^{\alg}_{\nu^*_n}} \\
    & = F \cdot \hat x_{n} \cdot \Bigl( v_n - \phi\Bigl(\frac{y_n+\hat x_{n}}{k}\Bigr)\Bigr) + k F \cdot \int_{\eta=\frac{y_n}{k}}^{\frac{y_n+\hat x_{n}}{k}} \phi(\eta)d\eta \\
    & = F \cdot \Bigl( v_n - \phi\Bigl(\frac{y_n+1}{k}\Bigr)\Bigr) + k F \cdot \int_{\eta=\frac{y_n}{k}}^{\frac{y_n+1}{k}} \phi(\eta)d\eta \\
    & \leq F \cdot \Bigl( v_n - \phi\Bigl(\frac{y_n+1}{k}\Bigr)\Bigr) + k F \cdot \frac{1}{k}\phi\Bigl(\frac{y_n+1}{k}\Bigr) \\
    & = F \cdot v_n = F  \Delta^{P^{\alg}}_{n},
\end{align*}
where the first equality follows from the objective value of the dual LP in Eq.~\eqref{lp:kRentalFD:dual}. The second equality follows from the dual updates done in Eq.~\eqref{eq:kRentalFD:dual:update1} and Eq.~\eqref{eq:kRentalFD:dual:update2}. In addition, the first inequality follows from the fact that $\phi$ is an increasing function.

\item \textbf{Case 2: $\hat x_{n} < 1$.} 
In this case, if $\hat x_{n} = 0$, then we have $\Delta^{D^{\alg}}_{n} = \Delta^{P^{\alg}}_{n} = 0$. Otherwise, if $\hat x_{n} \neq 0$, then by Eq.~\eqref{eq:kRental:FD:fraction:allocation} we have 
\begin{align*}
\hat x_{n} = \max\Bigl\{0, \min\Bigl\{1, k\cdot\phi^{-1}(v_n) - y_n\Bigr\}\Bigr\}.
\end{align*}
Since $\hat x_{n} < 1$, we must have
\begin{align*}
v_n = \phi\Bigl(\frac{y_n+\hat x_{n}}{k}\Bigr).
\end{align*}
It follows that
\begin{align*}
    \Delta^{D^{\alg}}_{n}  
    & = \Delta_n^{u^{\alg}_n} + k  \Delta_n^{\lambda^{\alg}_{\nu^*_n}} \\
    &= F \cdot \hat x_{n} \cdot \Bigl( v_n - \phi\Bigl(\frac{y_n+\hat x_{n}}{k}\Bigr)\Bigr) + k F \cdot \int_{\eta=\frac{y_n}{k}}^{\frac{y_n+\hat x_{n}}{k}} \phi(\eta)d\eta \\
    &= k F \cdot \int_{\eta=\frac{y_n}{k}}^{\frac{y_n+\hat x_{n}}{k}} \phi(\eta)d\eta, \\
    &\le F k \cdot \frac{\hat x_{n}}{k} \cdot \phi\Bigl(\frac{y_n+\hat x_{n}}{k}\Bigr) \\
    &= F \cdot v_n \cdot \hat x_{n} \\
    &= F \Delta^{P^{\alg}}_{n}.
\end{align*}
\end{itemize}
Here, the inequality follows from the fact that the function $\phi$ is increasing. Combining the results from both cases, we obtain 
\begin{align*}
\Delta^{D^{\alg}}_{n} \le F  \Delta^{P^{\alg}}_{n},
\end{align*}
which completes the proof of the third step and thus the proof for Theorem~\ref{them:kRentalFD}.

\section{Proof of Proposition~\ref{lemma:kRental:FD:lb}}
\label{apx:lemma:kRental:FD:lb}
Following the same proof structure as in \cite{sun2024static} for the online selection problem, we can prove the lower bound $1+\ln\bigl(\tfrac{v_{\max}}{v_{\min}}\bigr)$ on the competitiveness of every online algorithm for \problemkRentalFD.

We design a set of hard instances for the \problemkRentalFD problem similar to \cite{sun2024static}. Let $\cala(k,v)$ denote a batch of $k$ identical requests, each with valuation $v$ (with $v\in[v_{\min},v_{\max}]$). Divide the uncertainty range $[v_{\min},v_{\max}]$ into $m-1$ sub-ranges of equal length 
\begin{align*}
\Delta_v = \frac{v_{\max}-v_{\min}}{m-1}.
\end{align*}
Let 
\begin{align*}
\mathcal{V} := \{v_i\}_{i\in[m]},
\end{align*}
where $v_i = v_{\min} + (i-1)\Delta_v$ for $i \in [m]$. Define an instance 
\begin{align*}
I_{v_i} := \cala(k,v_1) \oplus \cala(k,v_2) \oplus \cdots \oplus \cala(k,v_i),
\end{align*}
which consists of a sequence of request batches with increasing valuations that arrive consecutively within an arbitrarily short time interval. (Here, $\cala(k,v_i) \oplus \cala(k,v_j)$ denotes a batch $\cala(k,v_i)$ immediately followed by a batch $\cala(k,v_j)$.) In this construction, all requests, including those from the first batch $\cala(k,v_1)$ to the last request in the final batch $\cala(k,v_m)$, arrive during the short interval $[0,\epsilon]$, where $\epsilon$ is a small value satisfying $\epsilon < d$. This setup guarantees that if a unit is allocated to a request in batch $\cala(k,v_i)$, it cannot be reallocated to any request arriving afterwards, because the unit will only become available after $d$ time steps, while subsequent requests arrive within a much shorter time span.

We consider the collection $\{I_{v_i}\}_{i\in[m]}$ as the set of hard instances for the \problemkRentalFD problem.
Following the same proof structure as in the proof of Lemma~2.3 in \cite{sun2024static}, to obtain the optimal online algorithm on the set of hard instances  $\{I_{v_i}\}_{i\in[m]}$, one can show that the lower bound $1+\ln\!\bigl(\tfrac{v_{\max}}{v_{\min}}\bigr)$ holds for the competitiveness of every online algorithm on the set of hard instances $\{I_{v_i}\}_{i\in[m]}$ and thus prove the lower bound $1+\ln\!\bigl(\tfrac{v_{\max}}{v_{\min}}\bigr)$.

\section{Impossibility Result for the Existence of a Lossless Rounding in the Variable Duration Setting}
\label{apx:proof:impos:rounding}

Consider the \gocr problem introduced in Section~\ref{sec:gocr}, where each player rents a ball for a fixed duration~$d$ that is identical across all players and independent of their identity. In this setting, the rental period is uniform and known in advance.

In the following, we study a variant of this problem in which the rental duration is player-dependent and may vary across players. Specifically, each player requests a rental duration that is revealed only upon their arrival. Thus, the algorithm receives information about a player's rental duration in an online fashion. A more formal definition of the \ocr problem under this variable-duration setting is as follows.

\begin{definition}[\gocrv]
Consider a set of $k$ identical balls, each uniquely labeled from the set $\{1, 2, \ldots, k\}$. Each ball can be assigned to a player for a variable duration, determined by the player’s requested rental period, after which the ball becomes available for reassignment. A sequence of $N$ players arrives one by one, with each player $n$ characterized by a triple $(a_n, \hat x_n, d_n)$, where $a_n$ is the arrival time, $\hat x_n \in [0,1]$ is the target probability with which a ball should be assigned to player $n$, and $d_n$ is the requested rental duration. For any $\gamma \in [0,1]$, a \gocrv\ algorithm, given the input sequence $\{(\hat x_n, a_n, d_n)\}_{n \in [N]}$, assigns a ball to each player $n \in [N]$ with probability at least $\gamma \hat x_n$. Similar to the definition of \gocr, for any given input instance of \gocrv, $\{(\hat x_n, a_n, d_n)\}_{n \in [N]}$, we impose the following condition:
\begin{align}\label{eq:ocrv_condition1}
\hat x_n \leq \min \!\left( 1, k - \sum_{j \in [n-1]} \hat x_j \cdot \mathbb{I}_{\{a_{j} + d_j > a_n\}} \right), \qquad \forall n \in [N].
\end{align}
\end{definition}

\paragraph{Proof of Theorem~\ref{prop:rounding:impossiblity}}
We proceed by contradiction. Suppose, for the sake of contradiction, that a (possibly randomized) online algorithm \alg is \emph{lossless} for the \gocrv problem; that is, it achieves $\gamma = 1$ on every instance.  Consider the following instance with inventory of $k=2$ balls and:
\begin{align*}
\mathcal{I} =
\Biggl\{ &
\underbrace{(a_1 = 1,    x_1 = 0.5,       d_1 = 5)}_{P_1},
\underbrace{(a_2 = 2,    x_2 = 0.5,       d_2 = 7)}_{P_2},
\underbrace{(a_3 = 5.5,  x_3 = \tfrac23, d_3 = 9)}_{P_3}, \\ &
\qquad \underbrace{(a_4 = 6,    x_4 = \tfrac13, d_4 = 8)}_{P_4},
\underbrace{(a_5 = 8,    x_5 = \tfrac12, d_5 = 10)}_{P_5},
\underbrace{(a_6 = 14,   x_6 = \tfrac56, d_6 = 10)}_{P_6}
\Biggr\}.
\end{align*}

The above instance clearly satisfies the constraint in Eq. \eqref{eq:ocrv_condition1}.
Let $A_n^{(i)}$ ($i = 1,2$) be the event that the $i$-th ball is allocated to player~$n$ by the online algorithm \alg\ upon processing the player’s request.  Likewise, let $E_t^{(i)}$ ($i = 1,2$) denote the event that the $i$-th ball is still available at time $t$.

Next, we prove the following claim:
    \begin{lemma}
    \label{lemma:player12}
        For the algorithm $\alg$ to be lossless, we must have:
        \begin{align*}
            \pr[A_{1}^{(1)} \cap A_{2}^{(2)}]\ + \  \pr[A_{1}^{(2)} \cap A_{2}^{(1)}] = 0
        \end{align*}
    \end{lemma}
    \begin{proof}
We prove the claim by contradiction.  Suppose the lemma fails, i.e.,
\[
\pr[A_{1}^{(1)}\cap A_{2}^{(2)}\bigr] 
+\pr[A_{1}^{(2)}\cap A_{2}^{(1)}\bigr] >0 .
\]
Then the probability that \emph{at least one} of the two balls is still
available immediately after the decision for player~2 satisfies
\begin{align*}
\pr[E_{2}^{(1)} \cup E_{2}^{(2)}\bigr]
  &= 1 - \pr[(A_{1}^{(1)} \cup A_{2}^{(1)}) \cap (A_{1}^{(2)} \cup A_{2}^{(2)})\bigr] \\
  &= 1 - \pr[(A_{1}^{(1)}\cap A_{1}^{(2)}) \cup (A_{2}^{(1)}\cap A_{2}^{(2)})\bigr] \\
  &= 1 - \pr[A_{1}^{(1)}\cap A_{2}^{(2)}\bigr]
      - \pr[A_{1}^{(2)}\cap A_{2}^{(1)}\bigr]  < 1,
\end{align*}
where the strict inequality follows from the contradictory assumption.

\smallskip
\noindent
Next, let us create a new instance~$\mathcal I'$ that is identical to $\mathcal I$ for
players~$P_1$ and $P_2$, but replaces the remaining sequence by a single
player~$P_3'$ with
\[
(a_3 = 3, x_3 = 1, d_3 = 10).
\]
The feasibility condition~\eqref{eq:ocrv_condition1} is easily verified for $\mathcal I'$.

Because the first two players in $\mathcal I$ and $\mathcal I'$ are identical,
algorithm~\alg\ follows the same random path on both instances for the
first two players.  Thus, at $a_3=3$ in $\mathcal I'$ we still have
$
\pr[E_{2}^{(1)} \cup E_{2}^{(2)}\bigr] < 1,
$
meaning that with probability less than one either of balls are available at arrival of request $P_3'$.  Yet a
lossless algorithm must allocate a ball to $P_3'$ with probability~$x_3
= 1$, an impossibility.  Hence our initial assumption is false, and the lemma follows.
\end{proof}

Next, consider the arrival of the third player $P_{3}$ in instance~$\mathcal I$,
with parameters $a_{3}=5$, $x_{3}= \tfrac23$, and $d_{3}=9$.
We prove the following statement.

\begin{lemma}\label{lem:P3prob}
For \alg\ to remain lossless on~$\mathcal I$, it must hold that
\[
\pr[A^{(1)}_{2}\cap A^{(2)}_{3}\bigr] +
\pr[A^{(2)}_{2}\cap A^{(1)}_{3}\bigr]
=
\frac16.
\]
\end{lemma}
\begin{proof}
First, let us prove that
\begin{align}
\label{eq:varaible-immposibility-proof-temp1}
\pr[A^{(1)}_{2} \cap A^{(2)}_{3}] + \pr[A^{(2)}_{2} \cap A^{(1)}_{3}]
\le \frac{1}{6}.
\end{align}
The proof is again by contradiction.  Assume that the above inequality does not hold.

Consider the instance $\mathcal{I}'$, in which the first three players are identical to those of instance $\mathcal{I}$, but the fourth
player $P'_{4}$ is characterised as
$
(a_{4}=7, d_{4}=10, x_{4}= \tfrac56).
$
It can be verified that the feasibility constraint in Eq.~\eqref{eq:ocr_conditions} is satisfied for every player. Because the first three arrivals in $\mathcal{I}$ and $\mathcal{I}'$ are identical, the online algorithm $\alg$ behaves the same on both instances up to arrival of $P'_{4}$. Consequently, when $P'_{4}$ arrives in $\mathcal{I}'$,
\begin{align*}
\pr[E^{(1)}_{7} \cup E^{(2)}_{7}]
  & = 1 - \pr[(A^{(1)}_{2} \cup A^{(1)}_{3}) \cap (A^{(2)}_{2} \cup A^{(2)}_{3})] \\[4pt]
  &= 1 - \pr[(A^{(1)}_{2} \cap A^{(2)}_{3}) \cup (A^{(1)}_{3} \cap A^{(2)}_{2})] \\[4pt]
  &= 1 - \pr[A^{(1)}_{2} \cap A^{(2)}_{3}] - \pr[A^{(2)}_{2} \cap A^{(1)}_{3}] < \frac{5}{6},
\end{align*}
where the last strict inequality is by the contradictory assumption to have $\pr[A^{(1)}_{2} \cap A^{(2)}_{3}] + \pr[A^{(2)}_{2} \cap A^{(1)}_{3}] > \frac{1}{6}$.
The allocation of $P_{1}$ does not influence the calculation above, because by $a_{4}=7$ any ball assigned to $P_{1}$ has already returned to the system.  Hence, at the arrival of player~$P'_4$ in $\mathcal{I}'$ the probability that either ball is still available,
$\pr[E^{(1)}_{7} \cup E^{(2)}_{7}]$, is strictly less than $\tfrac56$.  Therefore, \alg cannot allocate a ball to $P_{4}$ with the required target probability $\tfrac56$, contradicting the lossless property.  Thus we must have
\[
\pr[A^{(1)}_{2} \cap A^{(2)}_{3}] + \pr[A^{(2)}_{2} \cap A^{(1)}_{3}]
\le \frac{1}{6}.
\]

Next, let us prove that
\[
\pr[A^{(1)}_2 \cap A^{(2)}_3] + \pr[A^{(2)}_2 \cap A^{(1)}_3] \ge \frac{1}{6},
\]
after which we can conclude that the sum equals $\tfrac16$.

First, we compute the following probability:
\begin{align*}
\pr[(A^{(1)}_1 \cup A^{(2)}_1) \cup (A^{(1)}_2 \cup A^{(2)}_2)]
  &= \pr[(A^{(1)}_1 \cup A^{(2)}_1)] + \pr[(A^{(1)}_2 \cup A^{(2)}_2)] \\[-2pt]
  &\quad- \pr[(A^{(1)}_1 \cup A^{(2)}_1) \cap (A^{(1)}_2 \cup A^{(2)}_2)] \\[4pt]
  &= x_1 + x_2 - \pr[(A^{(1)}_1 \cap A^{(2)}_2) \cup (A^{(2)}_1 \cap A^{(1)}_2)] \\[4pt]
  &= 1 - \pr[(A^{(1)}_1 \cap A^{(2)}_2)] - \pr[(A^{(2)}_1 \cap A^{(1)}_2)] \\[4pt]
  &= 1,
\end{align*}
where the last equality follows from Lemma~\ref{lemma:player12}.  
By the same reasoning we have
\[
\pr[(A^{(1)}_1 \cup A^{(2)}_1) \cap (A^{(1)}_2 \cup A^{(2)}_2)] = 0 .
\]

Since $\pr[A^{(1)}_3 \cup A^{(2)}_3] = x_3 = \tfrac23$,  
$\pr[(A^{(1)}_1 \cup A^{(2)}_1)] = \tfrac12$,  
and
$\pr[(A^{(1)}_1 \cup A^{(2)}_1) \cup (A^{(1)}_2 \cup A^{(2)}_2)] = 1,$
we must have
\[
\pr[(A^{(1)}_3 \cup A^{(2)}_3) \cap (A^{(1)}_2 \cup A^{(2)}_2)] \ge \frac{1}{6}.
\]
Combining this with the earlier inequality establishes
\[
\pr[A^{(1)}_2 \cap A^{(2)}_3] + \pr[A^{(2)}_2 \cap A^{(1)}_3] = \frac{1}{6}.
\]
Thus, the proof is completed. 
\end{proof}

Next, consider the arrival of player $P_{4}$ in instance $\mathcal I$,
with parameters $a_{4}=6$, $x_{4}= \tfrac13$, and $d_{4}=7$.
\begin{lemma}
\label{lem:p5}
For \alg to be lossless, after the arrival of player $P_4$ we must have
\[
\pr[A^{(1)}_3 \cap A^{(2)}_4] + \pr[A^{(2)}_3 \cap A^{(1)}_4] = 0 .
\]
\end{lemma}

\begin{proof}
In case the above equality is not satisfied, a request $P_5$ could arrive at time $a_5 = 10$ with target probability $x_5 = 1$, and $\alg$ would not be lossless since $\pr[E^{(1)}_{10} \cup E^{(2)}_{10}] < 1$.
\end{proof}

Next, consider the arrival of player $P_5$ at time $a_5 = 8$, with target probability $x_5 = \tfrac12$ and $d_5 = 10$. Let us refer to the time right before the arrival of fifth player as $8^{-}$.

\begin{lemma}
\label{lem:p6}
After the arrival of player $P_5$,
\[
\pr[A^{(1)}_3 \cap A^{(2)}_5] + \pr[A^{(2)}_3 \cap A^{(1)}_5] \ge \frac12 .
\]
\end{lemma}

\begin{proof}
From the previous two lemmas \ref{lem:P3prob} and \ref{lem:p5}, we have
\[
\pr[(A^{(1)}_2 \cup A^{(2)}_2)\cap(A^{(1)}_3 \cup A^{(2)}_3)] = \frac16,
\qquad
\pr[(A^{(1)}_3 \cup A^{(2)}_3)\cap(A^{(2)}_4 \cup A^{(1)}_4)] = 0.
\]

Let us define $ C_1, C_2 $, and $ C_3 $ as follows:
\begin{align*}
    & C_1 = \left(A^{(1)}_3 \cup A^{(2)}_3 \right) 
\cap \left(A^{(2)}_2 \cup A^{(1)}_2 \right)',\\
    & C_2 = \left(A^{(2)}_2 \cup A^{(1)}_2\right) \cap \left(A^{(1)}_3 \cup A^{(2)}_3 \right)',\\
    & C_3 = \left(A^{(1)}_3 \cup A^{(2)}_3\right) \cap \left(A^{(2)}_2 \cup A^{(1)}_2 \right).
\end{align*}
These events are pairwise disjoint and  
$\pr[C_1 \cup C_2 \cup C_3] = 1$.

\begin{claim}
$ \pr[E_{8^{-}}^{(1)} \cup E_{8^{-}}^{(2)} \mid C_2] = 0  $.
\end{claim}

\begin{proof}
From Lemma~\ref{lem:p5} we have
\[
\pr[(A^{(1)}_3 \cup A^{(2)}_3)\cap(A^{(2)}_4 \cup A^{(1)}_4)] = 0,
\] 
thus it follows that $\pr[C_1 \cap (A^{(2)}_4 \cup A^{(1)}_4)] = 0$ and  
$\pr[C_3 \cap (A^{(2)}_4 \cup A^{(1)}_4)] = 0$.  
Note that we have $\pr[A^{(2)}_4 \cup A^{(1)}_4] = \pr[C_2] = \tfrac13$, so
$A^{(2)}_4 \cup A^{(1)}_4 = C_2$. Next,
\begin{align*}
C_2 &= C_2 \cap C_2 \\[2pt]
    &= \bigl((A^{(2)}_4 \cup A^{(1)}_4)\cap( (A^{(2)}_2 \cup A^{(1)}_2) \cap (A^{(1)}_3 \cup A^{(2)}_3)')\bigr) \\[2pt]
    &= \bigl((A^{(2)}_4 \cup A^{(1)}_4)\cap(A^{(2)}_2 \cup A^{(1)}_2)\bigr)
       \cap\bigl((A^{(2)}_4 \cup A^{(1)}_4)\cap(A^{(1)}_3 \cup A^{(2)}_3)'\bigr) \\[2pt]
  &  = \bigl((A^{(2)}_4 \cup A^{(1)}_4)\cap(A^{(2)}_2 \cup A^{(1)}_2)\bigr)
       \cap(A^{(2)}_4 \cup A^{(1)}_4) \\
       & = \bigl((A^{(2)}_4 \cup A^{(1)}_4)\cap(A^{(2)}_2 \cup A^{(1)}_2)\bigr) \\
    &= (A^{(2)}_2 \cap A^{(1)}_4) \cup (A^{(1)}_2 \cap A^{(2)}_4),
\end{align*}
where the third equality uses $\pr[(A^{(1)}_3 \cup A^{(2)}_3) \cap (A^{(2)}_4 \cup A^{(1)}_4)] = 0$.
Based on the above equations, under the event $C_2$, either $A^{(2)}_2 \cap A^{(1)}_4$ or $A^{(1)}_2 \cap A^{(2)}_4$ occurs, so neither ball~1 nor ball~2 is available. Therefore
\[
\pr[E_{8^-}^{(1)} \cup E_{8^-}^{(2)} \mid C_2] = 0.
\]
This completes the proof. 
\end{proof}

Next, we can verify that
\[
\pr[E_{8^-}^{(1)} \cup E_{8^-}^{(2)} \mid C_3] = 0.
\]
In fact, $C_3 = (A^{(1)}_3 \cup A^{(2)}_3) \cap (A^{(2)}_4 \cup A^{(1)}_4)
         = (A^{(1)}_3 \cap A^{(2)}_4) \cup (A^{(2)}_3 \cap A^{(1)}_4)$.
Thus, under $C_3$, either $(A^{(1)}_3 \cap A^{(2)}_4)$ or $(A^{(2)}_3 \cap A^{(1)}_4)$ occurs, so neither ball~1 nor ball~2 is available, and
$\pr[E_4^{(1)} \cup E_4^{(2)} \mid C_3] = 0$.

Since at the arrival of $ P_5 $ we have
$
\sum_{i \in [4]} x_i \cdot \mathbb{I}_{\{a_i + d_i \ge a_5\}} = \frac32,
$
with probability $\frac12$ one of the two balls is available, and
\begin{align*}
\frac12 &= \pr[E_{8^-}^{(1)} \cup E_{8^-}^{(2)}] \\
   & = \pr[E_{8^-}^{(1)} \cup E_{8^-}^{(2)} \mid C_1]
    + \pr[E_{8^-}^{(1)} \cup E_{8^-}^{(2)} \mid C_2]
    + \pr[E_{8^-}^{(1)} \cup E_{8^-}^{(2)} \mid C_3]
 \\ & = \pr[E_{8^-}^{(1)} \cup E_{8^-}^{(2)} \mid C_1]
  = \pr[C_1].
\end{align*}
Thus, we have $E_{8^-}^{(1)} \cup E_{8^-}^{(2)} = C_1$.

At the arrival of $ P_5 $, if either ball is available, the other ball is already allocated to $ P_3 $, since
$C_1 = (A^{(1)}_3 \cup A^{(2)}_3) \cap (A^{(2)}_2 \cup A^{(1)}_2)' $.
Therefore, a lossless \alg\ must allocate one of the two balls to $ P_5 $ with probability $\tfrac12$, and
\begin{align*}
    & \pr[(A^{(1)}_5 \cup A^{(2)}_5) \cap (A^{(1)}_3 \cup A^{(2)}_3)] \\
=\ & \pr[(A^{(1)}_3 \cap A^{(2)}_5) \cup (A^{(2)}_3 \cup A^{(1)}_5)] \\
=\ & \pr[A^{(1)}_3 \cap A^{(2)}_5]
  + \pr[A^{(2)}_3 \cup A^{(1)}_5] \\
=\ & \frac12.
\end{align*}
This completes the proof. 
\end{proof}

Now consider the arrival of the sixth player with $(a_6 = 13.5, x_6 = \tfrac56, d_6 = 10)$.  It can be verified that
\begin{align*}
\pr[E_{6}^{(1)} \cup E_{6}^{(2)}]
  &= 1 - \pr[(A_{3}^{(1)} \cup A_{5}^{(1)}) \cap (A_{3}^{(2)} \cup A_{5}^{(2)})] \\[2pt]
  &= 1 - \pr[(A_{3}^{(1)} \cap A_{5}^{(2)}) \cup (A_{3}^{(2)} \cap A_{5}^{(1)})] \\[2pt]
  &= \frac12,
\end{align*}
where the last equality follows from Lemma~\ref{lem:p6}.  Thus, at the arrival of player $ P_6 $, with probability equal to $\tfrac12$ either ball is available, and the algorithm cannot allocate a ball with target probability $\tfrac56$ losslessly to player $ P_6 $. We thus complete the proof.

\section{Proof of Theorem \ref{prop:kRentalVD-Dynamic-CR}}
\label{apx:prop:kRental:VD:CR}
To establish the $\alpha$-competitiveness of Algorithm~\ref{alg:kRentalVD-dynamic}, we apply the LP-free certificate approach developed \cite{goyal2020}. We define a system of linear constraints that acts as a certificate of $\alpha$-competitiveness when feasible. 
The variables in this system include $\{ \{u_n, \lambda_n^{(i)}\}_{n \in [N]}, \theta_i\}_{i \in [k]}$. Conceptually, we employ a pricing scheme inspired by the economic interpretation of the randomized primal-dual framework in~\cite{eden2021economics}. In this setting, $u_n$ denotes the utility of request $n$, while $\lambda_n^{(i)}$ represents the price of unit $i$ at the time of request $n$'s arrival. The variable $\theta_i$ corresponds to the revenue obtained by the algorithm from selling the $i$-th unit of the resource.
 The constraints are given by:
\begin{align}
\label{eq:kRentalVD:cons:1}
& \sum_{n \in [N]} u_n + \sum_{i \in [k]} \theta_i \le \alpha \cdot \alg(I),\\ 
\label{eq:kRentalVD:cons:2}
& \theta_i + \sum_{n \in P_i} u_n \ge \sum_{n \in P_i} d_n = \opt_i,
\qquad \forall i \in [k],
\end{align}
where $\alg(I)$ denotes the expected performance of the algorithm on instance~$I$. For each $i$, the set $P_i$ consists of the requests to whom the optimal clairvoyant algorithm allocates the $i$-th unit of the resource, and $\opt_i = \sum_{n \in P_i} d_n$.
Clearly, if the above constraints hold, then:
\begin{align*}
   \opt & = \sum_{i=1}^{k} \opt_i  \\
   & \le \sum_{i=1}^{k} \Bigl(\theta_i + \sum_{n \in P_i} u_n \Bigr) \\
   & \le \sum_{i \in [k]} \theta_i + \sum_{n \in [N]} u_n \\
   & \le \alpha \cdot \alg(I),
\end{align*}
where in the derivation above, the first inequality follows from Eq.~\eqref{eq:kRentalVD:cons:1}, 
the second inequality holds because the sets $\{P_i\}$ are pairwise disjoint, 
and the final inequality follows from Eq.~\eqref{eq:kRentalVD:cons:2}. 
Hence, if there is a solution to this linear system, 
it certifies the $\alpha$-competitiveness of Algorithm~\ref{alg:kRentalVD-dynamic}. 
In what follows, we first describe how to assign values to the system’s variables, 
and then show that these assignments satisfy all the inequalities of the system.

Based on the performance of Algorithm~\ref{alg:kRentalVD-dynamic}, we now specify how to assign values to the variables in our LP-free certificate. Initially, all variables are set to zero. After the arrival of the final request in instance~$I$, we update the variables as follows. For each rental request~$n$, let $\hat x_n$ be the fractional allocation variable determined by Algorithm~\ref{alg:kRentalVD-dynamic-fracitonal}. We then perform the following updates:
\begin{align}
\label{eq:kRentalVD:dual:update:1}
u_n &= \frac{\alpha}{3} d_n \hat x_n,\\ \label{eq:kRentalVD:dual:update:2}
\lambda_{j}^{(i^*_n)} &= \lambda_{j}^{(i^*_n)} + \frac{\alpha}{3} \bigl(a_{j+1} - a_{j}\bigr) \hat x_n, \quad \forall j, n \le j < \nu^*_n, a_{j} < a_n + d_n,\\
\label{eq:kRentalVD:dual:update:3}
\lambda_{\nu^*_n}^{(i^*_n)} & = \lambda_{\nu^*_n}^{(i^*_n)} + \frac{\alpha}{3} \Bigl(2 d_n - \bigl(a_{\nu^*_n} - a_n\bigr)\Bigr) \hat x_n,
\end{align}
where $\nu^*_n = \max \{j \ge n | a_{j} < a_n + d_n\}$. 
After updating the variables above for all the $N$ requests, let us set the value of variables $\theta_{i} = \sum_{n \in [N]} \lambda_n^{(i^*_n)}$ for all $i \in [k]$. 

\begin{remark}
The intuition behind the design of the variables in the system is as follows. We implement a pricing scheme in which the variable $u_n$ denotes the utility accrued by request $n$, and the set of variables $\lambda^{(i)}_j$ represents the price of unit $i$ at the arrival time of request $j$. Depending on the fractional allocation $\hat x_n$, which is allocated from unit $i^*_n$ to request $n$, the utility $u_n$ increases accordingly. Furthermore, we raise the price of unit $i^*_n$ at the arrival time of every request $j < \nu^*_n$ whose request interval overlaps with that of request $n$. This price adjustment is intended for any request arriving after request $n$ whose rental duration finishes before that of request $n$ and who might be rejected because a fraction of unit $i^*_n$ has already been allocated to request $n$. In addition, we also increase the price of unit $i^*_n$ for request $\nu^*_n$, which is the last arriving request in the horizon whose request interval overlaps with that of request $n$. This adjustment accounts for requests who might arrive very close to the end of request $n$'s request interval and whose requests may be rejected because the algorithm has already allocated a fraction of unit $i^*_n$ to request $n$.
\end{remark}

\paragraph{First constraint of the system.}
Based on the construction of the variables detailed above, we now verify the  constraint Eq.~\eqref{eq:kRentalVD:cons:1}.

Let $\Delta_{n}^{\alg}$ denote the increase in the expected objective value of Algorithm~\ref{alg:kRentalVD-dynamic} after processing request $n$. We show that $\Delta_{n}^{\alg} = d_n \cdot \hat x_n$. Algorithm~\ref{alg:kRentalVD-dynamic} first computes $y_{n}^{(i)}$, the probabilistic utilization level of item $i$, which sums the probabilities that unit $i$ has been allocated to previous requests. Consequently, at the arrival of request $n$, resource $i$ is available with probability $1 - y_{n}^{(i)}$. Therefore,
\begin{align*}
   \Delta_{n}^{\alg}
   &= d_{n} \cdot \mathbb{P}[\text{unit }i_n^*\text{ is available}] \cdot \mathbb{P}[\text{the unit is allocated}] \\
   &= d_n  \bigl(1 - y_{n}^{(i_n^*)}\bigr) \cdot \frac{ \hat x_n}{1 - y_{n}^{(i_n^*)}} \\
   &= d_n \hat x_n.
\end{align*}
Summing over all requests $n \in [N]$, $\alg(I) = \sum_{n \in [N]} \Delta_{n}^{\alg} = \sum_{n \in [N]} d_n \hat x_n$.

Next, let $\Delta_{n}^{\textsc{rhs}}$ denote the increase in the right-hand side of Eq.~\eqref{eq:kRentalVD:cons:1} after updating the variables of the system according to Eqs.~\eqref{eq:kRentalVD:dual:update:1}--\eqref{eq:kRentalVD:dual:update:3} for the $n$-th rental request. It follows directly that
\begin{align*}
   \Delta_{n}^{\textsc{rhs}} = \alpha  d_n \hat x_n.
\end{align*}
Summing these increments for all requests $n \in [N]$ gives
\begin{align*}
   \sum_{n \in [N]} u_n + \sum_{i \in [k]} \theta_{i} = \sum_{n \in [N]} \Delta_{n}^{\textsc{rhs}} = \sum_{n \in [N]} \alpha \cdot d_n \cdot \hat x_n.
\end{align*}
Combining this with our previous result on the expected increase in the objective of the algorithm, we conclude that the first constraint, Eq.~\eqref{eq:kRentalVD:cons:1}, is satisfied.

\paragraph{Second set of constraints.}
The remaining set of constraints are as follows:
\begin{align*}
    \theta_{i} + \sum_{n \in P_{i}} u_n \ge \alpha \cdot \sum_{n \in P_i}d_n = \alpha \cdot \opt_{i},
\end{align*}
where the value of $\theta_{i}$ can be lower bounded by:
\begin{align*}
    \theta_{i}  = \sum_{n \in [N]} \lambda_n^{(i)} \ge 
    \sum_{n \in P_{i}} \sum_{j = n}^{N}  \lambda_{j}^{(i)} \indicator{a_n + d_n > a_{j}}.
\end{align*}
The inequality holds because, for each request $n \in P_i$, the set 
$\{j \in [N] \mid a_{n} + d_n > a_{j}\}$ cannot intersect with any other request’s set.
Otherwise, the optimal offline algorithm would have allocated unit $i$ to two requests at the same time, which is infeasible.
Hence, for the left-hand side of Eq.~\eqref{eq:kRentalVD:cons:2}, we have
\begin{align*}
    \theta_i + \sum_{n \in P_i} u_n  \ge  \sum_{n \in P_i} 
    \sum_{j = n}^{N}  \lambda_{j}^{(i)} \indicator{a_n + d_n > a_{j}} + \sum_{n \in P_i} u_n \ge \sum_{n \in P_i}  \Bigl(\sum_{j = n}^{N}  \lambda_{j}^{(i)} \indicator{a_n + d_n > a_{j}} + u_n \Bigr).
\end{align*}
Thus, it suffices to show that for each request $n \in P_i$,
\begin{align*}
   \sum_{j = n}^{N} \lambda_{j}^{(i)} \indicator{a_n + d_n > a_{j}} + u_n \ge \alpha \cdot d_n.
\end{align*}
In what follows, we first derive a lower bound on the left-hand side of the above inequality by separately bounding
$\sum_{j = n}^{N} \lambda_{j}^{(i)} \cdot \indicator{a_n + d_n > a_{j}}$ and $u_n$. Then, invoking the constraints placed on the price function $\phi$ in Theorem~\ref{prop:kRentalVD-Dynamic-CR}, we conclude that this bound is at least $\alpha \cdot d_n$, thus ensuring that Eq.~\eqref{eq:kRentalVD:cons:2} is satisfied.

Let us define
\begin{align*}
    \Lambda_{n}^{(i)} = \sum_{j = n}^{N} \lambda_{j}^{(i)} \indicator{a_n + d_n > a_{j}}.
\end{align*}

Let us further define the set of requests $B_n^{(i)}$ by
\begin{align}
\label{subsec:proof:kRental:VD:lambada:term}
    B_n^{(i)}  = \bigl\{ 1 \leq j < n \mid a_{j} + d_{j} > a_n, \hat x_j  > 0, i^*_{j} = i \bigr\}.
\end{align}
In other words, $B_n^{(i)}$ is the set of all requests who arrive before request~$n$, request rental intervals overlapping the interval of request~$n$, and have a non-zero probability of receiving unit~$i$.

Furthermore, define $C_n^{(i)} \subseteq B_n^{(i)}$ by
\begin{align*}
    C_n^{(i)} = \bigl\{ j \in B_n^{(i)} \mid a_{j} + d_{j} < a_n + d_n  \bigr\}.
\end{align*}
Next, let the set of requests $\{c_l\}_{l \in [L]}$ be defined recursively such that
\begin{align*}
    c_1  & =  \argmin_{j \in C_n^{(i)}}\{a_{j} + d_{j}\},\\
    c_{l} &= \argmin_{\substack{j \in C_n^{(i)} \\ a_{j} > a_{c_{l-1}}}} \{a_{j} + d_{j}\}, \quad\quad 1 < l \le L.
\end{align*}
The request $c_L \in C_n^{(i)}$ is the one for which
$\{j \in C_n^{(i)} \mid a_{j} > a_{c_L}\}$ is empty.
Using the sequence of requests $\{c_l\}_{l \in [L]}$, we partition $C_n^{(i)}$ into $L$ sets:
\begin{align*}
    \mathcal{C}_{l} =  \bigl\{ j \in C_n^{(i)} \mid a_{c_{l-1}} \le a_{j} < a_{c_{l}} \bigr\}, \quad \forall l \in [L].
\end{align*}

Furthermore, define the set of values $\{z_l\}_{l \in [L]}$ by
\begin{align*}
    z_l  = \sum_{\substack{j \in C_n^{(i)} \\ a_{j} < a_{c_{l}}}} \hat  x_j.
\end{align*}

Let $\Delta_{\Lambda_{n}^{(i)}}^{(j)}$ denote the increase in $\Lambda_{n}^{(i)}$ after processing the rental request of request~$j$ and updating the variables according to Eqs.~\eqref{eq:kRentalVD:dual:update:1}--\eqref{eq:kRentalVD:dual:update:3}. From the way we perform the updates, it follows that
\begin{align*}
    \sum_{j \in C_n^{(i)}} \Delta_{\Lambda_{n}^{(i)}}^{(j)} 
    &\ge
    \sum_{l=1}^{L} 
    \sum_{j \in \mathcal{C}_l} 
    \frac{\alpha}{3}\hat x_j
    \Bigl( 
       d_{j} 
       + 
       \bigl(a_{j} + d_{j} - a_{n}\bigr) 
    \Bigr)
    \\[6pt]
    &\ge
    \sum_{l=1}^{L} 
    \sum_{j \in \mathcal{C}_l} 
    \frac{\alpha}{3}\hat  x_j
    \Bigl(
       d_{c_l} 
       + 
       \bigl(a_{c_l} + d_{c_l} - a_{n}\bigr)
    \Bigr)
    \\[6pt]
    &\ge
    \sum_{l=1}^{L} 
    \frac{\alpha}{3}\bigl(z_l - z_{l-1}\bigr)
    \Bigl(
       2 
       d_{c_l} 
       + 
       a_{c_l} 
       - 
       a_{c_1} 
    \Bigr),
\end{align*}
where the first inequality follows from the dual updates in 
Eqs.~\eqref{eq:kRentalVD:dual:update:2}--\eqref{eq:kRentalVD:dual:update:3}, 
the second inequality relies on the definition of $c_l$ and $\mathcal{C}_l$ (under which $d_{j} \le d_{c_l}$ and $a_{j} + d_{j} \ge a_{c_l} + d_{c_l}$ for each $j \in \mathcal{C}_l$), 
and the last inequality follows from the definition of the values $\{z_l\}_{l \in [L]}$, taking $z_0 = 0$.

Define the sequence $\{z'_l\}_{l \in [L]}$ recursively as follows:
\begin{align*}
    & z'_L = z_L,\\[6pt]
    & z'_l =
    \begin{cases}
        z_l, 
         \hspace{+2cm} \text{if } 
          \bigl(d_{c_{l+1}} + a_{c_{l+1}}\bigr) - \bigl(d_{c_l} + a_{c_l}\bigr) \ge \phi\bigl(z'_{l+1}\bigr) - \phi\bigl(z_{l}\bigr),
        \\[6pt]
        \phi^{*}\!\Bigl(
          \phi\bigl(z'_{l+1}\bigr) - \bigl[
            \bigl(d_{c_{l+1}} + a_{c_{l+1}}\bigr) - \bigl(d_{c_l} + a_{c_l}\bigr)
          \bigr] \Bigr),
        \quad \text{otherwise},
    \end{cases}
    \quad \forall l \in [L-1].
\end{align*}
Then, we have the following lemma.

\begin{lemma}
\label{lem:kRental:VD:z':frac}
From the definition of the sequence $\{z'_l\}_{l \in [L]}$, for each $l \in \{1, \dots, L\}$, we have
\begin{align*}
    d_{c_l} \ge \phi\bigl(z'_l\bigr).
\end{align*}

Moreover, 
\begin{align}
    & \sum_{l=1}^{L} \frac{\alpha}{3}\bigl(z_l - z_{l-1}\bigr) \Bigl( 2 d_{c_l} + a_{c_l} - a_{c_1} - d_{c_1} \Bigr) \nonumber \\
\ge \ & \frac{\alpha}{3} \cdot z'_1 \phi(z'_1) + \sum_{l=2}^{L}   
    \bigl(z'_{l}- z'_{l-1}\bigr)  \bigl(d_{c_{l}} + a_{c_l} + d_{c_l} - a_{c_1} - d_{c_1} \bigr). \label{eq:kRental:VD:random:inequality}
\end{align}
\end{lemma}

\begin{proof}
We prove the two parts of this lemma separately.

\textbf{Part-I (Proof of $d_{c_l} \ge \phi(z'_l)$).} We show $d_{c_l} \ge \phi\bigl(z'_l\bigr)$ by induction on $L$. For the base case $L = 1$:
\begin{align*}
   z'_1 = z_1,
\end{align*}
and by the definition of $z_1$, the probabilistic utilization level of item $i$ at the arrival of request $c_1$ is at least $z_1$. Since $x_{c_1} \neq 0$, it follows that
\begin{align*}
   d_{c_1} \ge \phi\bigl(z_1\bigr) 
   = \phi\bigl(z'_1\bigr).
\end{align*}

Assume the statement holds for any number of requests in the set $\{c_{l}\}_{l \in [L]}$ up to $M-1$. We prove it for $L = M$. By the same argument as in the base case, we have 
\begin{align*}
   d_{c_M} \ge \phi\bigl(z_M\bigr) 
   = \phi\bigl(z'_M\bigr).
\end{align*}
If $z'_{M-1} = z_{M-1}$, then by the definition of $z_{M-1}$, we immediately get
\begin{align*}
   \phi\bigl(z'_{M-1}\bigr) \le d_{c_{M-1}}.
\end{align*}
Otherwise, suppose
\begin{align*}
   \bigl(d_{c_{M}} + a_{c_{M}}\bigr) 
   - 
   \bigl(d_{c_{M-1}} + a_{c_{M-1}}\bigr) 
   <
   \phi\bigl(z'_{M}\bigr) 
   -  
   \phi\bigl(z_{M-1}\bigr).
\end{align*}
From the definition of $z'_{M-1}$, we have
\begin{align*}
   \phi\bigl(z'_{M-1}\bigr) 
   = 
   \phi\bigl(z'_M\bigr) 
   - 
   \Bigl[
      \bigl(d_{c_{M}} + a_{c_{M}}\bigr) 
      - 
      \bigl(d_{c_{M-1}} + a_{c_{M-1}}\bigr)
   \Bigr].
\end{align*}
Since $\phi\bigl(z'_M\bigr) \le d_{c_M}$ and $a_{c_{M-1}} < a_{c_{M}}$ (by the definition of $c_l$), it follows that
\begin{align*}
   \phi\bigl(z'_{M-1}\bigr)
   \le
   \bigl[\phi(z'_M) - d_{c_M}\bigr]
   +
   \bigl[a_{c_{M-1}} - a_{c_M}\bigr]
   +
   d_{c_{M-1}}
   \le
   d_{c_{M-1}}.
\end{align*}
For the remaining inequalities, for each $l \in \{1,\dots,M-2\}$, the induction hypothesis applies to the set of values 
$\{z_1,\dots,z_{M-2}\} \cup \{z'_{M-1}\}$, ensuring that 
$\phi\bigl(z'_l\bigr) \le d_{c_l}$ for all $l \in \{1,\dots,M-1\}.$

\medskip

\textbf{Part-II (Proof of Eq.\eqref{eq:kRental:VD:random:inequality})}. We again use induction on $L$ to prove
\begin{align*}
   & \sum_{l=1}^{L} 
   \frac{\alpha}{3}\bigl(z_l - z_{l-1}\bigr)
   \Bigl(
      2d_{c_l} + a_{c_l} - a_{c_1} - d_{c_1}
   \Bigr) \\
   \ge\ & 
   \frac{\alpha}{3} z'_1 \phi(z'_1) 
   + 
   \sum_{l=2}^{L}   
   \bigl(z'_{l}- z'_{l-1}\bigr) 
   \Bigl(
       d_{c_{l}} + a_{c_l}+d_{c_l} - a_{c_1} - d_{c_1}
   \Bigr).
\end{align*}
For $L=1$, the statement is trivial since $z_1 = z'_1$. Assume it holds for all sequences $\{z'_l\}$ of length $L-1$. We prove it for $L$.  

From the definition of $z'_{L-1}$, we have  $z'_{L-1} \ge z_{L-1}$. Thus, we will have the following two cases:

\textit{Case 1:} $z'_{L-1} > z_{L-1}$.  
 By the induction hypothesis, for the sequence 
$\{z_1,\dots,z_{L-2}\} \cup \{z'_{L-1}\}$, we have
\begin{align*}
    & \sum_{l=1}^{L-2} 
    \frac{\alpha}{3}\bigl(z_l - z_{l-1}\bigr)
    \Bigl(
       2d_{c_l} + a_{c_l} - a_{c_1}
    \Bigr)
    +
    \frac{\alpha}{3}\bigl(z'_{L-1} - z_{L-2}\bigr)
    \Bigl(
       2d_{c_{L-1}} + a_{c_{L-1}} - a_{c_1} - d_{c_1}
    \Bigr)  \\
    \ge &
    \frac{\alpha}{3}z'_1 d_{c_1}
    +
    \sum_{l=2}^{L-1} 
    \bigl(z'_l - z'_{l-1}\bigr)
    \Bigl(
       2d_{c_l} + a_{c_l} - a_{c_1} - d_{c_1}
    \Bigr).
\end{align*}
Next, we add the term $\tfrac{\alpha}{3}\bigl(z'_{L-1} - z_{L-2}\bigr)\bigl(2d_{c_{L-1}} + a_{c_{L-1}} - a_n \bigr)$ to both sides of above inequality.
  \begin{align*}
        & \sum_{l=1}^{L-2} 
    \frac{\alpha}{3}\bigl(z_l - z_{l-1}\bigr)
    \Bigl(2 d_{c_l} + a_{c_l} - a_{c_1} \Bigr) + \frac{\alpha}{3}\bigl(z'_{L-1} - z_{L-2}\bigr)
    \Bigl(
       2 d_{c_l} + a_{c_l} - a_{c_1} - d_{c_1}
    \Bigr) + \\
    & \qquad \frac{\alpha}{3}\bigl(z_L - z'_{L-1}\bigr)
    \Bigl(
       2 d_{c_l} + a_{c_l} - a_{c_1} - d_{c_1}
    \Bigr) \nonumber\\
    \ge & \frac{\alpha}{3} \cdot z'_1 \cdot  d_{c_1} + \sum_{l=2}^{L}   (z'_{l}- z'_{l-1}) \cdot  \left ( 2 d_{c_{l}} + a_{c_l} - a_{c_1} - d_{c_1} \right).
\end{align*}  
Next, we will upper-bound the left-hand-size of the above inequality as follows:
\begin{align*}
    & \sum_{l=1}^{L-2} 
    \frac{\alpha}{3}\bigl(z_l - z_{l-1}\bigr)
    \Bigl(2 d_{c_l} + a_{c_l} - a_{c_1} - d_{c_1} \Bigr) + \frac{\alpha}{3}\bigl(z'_{L-1} - z_{L-2}\bigr)
    \Bigl(
       2 d_{c_{L-1}} + a_{c_{L-1}} - a_{c_1} - d_{c_1}
    \Bigr) +\\
    & \qquad \frac{\alpha}{3}\bigl(z_L - z'_{L-1}\bigr)
    \Bigl(
       2 d_{c_{L}} + a_{c_{L}} - a_{c_1} - d_{c_1}
    \Bigr))  \\
    \leq & \sum_{l=1}^{L-2} 
    \frac{\alpha}{3}\bigl(z_l - z_{l-1}\bigr)
    \Bigl(2 d_{c_l} + a_{c_l} - a_{c_1} - d_{c_1} \Bigr) + \frac{\alpha}{3}\bigl(z_{L-1} - z_{L-2}\bigr)
    \Bigl(
       2 d_{c_{L-1}} + a_{c_{L-1}} - a_{c_1} - d_{c_1}
    \Bigr) +  \\
    & \qquad \frac{\alpha}{3}\bigl(z_L - z_{L-1}\bigr)
    \Bigl(
       2 d_{c_L} + a_{c_L} - a_{c_1} - d_{c_1}
    \Bigr)) \\
    = & \sum_{l=1}^{L} 
    \frac{\alpha}{3}\bigl(z_l - z_{l-1}\bigr)
    \Bigl(2 d_{c_l} + a_{c_l} - a_{c_1} - d_{c_1} \Bigr),
\end{align*}
where the inequality follows since $\Bigl(
       2 d_{c_{L-1}} + a_{c_{L-1}} - a_{c_1} - d_{c_1}
    \Bigr) \leq \Bigl(
       2 d_{c_L} + a_{c_L} - a_{c_1} - d_{c_1}
    \Bigr)$.
    Thus in this case the inequality in Eq.~\eqref{eq:kRental:VD:random:inequality} follows.

\textit{Case 2:} $z'_{L-1} = z_{L-1}$. In this scenario, the inequality follows by the same reasoning as previous steps.

This completes the proof of Eq. 
\eqref{eq:kRental:VD:random:inequality}. Thus, we complete the proof of the lemma.
\end{proof}

Returning to the point where we left off, we have:
\begin{align*}
    \sum_{j \in C_n^{(i)}} \Delta_{\Lambda_{n}^{(i)}}^{(j)} 
    &\ge
    \sum_{l=1}^{L} 
    \frac{\alpha}{3}\bigl(z_l - z_{l-1}\bigr)
    \bigl(
       d_{c_{l}} + a_{c_l} + d_{c_l} - a_{c_1} - d_{c_1}
    \bigr)
    \\[6pt]
    &\ge
    \frac{\alpha}{3}
    \Bigl[
       z'_1\phi(z'_1) 
       +
       \sum_{l=2}^{L} 
          \bigl(z'_{l} - z'_{l-1}\bigr)
          \bigl(
             d_{c_{l}} + a_{c_l} + d_{c_l} - a_{c_1} - d_{c_1}
          \bigr)
    \Bigr]
    \\[6pt]
    &\ge
    \frac{\alpha}{3}
    \Bigl[
       z'_1\phi(z'_1) 
       +
       \sum_{l=2}^{L} 
          \bigl(z'_{l} - z'_{l-1}\bigr)
          \bigl(
             2\phi(z'_l) - \phi(z'_1)
          \bigr)
    \Bigr]
    \\[6pt]
    &\ge
    \frac{\alpha}{3}
    \Bigl[
       z'_1\phi(z'_1) 
       +
       \int_{\eta = z'_1}^{z'_L}
          \bigl(2\phi(\eta) - \phi(z'_1)\bigr)
       d\eta
    \Bigr]
    \\[6pt]
    &=
    \frac{2\alpha}{3}z'_1\phi(z'_1) 
    -
    \frac{\alpha}{3}z'_L\phi(z'_1) 
    +
    \int_{z'_1}^{z'_L}
       \frac{2\alpha}{3}\phi(\eta)
    d\eta,
\end{align*}
where the second and third inequalities follow from 
Lemma~\ref{lem:kRental:VD:z':frac}, 
and the fourth inequality holds because $\phi$ is an increasing function. 

Next, we upper-bound 
$\sum_{j \in B_n^{(i)} \setminus C_n^{(i)}} \Delta_{\Lambda_{n}^{(i)}}^{(j)}$
as follows:
\begin{align*}
    \sum_{j \in B_n^{(i)} \setminus C_n^{(i)}} 
    \Delta_{\Lambda_{n}^{(i)}}^{(j)} 
    \ge 
    \frac{\alpha}{3}d_n
    \max \Bigl\{0,y_n^{(i^*_n)} - z'_L\Bigr\}.
\end{align*}
The above inequality holds because, for each request $j$ in 
$B_n^{(i)} \setminus C_n^{(i)}$, we have $a_{j} + d_{j} > a_n + d_n$ by the definitions of 
$B_n^{(i)}$ and $C_n^{(i)}$. Then, from 
Eqs.~\eqref{eq:kRentalVD:dual:update:2}--\eqref{eq:kRentalVD:dual:update:3}, 
for each such $j$, $\Delta_{\Lambda_{n}^{(i)}}^{(j)}$ is at least 
$\tfrac{\alpha}{3}d_n \hat x_j$. Furthermore,
\begin{align*}
    \sum_{j \in B_n^{(i)} \setminus C_n^{(i)}}  \hat x_j 
    \ge 
    \max \Bigl\{0,y_n^{(i^*_n)} - z'_L\Bigr\}.
\end{align*}

We move on to lower-bound the term $u_n$. From the update rule in Eq.~\eqref{eq:kRentalVD:dual:update:1}, for each request~$n$ we have:
\begin{align*}
    u_n 
    &\ge  
    \frac{\alpha}{3}d_n \hat x_n
    \\[4pt]
    &\ge
    \frac{\alpha}{3}d_n
    \max\Bigl\{0,\phi^{*}\bigl(d_n\bigr)-y_n^{(i^*_n)}\Bigr\},
\end{align*}
where the second inequality follows from the way Algorithm~\ref{alg:kRentalVD-dynamic} sets the value of $\hat x_n$ in Eq.~\eqref{eq:kRental:VD:fractional}.

Combining bounds obtained for the LHS of Eq.~\eqref{eq:kRentalVD:cons:2}, we obtain
\begin{align*}
    u_n 
    +
    \sum_{j = n}^{N}
        \lambda_{j}^{(i)} \indicator{a_n + d_n > a_{j}}
    &\ge
    \frac{\alpha}{3}d_n
    \max\Bigl\{0,\phi^{*}\bigl(d_n\bigr)-y_n^{(i^*_n)}\Bigr\}
    +
    \sum_{j \in B_n^{(i)}} 
       \Delta_{\Lambda_{n}^{(i)}}^{(j)}
    \\[4pt]
    &=
    \frac{\alpha}{3}d_n
    \max\Bigl\{0,\phi^{*}\bigl(d_n\bigr)-y_n^{(i^*_n)}\Bigr\}
    +
    \sum_{j \in C_n^{(i)}} 
        \Delta_{\Lambda_{n}^{(i)}}^{(j)}
    +
    \sum_{j \in B_n^{(i)} \setminus C_n^{(i)}} 
        \Delta_{\Lambda_{n}^{(i)}}^{(j)}
    \\[4pt]
    &\ge
    \frac{\alpha}{3}d_n
    \max\Bigl\{0,\phi^{*}\bigl(d_n\bigr)-y_n^{(i^*_n)}\Bigr\}
    +
    \frac{2\alpha}{3}z'_1\phi\bigl(z'_1\bigr)
    -
    \frac{\alpha}{3}z'_L\phi\bigl(z'_1\bigr) \\
    & \qquad +
    \int_{z'_1}^{z'_L}
        \frac{2\alpha}{3}\phi(\eta)d\eta
    +
    \frac{\alpha}{3}d_n
    \max\Bigl\{0,y_n^{(i^*_n)} - z'_L\Bigr\}
    \\
    &=
    \frac{2\alpha}{3}z'_1\phi\bigl(z'_1\bigr)
    -
    \frac{\alpha}{3}z'_L\phi\bigl(z'_1\bigr)
    +
    \int_{z'_1}^{z'_L}
       \frac{2\alpha}{3}\phi(\eta)d\eta \\
    & \qquad 
    +
    \frac{\alpha}{3}d_n
    \max\Bigl\{0,y_n^{(i^*_n)} - z'_L\Bigr\}
    +
    \frac{\alpha}{3}d_n
    \max\Bigl\{0,\phi^{*}\bigl(d_n\bigr)-y_n^{(i^*_n)}\Bigr\}. \\
    & \ge 
        \frac{2\alpha}{3}z'_1\phi\bigl(z'_1\bigr)
    -
    \frac{\alpha}{3}z'_L\phi\bigl(z'_1\bigr)
    +
    \int_{z'_1}^{z'_L}
       \frac{2\alpha}{3}\phi(\eta)d\eta \\
    &\qquad +
    \frac{\alpha}{3}d_n
    \max\Bigl\{0,\phi^{*}\bigl(d_n\bigr)- z'_L\Bigr\}\\
    & \ge d_n,
\end{align*}
where the last inequality follows from the lemma given below. Consequently, we conclude that the second set of constraints in Eq.~\eqref{eq:kRentalVD:cons:2} is satisfied by the design of system variables in Eqs.~\eqref{eq:kRentalVD:dual:update:1}--\eqref{eq:kRentalVD:dual:update:3}. Therefore, \textbf{if the increasing pricing function $\phi$ satisfies the system of constraints in Theorem~\ref{prop:kRentalVD-Dynamic-CR}, it establishes the $\alpha$-competitiveness of Algorithm~\ref{alg:kRentalVD-dynamic}. Thus, to complete the proof of Theorem~\ref{prop:kRentalVD-Dynamic-CR}, it suffices to prove the following lemma.}

\begin{lemma}
If the $\phi$ function satisfy the constraints in Eqs.~\eqref{eq:kRental:VD:price:design:1}, 
then 
    \begin{align*}
    & \frac{2\alpha}{3}z'_1\phi\bigl(z'_1\bigr)
    -
    \frac{\alpha}{3}z'_L\phi\bigl(z'_1\bigr)
    +
    \int_{z'_1}^{z'_L}
       \frac{2\alpha}{3}\phi(\eta)d\eta 
    +
    \frac{\alpha}{3}d_n
    \max\Bigl\{0,\phi^{*}\bigl(d_n\bigr)-z'_L\Bigr\}  \ge  d_n,
    \\
    & 
    \hspace{+5cm}
    \forall d_{n} \in [d_{\min},d_{\max}],  
    z'_1 \in \bigl(0,z'_L\bigr], z'_L \in \bigl(0,\phi^{*}(d_{min})\bigr].
    \end{align*}
\end{lemma}

\begin{proof}
Let us denote the left-hand side of the above inequality by 
\begin{align*}
\mathfrak{L}(z'_1,z'_L,d_n),
\end{align*}
which is a function of $z'_1$, $z'_L$, and $d_n$. We are going to prove that
\begin{align*}
\mathfrak{L}(z'_1,z'_L,d_n)
\ge
\min\Bigl\{
\mathfrak{L}\Bigl(\frac{z'_L}{2},z'_L,d_n\Bigr),\quad
\mathfrak{L}(z'_L,z'_L,d_n)
\Bigr\}.
\end{align*}

\emph{Step 1: Find the critical point in the interior.} Fix $z'_L$ and $d_n$ and differentiate $L$ with respect to $z'_1$. A straightforward calculation shows
\begin{align*}
\frac{\partial \mathfrak{L}}{\partial z'_1}
=
\frac{\alpha}{3}\phi'(z'_1)\Bigl(2z'_1 - z'_L\Bigr).
\end{align*}
Since $\phi'(z'_1) \neq 0$, setting this derivative to zero forces 
\begin{align*}
z'_1 = \frac{z'_L}{2}.
\end{align*}

A direct substitution $z'_1 = \frac{z'_L}{2}$ simplifies the first two terms and yields
\begin{align*}
\mathfrak{L}\Bigl(\frac{z'_L}{2},z'_L,d_n\Bigr) = \int_{z'_L/2}^{z'_L} \frac{2\alpha}{3}\phi(\eta)d\eta + \frac{\alpha}{3}d_n\max\{0,\phi^{*}(d_n)-z'_L\}.
\end{align*}

\emph{Step 2: Evaluate $\mathfrak{L}$ on the boundary plane where $z'_1 = z'_L$.}  When $z'_1 = z'_L$, the integral term vanishes. One obtains
\begin{align*}
\mathfrak{L}(z'_L,z'_L,d_n) = \frac{\alpha}{3}z'_L\phi(z'_L) + \frac{\alpha}{3}d_n\max\{0,\phi^{*}(d_n)-z'_L\}.
\end{align*}
This can be written as 
\begin{align*}
\frac{\alpha}{3}d_n\phi^*(d_n) - \frac{\alpha}{3}z'_L\Bigl(d_n-\phi(z'_L)\Bigr)
\end{align*}
for values of $z'_L \le \phi^{*}(d_n)$.

\emph{Step 3: Evaluate $\mathfrak{L}$ on the boundary plane where $z'_1$ converges to zero.} It can be verified that $\mathfrak{L}(z'
_1,z'_L,d_n)$ is lower-bounded by $L(z'_L/2,z'_L,d_n)$, as $z'_1$ converges to zero, as follows:
\begin{align*}
     \mathfrak{L}(0,z'_L,d_n)  & =
    \int_{0}^{z'_L}
       \frac{2\alpha}{3}\phi(\eta)d\eta 
    +
    \frac{\alpha}{3}d_n
    \max\Bigl\{0,\phi^{*}\bigl(d_n\bigr)-z'_L\Bigr\}  -     \frac{\alpha}{3}z'_L\phi\bigl(z'_1\bigr) \\
    & \ge  \int_{z'_L/2}^{z'_L}
       \frac{2\alpha}{3}\phi(\eta)d\eta 
    +
    \frac{\alpha}{3}d_n
    \max\Bigl\{0,\phi^{*}\bigl(d_n\bigr)-z'_L\Bigr\}  = L(z'_L/2,z'_L,d_n),
\end{align*}
where the inequality follows from the fact that $\phi$ function is increasing.

Putting together the above results, in the above three steps, since a continuous function on a closed interval attains its minimum either at a critical point or on the boundary, for each fixed $z'_L$ and $d_n$ we have
\begin{align*}
\mathfrak{L}(z'_1,z'_L,d_n) & \ge \min\Bigl\{
\mathfrak{L}\Bigl(\frac{z'_L}{2},z'_L,d_n\Bigr),\quad
\mathfrak{L}(z'_L,z'_L,d_n)
\Bigr\} \\
&=\min\Biggl\{
\frac{\alpha}{3}d_n\phi^*(d_n)
-\frac{\alpha}{3}z'_L\Bigl(d_n-\phi(z'_L)\Bigr), \\
&\quad\quad\quad \int_{z'_L/2}^{z'_L} \frac{2\alpha}{3}\phi(\eta)d\eta
+\frac{\alpha}{3}d_n\max\{0,\phi^{*}(d_n)-z'_L\}
\Biggr\} \\
&\ge d_n,
\end{align*}
where the last inequality follows from the constraints in Eqs.~\eqref{eq:kRental:VD:price:design:1}. 
This completes the proof.
\end{proof}

\section{Proof of Order-Optimality of Algorithm~\ref{alg:kRentalVD-dynamic}}

\label{apx:order:optimal:alg:variable}
In the following corollary, we present an analytical design of $\phi$ that satisfies the system of constraints in Theorem~\ref{prop:kRentalVD-Dynamic-CR}. Although this design is not the optimal design that satisfies the set of constraints in Theorem~\ref{prop:kRentalVD-Dynamic-CR}, it is sufficient to obtain the order optimality of Algorithm~\ref{alg:kRentalVD-dynamic}.

\begin{corollary}
\label{corollary:kRental:VD:design}
\algorithmVD is 
$3\cdot \Big(1+\ln\big(\frac{d_{\max}}{d_{\min}} \big) \Big)$-competitive for the \problemkRentalVD problem if the pricing function $\phi$ is designed as:
\begin{align*}
\phi(y) = d_{\min} \cdot \exp\Bigl(\Bigl[1+\ln\Bigl(\frac{d_{\max}}{d_{\min}}\Bigr)\Bigr]\cdot y - 1\Bigr), \quad  \forall y \in [0,1].
\end{align*}
\end{corollary}

\begin{proof}
We now prove that for the given design of the $\phi$ function, the constraints in Theorem~\ref{prop:kRentalVD-Dynamic-CR} (see Eqs.~\eqref{eq:kRental:VD:price:design:1}) are satisfied when
\begin{align*}
\alpha = 3\cdot\Bigl(1+\ln\frac{d_{\max}}{d_{\min}}\Bigr).
\end{align*}

\medskip

\noindent
\textbf{Proof for the first inequality of Eq.~\eqref{eq:kRental:VD:price:design:1}.} 
The left-hand side of Eq.~\eqref{eq:kRental:VD:price:design:1} can be lower-bounded as follows:
\begin{align*}
    &\int_{\eta=y_1}^{2y_1} \frac{2\alpha}{3}\phi(\eta)d\eta 
    + \frac{\alpha}{3}d_n\Bigl(\phi^*(d_n) - 2y_1\Bigr) \\[1mm]
    \ge & \int_{\eta=0}^{2y_1} \frac{\alpha}{3}\phi(\eta)d\eta 
    + \frac{\alpha}{3}d_n\Bigl(\phi^*(d_n) - 2y_1\Bigr) \\[1mm]
    \ge & \int_{\eta=0}^{\phi^*(d_n)} \frac{\alpha}{3}\phi(\eta)d\eta \\[1mm]
    \ge & d_n,
\end{align*}
where the first and second inequality follows because $\phi$ is an increasing function.
The final inequality is a consequence of the specific design of the $\phi$ function stated in the corollary.

\medskip
\noindent
\textbf{Proof for the second inequality of Eq.~\eqref{eq:kRental:VD:price:design:1}.}
The left-hand side of the second inequality of Eq.~\eqref{eq:kRental:VD:price:design:1} can be lower-bounded as:
\begin{align*}
    &\frac{\alpha}{3}d_n\phi^*(d_n)
    - \frac{\alpha}{3}y_2\Bigl(d_n - \phi(y_2)\Bigr) \\[1mm]
    = & \frac{\alpha}{3}y_2\phi(y_2)
    + \frac{\alpha}{3}d_n\Bigl(\phi^*(d_n) - y_2\Bigr) \\[1mm]
    \ge & \int_{\eta=0}^{\phi^*(d_n)} \frac{\alpha}{3}\phi(\eta)d\eta \\[1mm]
    \ge & d_n,
\end{align*}
where the equality is obtained by rearranging the terms.
The first inequality follows from the increasing nature of $\phi$ and the corresponding definitions of $y_2$ and $\phi^*(d_n)$.
The final inequality is ensured by the design of the $\phi$ function provided in the corollary.

Since the proposed design for the $\phi$ function satisfies both sets of constraints, it follows that Algorithm~\ref{alg:kRentalVD-dynamic} obtains the competitive ratio of $ 3\cdot\Bigl(1+\ln\frac{d_{\max}}{d_{\min}}\Bigr)$ for the given design in the corollary.   
\end{proof}

Next, we prove the following lower bound on the competitiveness of any online algorithm for the \problemkRentalVD problem. 

\begin{proposition}[Lower Bound of \problemkRentalVD] 
Under Assumption \ref{ass:bounded-duration-continuous}, no online algorithm, deterministic or randomized, can obtain a competitive ratio better than $ 1+\ln\left(\frac{d_{\max}}{d_{\min}}\right)$ for the \problemkRentalVD problem.
\end{proposition}

\begin{proof}
Following the same proof structure as in \cite{sun2024static} for the online selection problem, we can prove the lower bound $1+\ln\!\bigl(\tfrac{d_{\max}}{d_{\min}}\bigr)$ on the competitiveness of every online algorithm for the \problemkRentalVD problem. Based on Corollary~\ref{corollary:kRental:VD:design}, we can then claim that Algorithm~\ref{alg:kRentalVD-dynamic} is order-optimal.

We design a set of hard instances for the \problemkRentalVD problem similar to \cite{sun2024static}. Let $\cala(k,d)$ denote a batch of $k$ identical requests, each with a rental duration request $d$ (with $d\in[d_{\min},d_{\max}]$). Divide the uncertainty range $[d_{\min},d_{\max}]$ into $m-1$ sub-ranges of equal length 
\begin{align*}
\Delta_d = \frac{d_{\max}-d_{\min}}{m-1}.
\end{align*}
Let  $ \mathcal{D} := \{d_i\}_{i\in[m]} $,  where $d_i = d_{\min} + (i-1)\Delta_d$ for $i \in [m]$. Define an instance 
\begin{align*}
I_{d_i} := \cala(k,d_1) \oplus \cala(k,d_2) \oplus \cdots \oplus \cala(k,d_i),
\end{align*}
which consists of a sequence of request batches with increasing rental duration requests, arriving in an arbitrarily short interval one after the other (here, $\cala(k,d_i) \oplus \cala(k,d_j)$ denotes a batch $\cala(k,d_i)$ followed by a batch $\cala(k,d_j)$).
In this construction, all requests—from the first batch with valuation $d_1$ to the last batch with valuation $d_m$—arrive within a short time interval $[0,\epsilon]$, where $\epsilon$ is a small constant satisfying $\epsilon < d_{\min}$. 

We consider the collection $\{I_{d_i}\}_{i\in[m]}$ as a family of hard instances for the \problemkRentalVD problem. Using the same proof structure as those of \cite{sun2024static}, we derive the optimal competitive ratio that any online algorithm can achieve on these hard instances. Consequently, we establish the lower bound $1 + \ln\!\bigl(\tfrac{d_{\max}}{d_{\min}}\bigr)$ for the competitive ratio of every online algorithm in this setting.
\end{proof}

By the above lower bound result and Corollary~\ref{corollary:kRental:VD:design}, we argue that \algorithmVD is order-optimal.

\section{Numerical Computation of $\phi$ Satisfying Constraints in Theorem~\ref{prop:kRentalVD-Dynamic-CR}}
\label{apx:numerical:phi:computation}
To numerically obtain a $\phi$ function that satisfies the constraints in Theorem~\ref{prop:kRentalVD-Dynamic-CR} with lowest value of $\alpha$, we discretize the interval $[0,1]$ using a parameter $\epsilon \in (0,1)$, allowing the function $\phi$ to change only at the points $\{\epsilon \cdot i\}_{i=1}^{\lceil 1/\epsilon \rceil}$. 

By doing so, we transform the continuous constraints of Theorem~\ref{prop:kRentalVD-Dynamic-CR} into a finite system of constraints, leading to an optimization problem with a finite number of variables and constraints. The optimal solution for the set of variables $\{\pi^{(\epsilon)}_{i}\}_{i=1}^{\lceil 1/\epsilon \rceil}$ in the following discretized LP provides the desired price function design $\phi$; specifically, the values of $\phi$ at the points $\{\epsilon \cdot i\}_{i=1}^{\lceil 1/\epsilon \rceil}$ are defined according to the solution $\{\pi^{(\epsilon)}_{i}\}_{i=1}^{\lceil 1/\epsilon \rceil}$.

\begin{proposition}
    \label{lem:kRental:VD:discretized_properties}
For any given $\epsilon \in (0,1)$, \algorithmVD is $\alpha^*$-competitive for the \problemkRentalVD problem if the pricing function $\phi$ is designed as 
\begin{align*}
\phi(y) = \pi^{(\epsilon)}_{\lceil \frac{y}{\epsilon} \rceil}, \quad \forall y \in [0,1],
\end{align*}
where the set of prices $\{\pi^{(\epsilon)}_{i}\}_{i=1}^{\lceil 1/\epsilon \rceil}$  and the competitive ratio parameter $\alpha^*$ are the optimal solution to the following LP\footnote{Note that this LP can be solved by standard LP solvers. In addition, it is worth noting that the number of variables in the LP grows on the order of $\lceil 1/\epsilon \rceil$, while the number of constraints increases on the order of $(\lceil 1/\epsilon \rceil)^2$. Standard  techniques such as the interior point method can be employed to solve this optimization problem efficiently.}:
\begin{subequations}
\begin{align}
\label{lp:kRental:VD:discreete}
   \alpha^* =  \min_{\alpha,\ \{\pi^{(\epsilon)}_{i}\}_{i=1}^{\lceil 1/\epsilon \rceil}} \ &  \alpha  \\ 
   \text{s.t.} \qquad 
   & \pi^{(\epsilon)}_{i}  \le \pi^{(\epsilon)}_{i+1},
     \quad  \forall i \in \Bigl\{1,2,\dots,\lceil \tfrac{1}{\epsilon}\rceil - 1\Bigr\},  \\
   & \label{eq:kRental:discretized:cons2} \pi^{(\epsilon)}_{1} = d_{\min},\\
   & \label{eq:kRental:discretized:cons3}\pi^{(\epsilon)}_{\lceil \frac{1}{\epsilon}\rceil} \ge d_{\max},\\
   & \sum_{j=i}^{2i} 
       \frac{2\alpha}{3} \Bigl (\pi^{(\epsilon)}_{j-1}\Bigr) \epsilon + \frac{\alpha}{3}\,\Bigl(\epsilon \cdot l - 2\epsilon\cdot (i+1)\Bigr)\,\pi^{(\epsilon)}_{l-1} \ge \pi^{(\epsilon)}_{l}, \nonumber \\
   & \qquad\qquad \forall i \in \Bigl\{1,\dots,\lfloor \tfrac{l}{2}\rfloor\Bigr\}  \label{eq:kRental:discretized:cons4}  \\ 
   & \frac{\alpha}{3}\,\pi^{(\epsilon)}_{l-1} \epsilon l - \frac{\alpha}{3} \Bigl( \pi^{(\epsilon)}_{l} - \pi^{(\epsilon)}_{i+1} \Bigr)\,\epsilon \cdot (i+1) \ge \pi^{(\epsilon)}_{l}, \nonumber \\
   & \qquad\qquad  \forall i \in \{1,\dots,l\}.  \label{eq:kRental:discretized:cons5}
\end{align}
\end{subequations}
\end{proposition}
\begin{proof}
We need to prove that for each $d_n \in \bigl[d_{\min}, d_{\max}\bigr]$, for all 
\begin{align*}
   y_1 \in \Bigl[0,\frac{\phi^*(d_n)}{2}\Bigr]
   \quad\text{and}\quad
   y_2 \in \Bigl[0,\phi^*(d_n)\Bigr],
\end{align*}
the following constraints for $\phi$ are satisfied:
\begin{align}
\label{eq:kRental:VD:price:design:1:proof}
\int_{\eta=y_1}^{2y_1} \frac{2\alpha}{3}\phi(\eta)d\eta 
+ \frac{\alpha}{3}d_n\Bigl(\phi^*(d_n)-2y_1\Bigr)
&\ge d_n,
\\[6pt]
\label{eq:kRental:VD:price:design:2:proof}
\frac{\alpha}{3}d_n\phi^*(d_n)
- \frac{\alpha}{3}y_2\Bigl(d_n-\phi(y_2)\Bigr)
&\ge d_n.
\end{align}
where $\phi^*(d_n) = \sup\Bigl\{ x \in [0,1]  \mid  \phi(x) \le d_n\Bigr\}$. 

For any value of $d_n \in [d_{\min}, d_{\max}]$, following the design of the $\phi$ function, whose value changes only at the points 
\begin{align*}
\Bigl\{\epsilon, 2\epsilon, \dots, \lceil \tfrac{1}{\epsilon}\rceil \cdot \epsilon \Bigr\},
\end{align*}
we have $\phi^{*}(d_n) = \epsilon \cdot m$ for some $m \in \Bigl\{1,2, \dots, \lceil \tfrac{1}{\epsilon}\rceil\Bigr\}$ because of the constraints in Eq.~\eqref{eq:kRental:discretized:cons2} and Eq.~\eqref{eq:kRental:discretized:cons3}.

Thus, for any $y_1 \le \frac{\phi^*(d_n)}{2} = \frac{\epsilon \cdot m}{2}$, the left-hand side of the constraint in Eq.~\eqref{eq:kRental:VD:price:design:1:proof} can be lower-bounded as follows:
\begin{align*}
    & \int_{\eta=y_1}^{2y_1} \frac{2\alpha}{3}\phi(\eta)d\eta 
      +
      \frac{\alpha}{3}d_n\Bigl(\phi^*(d_n) - 2y_1\Bigr) \\
    =& \int_{\eta=y_1}^{2y_1} \frac{2\alpha}{3}\phi(\eta)d\eta 
      +
      \frac{\alpha}{3}d_n\Bigl(\epsilon \cdot m - 2y_1\Bigr) \\
    \ge& \sum_{j= \lfloor \frac{y_1}{\epsilon} \rfloor \cdot \epsilon}^{ 2\lfloor \frac{y_1}{\epsilon} \rfloor \cdot \epsilon}
           \frac{2\alpha}{3}\Bigl(\phi\bigl(\epsilon \cdot j\bigr) - \phi\bigl(\epsilon \cdot (j-1)\bigr)\Bigr)\epsilon 
      +
      \frac{\alpha}{3}  \phi\bigl(\epsilon \cdot (m-1)\bigr)\Bigl(\epsilon \cdot m - 2\epsilon \cdot \Bigl(\lfloor \frac{y_1}{\epsilon} \rfloor +1 \Bigr)\Bigr) \\
    =& \sum_{j= i \cdot \epsilon}^{ 2i \cdot \epsilon}
           \frac{2\alpha}{3}\Bigl(\phi\bigl(\epsilon \cdot j\bigr) - \phi\bigl(\epsilon \cdot (j-1)\bigr)\Bigr)\epsilon 
      +
      \frac{\alpha}{3}  \phi\bigl(\epsilon \cdot (m-1)\bigr)\Bigl(\epsilon \cdot m - 2\epsilon \cdot (i+1)\Bigr) \\
    \ge& \phi\bigl(\epsilon \cdot m\bigr) \ge d_n,
\end{align*}
where $i = \lfloor \frac{y_1}{\epsilon} \rfloor$. The second inequality follows from some simple algebra, the third inequality follows from the constraint in Eq.~\eqref{eq:kRental:discretized:cons4} (since $i \le \frac{m}{2}$), and the last inequality follows from the definition of $\phi^*$ and $m$.

Thus, the first constraint in Eq.~\eqref{eq:kRental:VD:price:design:1:proof} is satisfied for every $d_n \in [d_{\min}, d_{\max}]$ and every $y_1 \in \Bigl[0,\frac{\phi^*(d_n)}{2}\Bigr]$.

Next, let us prove the second inequality in Eq.~\eqref{eq:kRental:VD:price:design:1}. The LHS of that constraint can be lower-bounded as follows:
\begin{align*}
   & \frac{\alpha}{3}d_n\phi^*(d_n)
- \frac{\alpha}{3}y_2\Bigl(d_n-\phi(y_2)\Bigr) \\
&\ge  
    \frac{\alpha}{3}\phi(\epsilon \cdot (m-1))\epsilon \cdot m
- \frac{\alpha}{3} \epsilon \cdot \Bigl(\lfloor \frac{y_2}{\epsilon} \rfloor+1\Bigr) \Bigl(\phi(\epsilon \cdot m)-\phi(\lfloor \frac{y_2}{\epsilon} \rfloor)\Bigr) \\
& =  
    \frac{\alpha}{3}\phi(\epsilon \cdot (m-1))\epsilon \cdot m
- \frac{\alpha}{3} \epsilon \cdot (i+1) \Bigl(\phi(\epsilon \cdot m)-\phi((i+1) \cdot \epsilon)\Bigr) \\
&\ge \phi(\epsilon \cdot m) \ge d_n,
\end{align*}
where $i = \lfloor \frac{y_2}{\epsilon} \rfloor$, the first inequality follows from simple algebra, the second inequality follows from the constraint in Eq.~\eqref{eq:kRental:discretized:cons5} (since $y_2 \in [0,\phi^*(d_n)]$, thus $i \leq m$), and the final inequality follows from the definition of the $\phi$ function. Thus, the constraint in Eq.~\eqref{eq:kRental:VD:price:design:2:proof} is satisfied for every $d_n \in [d_{\min}, d_{\max}]$ and for every $y_2 \in \Bigl[0,\phi^*(d_n)\Bigr]$. We thus complete the proof.
\end{proof}

\section{Fractional \problemkRentalVD Problem}
\label{apx:krental:frac}
In this section, we study the \problemkRentalVD problem in the fractional setting, where the decision maker can allocate a fractional portion of each unit of resource to a request. We will compete against the optimal clairvoyant algorithm, whose objective value $\opt(I)$ can be computed as follows:
\begin{align} \label{lp:kRentalVD:frac:primal}
 \max_{\mathbf{x}} \quad & \sum_{n \in [N]} d_n \cdot x_n, \\
    s.t. \quad & \sum_{j\in [n]} \hat x_j \cdot 1_{\{d_j + a_j > a_n\}} \leq k,\quad \forall n \in [N], \\
    & x_n \in [0,1],\quad \forall\, n \in [N].
\end{align}

We present
\algorithmVDfrac in Algorithm \ref{alg:kRentalVD-dynamic-fracitonal}, which also adopts a utility-maximizing rule similar to Algorithm \ref{alg:kRentalFD} to produce fractional allocations; the key difference is that the decision for each request now depends on the \emph{total} utilization level of the entire item inventory, rather than on the state of a single unit. 

\begin{algorithm}[H]
    \caption{Duration Oblivious Price-based Algorithm for Fractional \problemkRentalVD  (\algorithmVDfrac)}
    \label{alg:kRentalVD-dynamic-fracitonal}
    
    \textbf{Input:} Pricing function $\phi:[0,1] \rightarrow [d_{min},d_{max}]$

    \While{a new request $n$ arrives}
    {
     Observe the request's duration request, $d_n$.
    
    Compute utilization level of the inventory at the arrival of request $n$, where
      \begin{align}
        y_{n} = \sum_{j = 1}^{n-1} \hat x_j
        \cdot \indicator{\,a_{j} + d_{j} > a_{n}\,} 
      \end{align}
      
      Compute the fractional allocation $\hat x_n$ as follows
    \begin{align}
\hat x_n = \argmax_{x \in \bigl[0,\,\min\{1,\,k - y_n\}\bigr]} \Biggl\{ x \cdot d_n - k\, \int_{\eta=y_n/k}^{(y_n+x)/k}  \phi(\eta)\,d\eta \Biggr\}.
    \end{align}
    
      Allocate the partial allocation $\hat x_n$ to request $n$.
    }
\end{algorithm}

\begin{proposition}
\label{prop:kRentalVD:frac-Dynamic-CR}
\algorithmVDfrac is $\alpha$-competitive for the fractional \problemkRentalVD problem, provided that the price function $\phi$ is increasing and is designed to satisfy the following inequalities (parameterized by $\alpha$) for all $d_n \in \bigl[d_{\min}, d_{\max}\bigr]$ and $ n \in [N] $:
\begin{align}
\label{eq:kRental:VD:frac:price:design:1}
\int_{\eta=y_1}^{2y_1}\alpha\,\phi(\eta)\,d\eta 
+ \frac{\alpha}{2}\,d_n\,\Bigl(\phi^*(d_n)-2\,y_1\Bigr)
&\ge d_n,
\quad \forall\,
y_1 \in \Bigl[0,\,\tfrac{\phi^*(d_n)}{2}\Bigr],\\[6pt]
\label{eq:kRental:VD:frac:price:design:2}
\frac{\alpha}{2}\,d_n\,\phi^*(d_n)
- \frac{\alpha}{2}\,y_2\,\Bigl(d_n-\phi(y_2)\Bigr)
&\ge d_n,
\quad \forall\,
y_2 \in \Bigl[0,\,\phi^*(d_n)\Bigr].
\end{align}
where $\phi^*(d_n) = \sup_{x \in [0,1]}\{\phi(x) \leq d_n \} $.
\end{proposition}

The proof of the above proposition is given in the following.  
As with Theorem~\ref{prop:kRentalVD-Dynamic-CR}, obtaining a closed-form expression for the pricing function $\phi$ that minimizes $\alpha$ while satisfying the constraints in the above proposition is challenging, as it requires solving a system of delayed differential inequalities involving an inverse term.  
Therefore, we adopt the same numerical approach used to construct a valid pricing function $\phi$ in Proposition~\ref{lem:kRental:VD:discretized_properties}.  
Specifically, we discretize the interval $[0,1]$ using a parameter $\epsilon \in (0,1)$, allowing $\phi$ to change only at the points $\{\epsilon \cdot i\}_{i=1}^{\lceil 1/\epsilon \rceil}$.  
Using this discretization, we obtain a pricing function $\phi$ that satisfies the constraints of Proposition~\ref{prop:kRentalVD:frac-Dynamic-CR}.  
The competitive ratio of Algorithm~\ref{alg:kRentalVD-dynamic-fracitonal} using this numerically constructed pricing function as input is shown in Figure~\ref{fig:kRental:VD:comparison:numerical}.

\paragraph{Proof of Proposition~\ref{prop:kRentalVD:frac-Dynamic-CR}.}
The proof structure is similar to the proof of Theorem~\ref{prop:kRentalVD-Dynamic-CR}. Specifically, we still adopt an online primal-dual approach to establish the competitive ratio of Algorithm~\ref{alg:kRentalVD-dynamic-fracitonal}. Consider the dual LP corresponding to the primal LP in Eq.~\eqref{lp:kRentalVD:primal}:
\begin{align}
\label{lp:kRentalVD:frac:dual}
\min_{\mathbf{\lambda,u}} \quad & \sum_{n \in [N]} u_{n} + k \cdot \sum_{n \in [N]} \lambda_n,\\ 
\label{eq:kRentalVD:frac:dual:cons}
\text{s.t.} \quad & u_{n} + \sum_{j= n}^{N} \lambda_{j} \cdot \mathbb{I}_{\{a_n + d_n > a_{j}\}} \ge d_{n}, \quad \forall n \in [N],\\ 
& \lambda_n \ge 0, \quad \forall n \in [N]. \nonumber
\end{align}

Let $\alg(\mathcal{I}) = \sum_{n \in [N]} d_n \cdot \hat x_n$ represent the expected objective value of Algorithm~\ref{alg:kRentalVD-dynamic-fracitonal} on an instance $\mathcal{I}$. 

In the first step, we design a feasible solution for the dual problem, denoted as $\{u_{n},\lambda_{n}\}_{n \in [N]}$, which corresponds to the dual linear program, and define the dual objective
\begin{align*}
   D = \sum_{n \in [N]} u_{n} + k \cdot \sum_{n \in [N]} \lambda_{n}.
\end{align*}

We first establish that for the designed dual solution we have 
$D = \alpha \cdot \alg(\mathcal{I})$ for every instance $\mathcal{I}$ of the problem. 
We then establish that the dual vector $\{u_{n},\lambda_{n}\}_{n \in [N]}$ is feasible for the dual LP in Eq.~\eqref{lp:kRentalVD:frac:dual}. 
Thus, the $\alpha$-competitiveness of Algorithm~\ref{alg:kRentalVD-dynamic-fracitonal} follows. 
Let us initialize all dual variables to zero. 
For each rental request $n$, let $\hat x_n$ be the fractional allocation determined by Algorithm~\ref{alg:kRentalVD-dynamic-fracitonal}. 
We then perform the following updates:
\begin{align}
\label{eq:kRentalVD:frac:dual:update:1}
u_n &= \frac{\alpha}{2}\,\bigl(d_n - \phi\bigl((y_n+\hat x_n)/k\bigr)\bigr)\,\hat x_n,\\[4pt]
\label{eq:kRentalVD:frac:dual:update:2}
\lambda_{j} &\leftarrow \lambda_{j} + \frac{\alpha}{2}\,
\bigl(a_{j+1}-a_{j}\bigr)\,\frac{\hat x_n}{k},
\quad\forall j, n \le j < \nu^*_n, a_{j} < a_n+d_n,\\[4pt]
\label{eq:kRentalVD:frac:dual:update:3}
\lambda_{\nu^*_n} &\leftarrow \lambda_{\nu^*_n}
     + \frac{\alpha}{2}\,
       \bigl(d_n - (a_{\nu^*_n}-a_n)
             + \phi\bigl((y_n+\hat x_n)/k\bigr)\bigr)\,\frac{\hat x_n}{k},
\end{align}
where $\displaystyle \nu^*_n = \max\{\,j\ge n : a_{j}<a_n+d_n\,\}$.

Next, let $\Delta_{n}^D$ denote the increase in the dual solution objective value (Eq. \eqref{lp:kRentalVD:frac:dual}) after updating the variables of the system according to Eqs. \eqref{eq:kRentalVD:frac:dual:update:1}--\eqref{eq:kRental:VD:fractional} for the $n$-th rental request. It follows directly that
\begin{align*}
   \Delta_{n}^{D} = \alpha \, d_n \hat x_n.
\end{align*}
Summing over all requests $n \in [N]$, we have $\alpha \cdot \alg(\mathcal{I}) = \alpha \cdot \sum_{n \in [N]} d_n \cdot \hat x_n = \alpha \cdot D$.

\paragraph{Feasibility of the dual constraints in Eq.~\eqref{eq:kRentalVD:frac:dual:cons}.}
For each request $n$ in instance $\mathcal{I}$, we show that
\begin{align*}
   \sum_{j = n}^{N} \lambda_{j} \cdot \indicator{a_n + d_n > a_{j}} + u_n \ge   d_n.
\end{align*}

Let us define
\begin{align*}
    \Lambda_{n} = \sum_{j = n}^{N} \lambda_{j} \cdot \indicator{a_n + d_n > a_{j}}.
\end{align*}

Let us further define the set of requests $B_n$ by
\begin{align}
    B_n  = \bigl\{ 1 \leq j < n \mid a_{j} + d_{j} > a_n, \hat x_j > 0\bigr\}.
\end{align}
In other words, $B_n$ is the set of all requests who arrive before request~$n$, request rental intervals overlapping the interval of request~$n$, and have a non-zero probability of receiving unit~$i$.

Furthermore, define $C_n \subseteq B_n$ by
\begin{align*}
    C_n = \bigl\{ j \in B_n \mid a_{j} + d_{j} < a_n + d_n  \bigr\}.
\end{align*}
Next, let the set of requests $\{c_l\}_{l \in [L]}$ be defined recursively such that
\begin{align*}
    c_1  & =  \argmin_{j \in C_n}\{a_{j} + d_{j}\},\\
    c_{l} &= \argmin_{\substack{j \in C_n \\ a_{j} > a_{c_{l-1}}}} \{a_{j} + d_{j}\}, \quad\quad 1 < l \le L.
\end{align*}
The request $c_L \in C_n$ is the one for which
$\{j \in C_n \mid a_{j} > a_{c_L}\}$ is empty.
Using the sequence of requests $\{c_l\}_{l \in [L]}$, we partition $C_n$ into $L$ sets:
\begin{align*}
    \mathcal{C}_{l} =  \bigl\{ j \in C_n \mid a_{c_{l-1}} \le a_{j} < a_{c_{l}} \bigr\}, \quad \forall l \in [L].
\end{align*}

Furthermore, define the set of values $\{z_l\}_{l \in [L]}$ by
\begin{align*}
    z_l  = \sum_{\substack{j \in C_n \\ a_{j} < a_{c_{l}}}} \hat x_j /k.
\end{align*}

Let $\Delta_{\Lambda_{n}}^{(j)}$ denote the increase in $\Lambda_{n}$ after processing the rental request of request~$j$ and updating the variables according to Eqs.~\eqref{eq:kRentalVD:frac:dual:update:2}--\eqref{eq:kRentalVD:frac:dual:update:3}. From the way we perform the updates, it follows that
\begin{align}
\sum_{j \in C_n} \Delta_{\Lambda_n}^{(j)}
  &\ge
    \sum_{l=1}^{L}
    \sum_{j \in \mathcal{C}_l}
    \frac{\alpha}{2}\,\frac{\hat x_j}{k}\,
    \Bigl(
      \phi\!\bigl((y_j+\hat x_j)/k\bigr)
      + a_{j}+d_{j}-a_{n}
    \Bigr) \notag \\[6pt]
  &\ge
    \sum_{l=1}^{L}
    \sum_{j \in \mathcal{C}_l}
    \frac{\alpha}{2}\,\frac{\hat x_j}{k}\,
    \bigl(a_{j}+d_{j}-a_{n}\bigr)
    + \frac{\alpha}{2}\!\int_{0}^{z_L}\!\!\phi(\eta)\,d\eta \notag \\[6pt]
  &\ge
    \sum_{l=1}^{L}
    \sum_{j \in \mathcal{C}_l}
    \frac{\alpha}{2}\,\frac{\hat x_j}{k}\,
    \bigl(a_{c_l}+d_{c_l}-a_{n}\bigr)
    + \frac{\alpha}{2}\!\int_{0}^{z_L}\!\!\phi(\eta)\,d\eta \notag \\[6pt]
  &\ge
    \sum_{l=1}^{L}
    \frac{\alpha}{2}\,(z_l - z_{l-1})\,
    \bigl(d_{c_l}+a_{c_l}-a_{c_1}\bigr)
    + \frac{\alpha}{2}\!\int_{0}^{z_L}\!\!\phi(\eta)\,d\eta.
\end{align}
where the first inequality follows from the dual updates in 
Eqs.~\eqref{eq:kRentalVD:frac:dual:update:2}--\eqref{eq:kRentalVD:frac:dual:update:3}, 
the second inequality follows from rearranging the terms and the fact that $\phi$ function is increasing. The third inequality relies on the definition of $c_l$ and $\mathcal{C}_l$ (under which $d_{j} \le d_{c_l}$ and $a_{j} + d_{j} \ge a_{c_l} + d_{c_l}$ for each $j \in \mathcal{C}_l$), 
and the last inequality follows from the definition of the values $\{z_l\}_{l \in [L]}$, taking $z_0 = 0$.

Define the sequence $\{z'_l\}_{l \in [L]}$ recursively as follows:
\begin{align*}
    & z'_L = z_L,\\[6pt]
    & z'_l =
    \begin{cases}
        z_l, 
         \hspace{+2cm} \text{if } 
          \bigl(d_{c_{l+1}} + a_{c_{l+1}}\bigr) - \bigl(d_{c_l} + a_{c_l}\bigr) \ge \phi\bigl(z'_{l+1}\bigr) - \phi\bigl(z_{l}\bigr),
        \\[6pt]
        \phi^{*}\!\Bigl(
          \phi\bigl(z'_{l+1}\bigr) - \bigl[
            \bigl(d_{c_{l+1}} + a_{c_{l+1}}\bigr) - \bigl(d_{c_l} + a_{c_l}\bigr)
          \bigr] \Bigr),
        \quad \text{otherwise},
    \end{cases}
    \quad \forall l \in [L-1].
\end{align*}

\begin{lemma}
\label{lem:kRental:VD:z'}
From the definition of the sequence $\{z'_l\}_{l \in [L]}$, for each $l \in \{1, \dots, L\}$, we have
\begin{align*}
    d_{c_l} \ge \phi\bigl(z'_l\bigr).
\end{align*}

Moreover, 
\begin{align}
    & \sum_{l=1}^{L} \frac{\alpha}{2}\bigl(z_l - z_{l-1}\bigr) \Bigl( d_{c_l} + a_{c_l} - a_{c_1} - d_{c_1} \Bigr) \nonumber \\
\ge \ & \frac{\alpha}{2} \cdot z'_1 \phi(z'_1) + \sum_{l=2}^{L}   
    \bigl(z'_{l}- z'_{l-1}\bigr)  \bigl(a_{c_l} + d_{c_l} - a_{c_1} - d_{c_1} \bigr). \label{eq:kRental:VD:random:inequality:frac}
\end{align}
\end{lemma}

\begin{proof}
Again, we prove the two parts of the lemma separately.

\textbf{Part-I (Proof of $d_{c_l} \ge \phi(z'_l)$).} We show $d_{c_l} \ge \phi\bigl(z'_l\bigr)$ by induction on $L$. For the base case $L = 1$:
\begin{align*}
   z'_1 = z_1,
\end{align*}
and by the definition of $z_1$, the probabilistic utilization level of item $i$ at the arrival of request $c_1$ is at least $z_1$. Since $x_{c_1} \neq 0$, it follows that
\begin{align*}
   d_{c_1} \ge \phi\bigl(z_1\bigr) 
   = \phi\bigl(z'_1\bigr).
\end{align*}

Assume the statement holds for any number of requests in the set $\{c_{l}\}_{l \in [L]}$ up to $M-1$. We prove it for $L = M$. By the same argument as in the base case, we have 
\begin{align*}
   d_{c_M} \ge \phi\bigl(z_M\bigr) 
   = \phi\bigl(z'_M\bigr).
\end{align*}
If $z'_{M-1} = z_{M-1}$, then by the definition of $z_{M-1}$, we immediately get
\begin{align*}
   \phi\bigl(z'_{M-1}\bigr) \le d_{c_{M-1}}.
\end{align*}
Otherwise, suppose
\begin{align*}
   \bigl(d_{c_{M}} + a_{c_{M}}\bigr) 
   - 
   \bigl(d_{c_{M-1}} + a_{c_{M-1}}\bigr) 
   <
   \phi\bigl(z'_{M}\bigr) 
   -  
   \phi\bigl(z_{M-1}\bigr).
\end{align*}
From the definition of $z'_{M-1}$, we have
\begin{align*}
   \phi\bigl(z'_{M-1}\bigr) 
   = 
   \phi\bigl(z'_M\bigr) 
   - 
   \Bigl[
      \bigl(d_{c_{M}} + a_{c_{M}}\bigr) 
      - 
      \bigl(d_{c_{M-1}} + a_{c_{M-1}}\bigr)
   \Bigr].
\end{align*}
Since $\phi\bigl(z'_M\bigr) \le d_{c_M}$ and $a_{c_{M-1}} < a_{c_{M}}$ (by the definition of $c_l$), it follows that
\begin{align*}
   \phi\bigl(z'_{M-1}\bigr)
   \le
   \bigl[\phi(z'_M) - d_{c_M}\bigr]
   +
   \bigl[a_{c_{M-1}} - a_{c_M}\bigr]
   +
   d_{c_{M-1}}
   \le
   d_{c_{M-1}}.
\end{align*}
For the remaining inequalities, for each $l \in \{1,\dots,M-2\}$, the induction hypothesis applies to the set of values 
$\{z_1,\dots,z_{M-2}\} \cup \{z'_{M-1}\}$, ensuring that 
$\phi\bigl(z'_l\bigr) \le d_{c_l}$ for all $l \in \{1,\dots,M-1\}.$

\medskip

\textbf{Part-II (Proof of \eqref{eq:kRental:VD:random:inequality:frac})}. We again use induction on $L$ to prove
\begin{align*}
   & \sum_{l=1}^{L} 
   \frac{\alpha}{2}\bigl(z_l - z_{l-1}\bigr)
   \Bigl(
      d_{c_l} + a_{c_l} - a_{c_1} - d_{c_1}
   \Bigr) \\
   \ge\ & 
   \frac{\alpha}{2} z'_1 \phi(z'_1) 
   + 
   \sum_{l=2}^{L}   
   \bigl(z'_{l}- z'_{l-1}\bigr) 
   \Bigl(
       a_{c_l}+d_{c_l} - a_{c_1} - d_{c_1}
   \Bigr).
\end{align*}
For $L=1$, the statement is trivial since $z_1 = z'_1$. Assume it holds for all sequences $\{z'_l\}$ of length $L-1$. We prove it for $L$.  

From the definition of $z'_{L-1}$, we have  $z'_{L-1} \ge z_{L-1}$. Thus, we will have the following two cases:

\textit{Case 1:} $z'_{L-1} > z_{L-1}$.  
 By the induction hypothesis, for the sequence 
$\{z_1,\dots,z_{L-2}\} \cup \{z'_{L-1}\}$, we have
\begin{align*}
    & \sum_{l=1}^{L-2} 
    \frac{\alpha}{2}\bigl(z_l - z_{l-1}\bigr)
    \Bigl(
       d_{c_l} + a_{c_l} - a_{c_1}
    \Bigr)
    +
    \frac{\alpha}{2}\bigl(z'_{L-1} - z_{L-2}\bigr)
    \Bigl(
       d_{c_{L-1}} + a_{c_{L-1}} - a_{c_1} - d_{c_1}
    \Bigr)  \\
    \ge &
    \frac{\alpha}{2}z'_1 d_{c_1}
    +
    \sum_{l=2}^{L-1} 
    \bigl(z'_l - z'_{l-1}\bigr)
    \Bigl(
       d_{c_l} + a_{c_l} - a_{c_1} - d_{c_1}
    \Bigr).
\end{align*}
Next, we add the term $\tfrac{\alpha}{3}\bigl(z'_{L-1} - z_{L-2}\bigr)\bigl(d_{c_{L-1}} + a_{c_{L-1}} - a_n \bigr)$ to both sides of above inequality.
  \begin{align*}
        & \sum_{l=1}^{L-2} 
    \frac{\alpha}{2}\bigl(z_l - z_{l-1}\bigr)
    \Bigl(d_{c_l} + a_{c_l} - a_{c_1} \Bigr) + \frac{\alpha}{2}\bigl(z'_{L-1} - z_{L-2}\bigr)
    \Bigl(
       d_{c_l} + a_{c_l} - a_{c_1} - d_{c_1}
    \Bigr) + \\
    & \qquad \frac{\alpha}{2}\bigl(z_L - z'_{L-1}\bigr)
    \Bigl(
       d_{c_l} + a_{c_l} - a_{c_1} - d_{c_1}
    \Bigr) \nonumber\\
    \ge & \frac{\alpha}{2} \cdot z'_1 \cdot  d_{c_1} + \sum_{l=2}^{L}   (z'_{l}- z'_{l-1}) \cdot  \left ( d_{c_{l}} + a_{c_l} - a_{c_1} - d_{c_1} \right).
\end{align*}  
Next, we will upper-bound the left-hand-size of the above inequality as follows:
\begin{align*}
    & \sum_{l=1}^{L-2} 
    \frac{\alpha}{2}\bigl(z_l - z_{l-1}\bigr)
    \Bigl( d_{c_l} + a_{c_l} - a_{c_1} - d_{c_1} \Bigr) + \frac{\alpha}{2}\bigl(z'_{L-1} - z_{L-2}\bigr)
    \Bigl(
       d_{c_{L-1}} + a_{c_{L-1}} - a_{c_1} - d_{c_1}
    \Bigr) +\\
    & \qquad \frac{\alpha}{2}\bigl(z_L - z'_{L-1}\bigr)
    \Bigl(
        d_{c_{L}} + a_{c_{L}} - a_{c_1} - d_{c_1}
    \Bigr))  \\
    \leq & \sum_{l=1}^{L-2} 
    \frac{\alpha}{2}\bigl(z_l - z_{l-1}\bigr)
    \Bigl(d_{c_l} + a_{c_l} - a_{c_1} - d_{c_1} \Bigr) + \frac{\alpha}{2}\bigl(z_{L-1} - z_{L-2}\bigr)
    \Bigl(
       d_{c_{L-1}} + a_{c_{L-1}} - a_{c_1} - d_{c_1}
    \Bigr) +  \\
    & \qquad \frac{\alpha}{2}\bigl(z_L - z_{L-1}\bigr)
    \Bigl(
       d_{c_L} + a_{c_L} - a_{c_1} - d_{c_1}
    \Bigr)) \\
    = & \sum_{l=1}^{L} 
    \frac{\alpha}{2}\bigl(z_l - z_{l-1}\bigr)
    \Bigl(d_{c_l} + a_{c_l} - a_{c_1} - d_{c_1} \Bigr),
\end{align*}
where the inequality follows since $\Bigl(
        d_{c_{L-1}} + a_{c_{L-1}} - a_{c_1} - d_{c_1}
    \Bigr) \leq \Bigl(
       d_{c_L} + a_{c_L} - a_{c_1} - d_{c_1}
    \Bigr)$.
    Thus in this case the inequality in Eq.~\eqref{eq:kRental:VD:random:inequality:frac} follows.

\textit{Case 2:} $z'_{L-1} = z_{L-1}$. In this scenario, the inequality follows by the same reasoning as previous steps.

This completes the proof of 
\eqref{eq:kRental:VD:random:inequality:frac}.
\end{proof}

Returning to the point where we left off, we have:
\begin{align*}
    \sum_{j \in C_n^{(i)}} \Delta_{\Lambda_{n}^{(i)}}^{(j)} 
    &\ge
    \sum_{l=1}^{L} 
    \frac{\alpha}{2}\bigl(z_l - z_{l-1}\bigr)
    \bigl(
       a_{c_l} + d_{c_l} - a_{c_1} - d_{c_1}
    \bigr)  + \frac{\alpha}{2} \int_{\eta = 0}^{z'_L} \phi(\eta) d \eta
    \\[6pt]
    &\ge
    \frac{\alpha}{2}
    \Bigl[
       z'_1\phi(z'_1) 
       +
       \sum_{l=2}^{L} 
          \bigl(z'_{l} - z'_{l-1}\bigr)
          \bigl(
              a_{c_l} + d_{c_l} - a_{c_1} - d_{c_1}
          \bigr)
    \Bigr] + \frac{\alpha}{2} \int_{\eta = 0}^{z'_L} \phi(\eta) d \eta
    \\[6pt]
    &\ge
    \frac{\alpha}{2}
    \Bigl[
       z'_1\phi(z'_1) 
       +
       \sum_{l=2}^{L} 
          \bigl(z'_{l} - z'_{l-1}\bigr)
          \bigl(
             \phi(z'_l) - \phi(z'_1)
          \bigr)
    \Bigr]  + \frac{\alpha}{2} \int_{\eta = 0}^{z'_L} \phi(\eta) d \eta
    \\[6pt]
    &\ge
    \frac{\alpha}{2}
    \Bigl[
       z'_1\phi(z'_1) 
       +
       \int_{\eta = z'_1}^{z'_L}
          \bigl(\phi(\eta) - \phi(z'_1)\bigr)
       d\eta
    \Bigr] + \frac{\alpha}{2} \int_{\eta = 0}^{z'_L} \phi(\eta) d \eta
    \\[6pt]
    & \ge 
    \alpha z'_1\phi(z'_1) 
    -
    \frac{\alpha}{2}z'_L\phi(z'_1) 
    +
    \int_{z'_1}^{z'_L}
      \alpha \phi(\eta)
    d\eta,
\end{align*}
where the second and third inequalities follow from 
Lemma~\ref{lem:kRental:VD:z'}, 
and the fourth inequality holds because $\phi$ is an increasing function. 

Next, we lower-bound 
$\sum_{j \in B_n \setminus C_n} \Delta_{\Lambda_{n}}^{(j)}$
as follows:
\begin{align*}
    \sum_{j \in B_n \setminus C_n} 
    \Delta_{\Lambda_{n}}^{(j)} 
    \ge 
    \frac{\alpha}{2}d_n
    \max\!\Bigl\{0,y_n - z'_L\Bigr\}.
\end{align*}
The above inequality holds because, for each request $j$ in 
$B_n \setminus C_n$, we have $a_{j} + d_{j} > a_n + d_n$ by the definitions of 
$B_n$ and $C_n$. Then, from 
Eqs.~\eqref{eq:kRentalVD:frac:dual:update:2}--\eqref{eq:kRentalVD:frac:dual:update:3}, 
for each such $j$, $\Delta_{\Lambda_{n}}^{(j)}$ is at least 
$\tfrac{\alpha}{3}d_n\hat x_j$. Furthermore,
\begin{align*}
    \sum_{j \in B_n \setminus C_n} \hat x_j 
    \ge 
    \max\!\Bigl\{0,y_n - z'_L\Bigr\}.
\end{align*}

We proceed to lower-bound the term $u_n$. From the update rule in Eq.~\eqref{eq:kRentalVD:frac:dual:update:1}, for each request~$n$ we have:
\begin{align*}
    u_n 
    &\ge  
    \frac{\alpha}{2}(d_n - \phi(y_n+\hat x_n/k))\hat x_n
\end{align*}

Combining bounds obtained for the LHS of dual constraint in Eq.~\eqref{eq:kRentalVD:frac:dual:cons}, we obtain
\begin{align*}
u_n
  + \sum_{j = n}^{N}
      \lambda_{j}\,
      \indicator{a_n + d_n > a_{j}}
&\ge
  \frac{\alpha}{2}\bigl(d_n - \phi\bigl((y_n+\hat x_n)/k\bigr)\bigr)\hat x_n
  + \sum_{j \in B_n} \Delta_{\Lambda_{n}}^{(j)}
\\[4pt]
&=
  \frac{\alpha}{2}\bigl(d_n - \phi\bigl((y_n+\hat x_n)/k\bigr)\bigr)\hat x_n
  + \sum_{j \in C_n} \Delta_{\Lambda_{n}}^{(j)}
  + \sum_{j \in B_n \setminus C_n} \Delta_{\Lambda_{n}}^{(j)}
\\[4pt]
&\ge
  \frac{\alpha}{2}\bigl(d_n - \phi(Z)\bigr)\hat x_n
  + \alpha z'_1 \phi(z'_1)
  - \frac{\alpha}{2} z'_L \phi(z'_1)
  + \int_{z'_1}^{z'_L} \!\alpha \phi(\eta)\,d\eta
  + \frac{\alpha}{2} d_n \,
    \max\!\bigl\{0,\,Z - z'_L\bigr\}
\\
&\ge
  \alpha z'_1 \phi(z'_1)
  - \frac{\alpha}{2} z'_L \phi(z'_1)
  + \int_{z'_1}^{z'_L} \!\alpha \phi(\eta)\,d\eta
  + \frac{\alpha}{2} d_n \,
    \max\!\bigl\{0,\,\phi^*(d_n) - z'_L\bigr\}
\\
&\ge d_n,
\end{align*}
where the last inequality follows from the lemma given below. Consequently, we conclude that the set of dual constraints in Eq.~\eqref{eq:kRentalVD:frac:dual:cons} is satisfied for each request $n$ given the design of dual constraints in Eqs.~\eqref{eq:kRentalVD:frac:dual:update:1}--\eqref{eq:kRentalVD:frac:dual:update:3}. \textbf{Therefore, if the increasing pricing function $\phi$ satisfies the system of constraints in Proposition~\ref{prop:kRentalVD:frac-Dynamic-CR}, it establishes the $\alpha$-competitiveness of Algorithm~\ref{alg:kRentalVD-dynamic-fracitonal}. Thus, we complete the proof of Proposition \ref{prop:kRentalVD:frac-Dynamic-CR} if the following lemma follows.}

\begin{lemma}
If the $\phi$ function satisfy the constraints in Eqs.~\eqref{eq:kRental:VD:frac:price:design:1}--\eqref{eq:kRental:VD:frac:price:design:2} 
then 
    \begin{align*}
    & \alpha z'_1\phi(z'_1) 
    -
    \frac{\alpha}{2}z'_L\phi(z'_1) 
    +
    \int_{z'_1}^{z'_L}
      \alpha \phi(\eta)
    d\eta +   \frac{\alpha}{2}d_n
    \max\!\Bigl\{0,\phi^*(d_n) - z'_L\Bigr\} \ge  d_n,
    \\
    & 
    \hspace{+4cm}
    \forall d_{n} \in [d_{\min},d_{\max}],  
    z'_1 \in \bigl(0,z'_L\bigr], z'_L \in \bigl(0,\phi^{*}(d_{min})\bigr].
    \end{align*}
\end{lemma}

\begin{proof}
Let us denote the left-hand side of the above inequality by 
\begin{align*}
\mathfrak{L}(z'_1,z'_L,d_n),
\end{align*}
which is a function of $z'_1$, $z'_L$, and $d_n$. We are going to prove that
\begin{align*}
\mathfrak{L}(z'_1,z'_L,d_n)
\ge
\min\Bigl\{
\mathfrak{L}\Bigl(\frac{z'_L}{2},z'_L,d_n\Bigr),\quad
\mathfrak{L}(z'_L,z'_L,d_n)
\Bigr\},
\end{align*}
for all values of $ z'_1 \in \bigl(0,z'_L\bigr]$.

\emph{Step 1: Find the critical point in the interior.} Fix $z'_L$ and $d_n$ and differentiate $L$ with respect to $z'_1$. A straightforward calculation shows
\begin{align*}
\frac{\partial \mathfrak{L}}{\partial z'_1}
=
\frac{\alpha}{2}\phi'(z'_1)\Bigl(2z'_1 - z'_L\Bigr).
\end{align*}
Since $\phi'(z'_1) \neq 0$, setting this derivative to zero forces 
\begin{align*}
z'_1 = \frac{z'_L}{2}.
\end{align*}

A direct substitution $z'_1 = \frac{z'_L}{2}$ simplifies the first two terms and yields
\begin{align*}
\mathfrak{L}\Bigl(\frac{z'_L}{2},z'_L,d_n\Bigr) = \int_{z'_L/2}^{z'_L} \alpha \cdot \phi(\eta)d\eta + \frac{\alpha}{2}d_n\max\{0,\phi^{*}(d_n)-z'_L\}.
\end{align*}

\emph{Step 2: Evaluate $\mathfrak{L}$ on the boundary plane where $z'_1 = z'_L$.}  When $z'_1 = z'_L$, the integral term vanishes. One obtains
\begin{align*}
\mathfrak{L}(z'_L,z'_L,d_n) = \frac{\alpha}{2}z'_L\phi(z'_L) + \frac{\alpha}{2}d_n\max\{0,\phi^{*}(d_n)-z'_L\}.
\end{align*}
This can be written as 
\begin{align*}
\frac{\alpha}{2}d_n\phi^*(d_n) - \frac{\alpha}{2}z'_L\Bigl(d_n-\phi(z'_L)\Bigr)
\end{align*}
for values of $z'_L \le \phi^{*}(d_n)$.

\emph{Step 3: Evaluate $\mathfrak{L}$ on the boundary plane where $z'_1$ converges to zero.} It can be verified that $\mathfrak{L}(z'
_1,z'_L,d_n)$ is lower-bounded by $L(z'_L/2,z'_L,d_n)$, as $z'_1$ converges to zero, as follows:
\begin{align*}
     \mathfrak{L}(0,z'_L,d_n)  & =
    \int_{0}^{z'_L}
       \alpha \cdot \phi(\eta)d\eta 
    +
    \frac{\alpha}{2}d_n
    \max\Bigl\{0,\phi^{*}\bigl(d_n\bigr)-z'_L\Bigr\}  -     \frac{\alpha}{2}z'_L\phi\bigl(z'_1\bigr) \\
    & \ge  \int_{z'_L/2}^{z'_L}
      \alpha \cdot \phi(\eta)d\eta 
    +
    \frac{\alpha}{2}d_n
    \max\Bigl\{0,\phi^{*}\bigl(d_n\bigr)-z'_L\Bigr\}  = L(z'_L/2,z'_L,d_n),
\end{align*}
where the inequality follows from the fact that $\phi$ function is increasing.

Putting together the above results, in the above three steps, since a continuous function on a closed interval attains its minimum either at a critical point or on the boundary, for each fixed $z'_L$ and $d_n$ we have
\begin{align*}
\mathfrak{L}(z'_1,z'_L,d_n) & \ge \min\Bigl\{
\mathfrak{L}\Bigl(\frac{z'_L}{2},z'_L,d_n\Bigr),\quad
\mathfrak{L}(z'_L,z'_L,d_n)
\Bigr\} \\
&=\min\Biggl\{
\frac{\alpha}{2}d_n\phi^*(d_n)
-\frac{\alpha}{2}z'_L\Bigl(d_n-\phi(z'_L)\Bigr), \\
&\quad\quad\quad \int_{z'_L/2}^{z'_L} \alpha \cdot \phi(\eta)d\eta
+\frac{\alpha}{2}d_n\max\{0,\phi^{*}(d_n)-z'_L\}
\Biggr\} \\
&\ge d_n,
\end{align*}
where the last inequality follows from the constraints in Eq.~\eqref{eq:kRental:VD:frac:price:design:1} and Eq.~\eqref{eq:kRental:VD:frac:price:design:2}. 
This completes the proof of the lemma.
\end{proof}

\section{Connection to the Results of \cite{ekbatani2025}}
\label{apx:comparison}

In the following, building on the techniques from \cite{ekbatani2025}, we reproduce their result for the fractional version of the \problemkRentalVD\ problem using our pseudo-utility maximization approach. We focus on a special case in which rental durations are restricted to integer values and requests arrive at discrete time steps. Accordingly, for each request $n$, we define its rental period as $\calt_n := \{a_n, a_n + 1, \dots, a_n + d_n - 1\}$. Furthermore, we make the following assumption where:
\begin{assumption}
    \label{ass:bounded-duration-oja}
The rental duration of each request $n$, $d_n$ is integral and bounded, i.e., $d_n \in \{1,2,\dots,d_{\max}\}, \forall n\in[N]$.
\end{assumption}

In their work, they consider a parametrized pricing function in the following form
\begin{align}
    \phi(y) = \eta (\beta^{\frac{y}{k}} - 1), y\in [0,k],
\end{align}
where $\eta >0$ and $\beta \ge e$ are two parameters to be determined.

We propose Algorithm~\ref{alg:ota} below, which is also referred to as online forward looking price-based algorithm, \oja. 

\begin{algorithm}[H]
\caption{$\oja$ for Fractional \problemkRentalVD}
\label{alg:ota}

\textbf{Input:} Pricing function $\phi(\cdot)$

\textbf{Initialization:} Set utilization $y_t^{(0)} = 0$ for all $t$.

\While{item $n$ arrives}{
  Observe item $n$'s request $(a_n, d_n)$.

  Determine $\hat x_n$ by solving the pseudo-utility maximization problem:
  \begin{align}
      \label{p:utility-maximization}
  \hat x_n = \argmax_{x \in [0,1]} \left\{ 
  x \cdot d_n 
  - \sum_{t \in \mathcal{T}_n} \int_{y_t^{(n-1)}}^{y_t^{(n-1)} + x} \phi(u)\,du 
  \right\}.
  \end{align}

  Update the utilization profile:
  \[
  y_t^{(n)} =
  \begin{cases}
    y_t^{(n-1)} + \hat x_n & \text{if } t \in \mathcal{T}_n, \\
    y_t^{(n-1)}       & \text{otherwise}.
  \end{cases}
  \]
}
\end{algorithm}

Similar to Algorithm~\ref{alg:kRentalVD-dynamic}, \oja employs a pricing function $\phi$ to determine the fractional allocation $\hat{x}_n$ for each arriving request~$n$. However, it is important to highlight two key differences between Algorithm~\ref{alg:ota} and Algorithm~\ref{alg:kRentalVD-dynamic}.
\begin{itemize}
    \item \textbf{Duration-dependent \textit{vs} Duration-oblivious}.  It is evident that the pseudo-utility maximization problem in Eq.~\eqref{p:utility-maximization} accounts for the item's utilization over the entire duration $\mathcal{T}_n$ of request~$n$. Consequently, the pseudo-cost is explicitly \emph{duration-dependent}. This is in sharp contrast to the \emph{duration-oblivious} nature of Algorithm~\ref{alg:kRentalVD-dynamic}, introduced in Section~\ref{sec:krental:variable} for the \problemkRentalVD\ problem, which considers only the item's utilization at the arrival time of each request~$n$.

    \item \textbf{Different design principles to ensure feasibility}. In \oja, the online decision is solely determined by the pseudo-utility maximization problem, with no additional mechanism to explicitly enforce capacity constraints. This represents another key difference between the design principles of Algorithm~\ref{alg:ota} and Algorithm~\ref{alg:kRentalVD-dynamic}. Intuitively, if the pricing function is designed to be sufficiently steep (i.e., with large enough $\eta$ and $\beta$), the feasibility of the online solution can be guaranteed. However, increasing $\eta$ and $\beta$ also leads to a larger competitive ratio, which is undesirable. Lemma~\ref{lem:feasibility} provides a tighter condition on $\eta$ and $\beta$ to ensure the feasibility of \oja.
\end{itemize}

\begin{lemma}
    \label{lem:feasibility}
     \oja is $(1+\eta) \ln(\beta)$-competitive and generates  feasible online solutions, with respect to the resource constraint, if the parameters of the pricing function satisfies
    \begin{align}
    \label{eq:cond}
        \ln\beta \ge - \ln \left(\prod_{i\in[d_{\max}]} (1 - \frac{1}{i(1 + \eta)}) \right).
    \end{align}
\end{lemma}

\begin{proof}
In the following, we first prove that \oja generates a feasible solution and then we prove the competitive ratio of the algorithm.

\paragraph{Feasibility of \oja.} 
We only need to show that after processing each request $n$, the utilization at the arrival time $a_n$ is no larger than $k$, i.e., $y_{a_n}^{(n)} \le k, \forall n\in[N]$ since $y_{a_n}^{(n)}$ is the largest utilization over time after processing each request $n$ in the online $k$-rental problem. To show this, we prove a stronger claim that after processing request $n$, the utilization difference between two time slots $a_n + \tau$ and $a_n + d$ is upper bounded as follows
\begin{align}
\label{eq:claim1}
    y_{a_n + \tau}^{(n)} - y_{a_n+d}^{(n)} \le - \frac{k}{\ln\beta} \ln \left(\prod_{i\in[d - \tau]} (1 - \frac{1}{i(1 + \eta)}) \right),  \quad 0 \le \tau \le d. 
\end{align}
Since the rental period of an item is at most $d_{\max}$, we have $y_{a_n + d_{\max}}^{(n)} = 0$. 
If the above claim in Eq.  \eqref{eq:claim1} and the condition in Eq. \eqref{eq:cond} both hold, we can show the online solution is feasible by noting
$$y_{a_n}^{(n)} = y_{a_n}^{(n)} - y_{a_n+d_{\max}}^{(n)} \le - \frac{k}{\ln\beta} \ln \left(\prod_{i\in[d_{\max}]} (1 - \frac{1}{i(1 + \eta)}) \right) \le k,$$
where the first and second inequalities are due to Eqs. \eqref{eq:claim1} and \eqref{eq:cond}, respectively. 

Next, we prove the claim in Eq. \eqref{eq:claim1} by induction. We refer to the processing of request $n$ as step $n$. In the base case when $n = 0$, the initial utilization of all slots is $y_t^{(0)} = 0, \forall t\in[T]$, and the claim holds. Now suppose the claim holds in step $n-1$. Consider the following cases.

\textbf{Case I.} Request $n$ is completely rejected, i.e., $\hat x_n = 0$. In this case, $y_t^{(n-1)} = y_t^{(n)}, \forall t\in[T]$, and thus the claim holds in step $n$.

\textbf{Case II.} Request $n$ is partially or completely accepted, i.e., $\hat x_n \in (0, 1]$. Based on the pseudo-utility maximization in \oja, we must have $v_n \ge \sum_{t\in\calt_n} \phi(y_t^{(n)})$ (if $\hat x_n < 1$, $v_n = \sum_{t\in\calt_n} \phi(y_t^{(n)})$; if $\hat x_n = 1$, $v_n \ge \sum_{t\in\calt_n} \phi(y_t^{(n)})$).
Then we have
\begin{align}
\label{eq:lem1}
d_n  \ge  \sum_{t\in\calt_n} \phi(y_t^{(n)}) = \sum_{t \in \calt_n} \eta (\beta^{\frac{y_t^{(n)}}{k}} - 1). 
\end{align}
Note that $y_{a_n+\tau}^{(n)} - y_{a_n+d}^{(n)} = (y_{a_n+\tau}^{(n -1)} + \hat x_n) - (y_{a_n+d}^{(n-1)} + \hat x_n) = y_{a_n+\tau}^{(n -1)} - y_{a_n+d}^{(n-1)}$ if $a_n+\tau, a_n + d \in \calt_n = \{a_n,\dots,a_n+d_n-1\}$ and $y_{a_n+\tau}^{(n)} - y_{a_n+d}^{(n)} = y_{a_n+\tau}^{(n -1)} - y_{a_n+d}^{(n-1)}$ if $a_n+\tau, a_n + d \in \{a_n+d_n,\dots,T\}$. In these cases, the utilization difference between $y_{a_n+\tau}^{(n)}$ and $y_{a_n+d}^{(n)}$ is the same as the previous step, and the induction holds.
Thus, the only interesting case is when 
$$a_n + \tau < a_n + d_n \le a_n + d,$$
where the utilization of slot $a_n + \tau$ increases due to newly leased item while the utilization of $a_n + d$ remains unchanged. 

From Eq. \eqref{eq:lem1}, we can have
\begin{subequations}
\begin{align}
\frac{1}{\eta} d_n  &\ge  \sum_{t = a_{n}}^{a_{n} + d_n - 1} (\beta^{\frac{y_t^{(n)}}{k}} - 1)\\
&= \sum_{t = a_{n}}^{a_n + \tau -1} (\beta^{\frac{y_t^{(n)}}{k}} - 1) + \sum_{t = a_n + \tau}^{a_n + d_n - 1} (\beta^{\frac{y_t^{(n)}}{k}} - 1)\\
\label{eq:lem-eq1}
&\ge  - d_n  + \tau \beta^{\frac{y_{a_n+\tau}^{(n)}}{k}} + \sum_{t = a_n + \tau}^{a_n + d_n - 1} \beta^{\frac{y_t^{(n)}}{k}}\\
\label{eq:lem-eq2}
&\ge - d_n  + \tau \beta^{\frac{y_{a_n + \tau}^{(n)}}{k}} + \sum_{t = a_n + \tau}^{a_n + d_n - 1} \beta^{\frac{y_{a_n + \tau}^{(n)}}{k}  + \frac{1}{\ln\beta} \ln \left(\prod_{i\in[t-a_n - \tau]} (1 - \frac{1}{i(1 + \eta)}) \right)}\\
&= - d_n + \beta^{\frac{y_{a_n + \tau}^{(n)}}{k} }\left[\tau +\sum_{t' = 0}^{d_n - \tau - 1} \prod_{i\in[t']} (1 - \frac{1}{i(1 + \eta)}) \right]\\
\label{eq:lem-eq3}
& = - d_n + \beta^{\frac{y_{a_n + \tau}^{(n)}}{k} } \left[\tau + \frac{d_n - \tau}{1- \frac{1}{1 + \eta}} \cdot \prod_{i\in[d_n - \tau]} (1 - \frac{1}{i(1 + \eta)}) \right]\\
\label{eq:lem-eq4}
& = -d_n + \beta^{\frac{y_{a_n + \tau}^{(n)}}{k} } \cdot \frac{1 + \eta}{\eta}\left[\tau \cdot \frac{\eta}{1 + \eta} + (d_n - \tau) \cdot \prod_{i\in[d_n - \tau]} (1 - \frac{1}{i(1 + \eta)}) \right] \\
& \ge -d_n + \beta^{\frac{y_{a_n + \tau}^{(n)}}{k} } \cdot \frac{1 + \eta}{\eta} \cdot d_n \cdot  \prod_{i\in[d_n - \tau]} (1 - \frac{1}{i(1 + \eta)}),
\end{align} 
\end{subequations}
which gives
\begin{align*}
  y_{a_n + \tau}^{(n)}\le - \frac{k}{\ln\beta} \cdot \ln \prod_{i\in[d_n - \tau]} (1 - \frac{1}{i(1 + \eta)}). 
\end{align*} 
In the above set of equations, the inequality~\eqref{eq:lem-eq1} holds since the utilization is monotonically non-increasing over time, and thus $y_{t}^{(n)} \ge y_{a_n + \tau}^{(n)}, \forall t = a_n,\dots,a_n + \tau - 1$. The inequality~\eqref{eq:lem-eq2} is obtained by applying the induction hypothesis, i.e., $\forall t = a_{n}+\tau,\dots,a_n + d_n -1$,
\begin{align*}
y_t^{(n)} = \hat x_n + y_t^{(n-1)} \ge  \hat x_n + y_{a_n + \tau}^{(n-1)} + \frac{k}{\ln\beta} \ln \left(\prod_{i\in[t- a_n - \tau]} (1 - \frac{1}{i(1 + \eta)}) \right).
\end{align*}
The equality~\eqref{eq:lem-eq3} holds due to the following equation that can be shown by induction: 
\begin{align*}
    \sum_{t = 0}^{b-1} \prod_{i\in[t]} (1 - \frac{a}{i}) = \frac{b}{1-a} \prod_{i\in[b]} (1 - \frac{a}{i}), \forall b\in\mathbb{N}, a\ge 0. 
\end{align*}
The last inequality holds since $\prod_{i\in[d_n - \tau]} (1 - \frac{1}{i(1 + \eta)}) \le \frac{\eta}{1 + \eta}$, where $d_n - \tau \ge 1$. 

Thus, we have
\begin{align}
 y_{a_n + \tau}^{(n)} -  y_{a_n + d}^{(n)} \le y_{a_n + \tau}^{(n)} \le  - \frac{k}{\ln\beta} \cdot \ln \prod_{i\in[d_n - \tau]} (1 - \frac{1}{i(1 + \eta)}) \le - \frac{k}{\ln\beta} \cdot \ln \prod_{i\in[d - \tau]} (1 - \frac{1}{i(1 + \eta)}), 
\end{align}
which shows that the induction holds in step $n$. This completes the feasibility proof.

\paragraph{Proof of $(1+\eta) \cdot \ln(\beta)$-competitiveness of \oja.}
We consider the following relaxed primal problem and its corresponding duel problem:
\begin{subequations}
\label{p:primal}
\begin{align}
\max_{x_n \ge 0} \quad&\sum_{n\in[N]} d_n x_n\\
{\rm s.t.} \quad&\sum_{n\in [N]: t\in \calt_n} x_n \le k, \forall t\in[T], \quad (\lambda_t)\\
& x_n \le 1, \forall n\in[N], \quad (\gamma_n)
\end{align}
\end{subequations}
\begin{subequations}
\label{p:dual}
\begin{align}
\min_{\lambda_t,\gamma_n \ge 0} \quad& k \sum_{t\in[T]} \lambda_t +\sum_{n\in[N]} \gamma_n\\
{\rm s.t.} \quad& d_n \le \gamma_n + \sum_{t\in\calt_n} \lambda_t, \forall t\in[T], n\in[N].
\end{align}
\end{subequations}

We construct a set of feasible dual variables as follows:
\begin{align*}
    \lambda_t &= \phi(y_t^{(N)}), \qquad \forall t\in[T], \\
    \gamma_n &= 
    \begin{cases}
    0 &  0 \le \hat x_n <1,\\
    d_n - \sum_{t\in\calt_n}\phi(y_t^{(n)}) & \hat x_n = 1,
    \end{cases} \qquad \forall n\in[N],\quad
\end{align*}
where $\hat x_n$ is the online solution of \oja and $y_t^{(n)}$ is the utilization of time $t$ after processing the $n$-th item by \oja.
We show in the following that the constructed dual solution is feasible: 

\textbf{Case I.} If $\hx = 0$, then we have $\gamma_n = 0$ and $d_n \le \sum_{t\in\calt_n}\phi(y_t^{(n)})$ (based on~\eqref{p:utility-maximization}). Thus,
\begin{align*}
 d_n \le \sum_{t\in\calt_n}\phi(y_t^{(n)}) \le \sum_{t\in\calt_n}\phi(y_t^{(N)}) = \sum_{t\in\calt_n} \lambda_t + \gamma_n.   
\end{align*}

\textbf{Case II.} If $0 < \hat x_n < 1$, then we have $\gamma_n = 0$ and $d_n = \sum_{t\in\calt_n}\phi(y_t^{(n)})$. Thus,
\begin{align*}
 d_n = \sum_{t\in\calt_n}\phi(y_t^{(n)}) \le \sum_{t\in\calt_n} \lambda_t + \gamma_n. 
\end{align*}

\textbf{Case III.} If $\hat x_n = 1$, then we have $\gamma_n = d_n -\sum_{t\in\calt_n}\phi(y_t^{(n)}) \ge 0$. Thus,
\begin{align*}
 d_n = \sum_{t\in\calt_n}\phi(y_t^{(n)}) + \gamma_n \le \sum_{t\in\calt_n}\lambda_t + \gamma_n. 
\end{align*}

The dual objective at the feasible dual solution provides an upper bound on the relaxed primal and the original primal objective. Thus, we have
\begin{align*}
\opt &\le k\sum_{t\in[T]} \phi(y_t^{(N)}) + \sum_{n\in[N]} \gamma_n\\
& = \sum_{n\in[N]} \left[k\sum_{t\in\calt_n}(\phi(y_t^{(n)}) - \phi(y_t^{(n-1)})) + \gamma_n\right],
\end{align*}
which holds as $\phi(0) = 0$.
We next show that 
$  k\sum_{t\in\calt_n}(\phi(y_t^{(n)}) - \phi(y_t^{(n-1)})) + \gamma_n \le (1+\eta)\ln\beta \cdot \hat x_n d_n , \forall n\in[N]$.

Consider the following three cases.

\textbf{Case I.} If $\hat x_n = 0$, then we have $ k\sum_{t\in\calt_n}(\phi(y_t^{(n)}) - \phi(y_t^{(n-1)})) + \gamma_n = (1+\eta)\ln\beta \cdot \hat x_n d_n =0$.

\textbf{Case II.} If $0< \hat x_n <1$, we have $\gamma_n = 0$ and $d_n = \sum_{t\in\calt_n} \phi(y_t^{(n)})$. This gives
\begin{align*}
 k\sum_{t\in\calt_n}(\phi(y_t^{(n)}) - \phi(y_t^{(n-1)})) &= k \eta \sum_{t\in\calt_n} \beta^{\frac{y_t^{(n)}}{k}}(1 - \beta^{-\frac{\hat x_n}{k}})\\
&\le k \eta \sum_{t\in\calt_n} \beta^{\frac{y_t^{(n)}}{k}}\frac{\hat x_n}{k} \ln\beta\\
&= \ln\beta \cdot \sum_{t\in\calt_n} \hat x_n \eta [\beta^{\frac{y_t^{(n)}}{k}} - 1] + \ln\beta \cdot \eta \sum_{t\in\calt_n} \hat x_n\\
&= \ln\beta \cdot \hat x_n \cdot\sum_{t\in\calt_n} \phi(y_t^{(n)}) + \ln\beta \cdot \eta \cdot \hat x_n \cdot d_n \\
&\le \ln\beta \cdot\hat x_n d_n + \ln\beta \cdot \eta \cdot\hat x_n d_n = (1+\eta)\ln\beta \cdot \hat x_n d_n , 
\end{align*}    
where the last inequity holds since  $d_n = \sum_{t\in\calt_n} \phi(y_t^{(n)})$.

\textbf{Case III.} If $\hat x_n = 1$, we have $\gamma_n = d_n - \sum_{t\in\calt_n} \phi(y_t^{(n)})$ and $d_n > \sum_{t\in\calt_n} \phi(y_t^{(n)})$. This gives
\begin{align*}
k\sum_{t\in\calt_n}(\phi(y_t^{(n)}) - \phi(y_t^{(n-1)})) + \gamma_n
&= k \eta \sum_{t\in\calt_n} \beta^{\frac{y_t^{(n)}}{k}}(1 - \beta^{-\frac{1}{k}}) + \gamma_n\\
&\le k \eta \sum_{t\in\calt_n} \beta^{\frac{y_t^{(n)}}{k}}\frac{1}{k} \ln\beta + \gamma_n\\
&= \ln\beta \cdot \sum_{t\in\calt_n} \eta [\beta^{\frac{y_t^{(n)}}{k}} - 1] + \ln\beta \cdot \eta \cdot d_n + \gamma_n\\
&= \ln\beta \cdot\sum_{t\in\calt_n} \phi(y_t^{(n)}) + \ln\beta \cdot \eta \cdot d_n + \gamma_n \\
&= \ln\beta \cdot (d_n - \gamma_n) + \ln\beta \cdot \eta \cdot d_n + \gamma_n\\
&\le \ln\beta \cdot d_n + \ln\beta \cdot \eta \cdot d_n + (1-\ln\beta)\gamma_n\\
&\le (1+\eta)\ln\beta \cdot d_n , 
\end{align*}    
where the last inequality holds when $\beta \ge e$.

This completes the proof of Lemma \ref{lem:feasibility}. 
\end{proof}

Based on Lemma~\ref{lem:feasibility}, we can design the parameters of the pricing function by formulating and solving the following optimization problem, which aims to minimize the competitive ratio of \oja:
\begin{subequations}
\begin{align}
\min_{\beta \ge e, \eta > 0} \quad &(1 + \eta) \ln\beta\\
{\rm s.t.}\quad & \ln\beta \ge - \ln \left(\prod_{i\in[d_{\max}]} (1 - \frac{1}{i(1 + \eta)}) \right).
\end{align}   
\end{subequations}
Thus, using the standard online primal-dual analysis, the results of \cite{ekbatani2025} can be extended to the \emph{fractional} setting through the same pseudo-utility maximization framework as in Algorithm \ref{alg:kRentalFD} and Algorithm \ref{alg:kRentalVD-dynamic}. 

We end this section with a discussion of some interesting future directions. Based on the results of this work and those of \cite{ekbatani2025}, a natural direction for future work is to develop a rounding scheme that converts the fractional algorithm such as Algorithm \ref{alg:ota} into integral decisions, thereby addressing the integral version of the \problemkRentalVD problem. However, due to the impossibility result in Theorem \ref{prop:rounding:impossiblity}, any such rounding scheme will necessarily incur some loss in performance. An alternative and promising direction is to adopt the relax-and-round approach from Section \ref{sec:krental:variable}, where the rounding step is integrated into the design of the fractional solution. This would require a more sophisticated treatment of the pricing function $\phi$, both to ensure feasibility and to optimize the competitive ratio. 


\clearpage

\begin{thebibliography}{10}

\bibitem{Alaei2014}
Saeed Alaei.
\newblock Bayesian combinatorial auctions: Expanding single buyer mechanisms to
  many buyers.
\newblock {\em SIAM Journal on Computing}, 43(2):930--972, 2014.

\bibitem{ball2009toward}
Michael~O Ball and Maurice Queyranne.
\newblock Toward robust revenue management: Competitive analysis of online
  booking.
\newblock {\em Operations Research}, 57(4):950--963, 2009.

\bibitem{huang2025prophet}
Ziyun Chen, Zhiyi Huang, Dongchen Li, and Zhihao~Gavin Tang.
\newblock {\em Prophet Secretary and Matching: the Significance of the Largest
  Item}, pages 1371--1401.

\bibitem{delong2022}
Steven Delong, Alireza Farhadi, Rad Niazadeh, and Balasubramanian Sivan.
\newblock Online bipartite matching with reusable resources.
\newblock In {\em Proceedings of the 23rd ACM Conference on Economics and
  Computation}, EC '22, page 962–963, New York, NY, USA, 2022. Association
  for Computing Machinery.

\bibitem{eden2021economics}
Alon Eden, Michal Feldman, Amos Fiat, and Kineret Segal.
\newblock An economics-based analysis of ranking for online bipartite matching.
\newblock In {\em Symposium on Simplicity in Algorithms (SOSA)}, pages
  107--110. SIAM, 2021.

\bibitem{ekbatani2025}
Farbod Ekbatani, Yiding Feng, Ian Kash, and Rad Niazadeh.
\newblock Online job assignment.
\newblock 2025.

\bibitem{el2001optimal}
Ran El-Yaniv, Amos Fiat, Richard~M Karp, and Gordon Turpin.
\newblock Optimal search and one-way trading online algorithms.
\newblock {\em Algorithmica}, 30:101--139, 2001.

\bibitem{fahrbach2022edge}
Matthew Fahrbach, Zhiyi Huang, Runzhou Tao, and Morteza Zadimoghaddam.
\newblock Edge-weighted online bipartite matching.
\newblock {\em Journal of the ACM}, 69(6):1--35, 2022.

\bibitem{Faigle96}
Ulrich Faigle, Renate Garbe, and Walter Kern.
\newblock Randomized online algorithms for maximizing busy time interval
  scheduling.
\newblock {\em Computing}, 56(2):95--104, 1996.

\bibitem{Feldman2021}
Moran Feldman, Ola Svensson, and Rico Zenklusen.
\newblock Online contention resolution schemes with applications to bayesian
  selection problems.
\newblock {\em SIAM Journal on Computing}, 50(2):255--300, 2021.

\bibitem{fu2024prophet}
Hu~Fu, Pinyan Lu, Zhihao~Gavin Tang, Hongxun Wu, Jinzhao Wu, and Qianfan Zhang.
\newblock Sample-based matroid prophet inequalities.
\newblock In {\em Proceedings of the 25th ACM Conference on Economics and
  Computation}, EC '24, page 781, New York, NY, USA, 2024. Association for
  Computing Machinery.

\bibitem{gaoOCS}
Ruiquan Gao, Zhongtian He, Zhiyi Huang, Zipei Nie, Bijun Yuan, and Yan Zhong.
\newblock Improved online correlated selection.
\newblock {\em CoRR}, abs/2106.04224, 2021.

\bibitem{gardner1970mathematical}
Martin Gardner.
\newblock Mathematical games.
\newblock {\em Scientific american}, 222(6):132--140, 1970.

\bibitem{Garofalakis2002}
Minos Garofalakis, Yannis Ioannidis, Banu Özden, and Avi Silberschatz.
\newblock Competitive on-line scheduling of continuous-media streams.
\newblock {\em Journal of Computer and System Sciences}, 64(2):219--248, 2002.

\bibitem{Goyal2020OIS}
Shashank Goyal and Diwakar Gupta.
\newblock The online reservation problem.
\newblock {\em Algorithms}, 13(10), 2020.

\bibitem{goyal2020}
Vineet Goyal, Garud Iyengar, and Rajan Udwani.
\newblock Asymptotically optimal competitive ratio for online allocation of
  reusable resources.
\newblock {\em Operations Research}, 2024.

\bibitem{Gupta2016}
Diwakar Gupta and Fei Li.
\newblock Reserve driver scheduling.
\newblock {\em IIE Transactions}, 48(3):193--204, 2016.

\bibitem{Huang_2024_Survey}
Zhiyi Huang, Zhihao~Gavin Tang, and David Wajc.
\newblock Online matching: A brief survey.
\newblock {\em SIGecom Exch.}, 22(1):135–158, October 2024.

\bibitem{huang2024adwords}
Zhiyi Huang, Qiankun Zhang, and Yuhao Zhang.
\newblock Adwords in a panorama.
\newblock {\em SIAM Journal on Computing}, 53(3):701--763, 2024.

\bibitem{huo2023online}
Tianming Huo and Wang~Chi Cheung.
\newblock Online reusable resource assortment planning with customer-dependent
  usage durations.
\newblock SSRN Working Paper, July 9 2023.

\bibitem{jazi2025posted}
Hossein~Nekouyan Jazi, Bo~Sun, Raouf Boutaba, and Xiaoqi Tan.
\newblock Posted price mechanisms for online allocation with diseconomies of
  scale.
\newblock In {\em Proceedings of the ACM Web Conference 2025 (WWW '25)}, New
  York, NY, USA, April 2025. ACM.

\bibitem{Jiashuo2023prophet}
Jiashuo Jiang, Will Ma, and Jiawei Zhang.
\newblock Tightness without counterexamples: A new approach and new results for
  prophet inequalities.
\newblock {\em Mathematics of Operations Research}, 0(0):null, 0.

\bibitem{jiang2021online}
Zhihao Jiang, Pinyan Lu, Zhihao~Gavin Tang, and Yuhao Zhang.
\newblock Online selection problems against constrained adversary.
\newblock In {\em International Conference on Machine Learning}, ICML '21,
  pages 5002--5012. PMLR, 2021.

\bibitem{karp90}
R.~M. Karp, U.~V. Vazirani, and V.~V. Vazirani.
\newblock An optimal algorithm for on-line bipartite matching.
\newblock In {\em Proceedings of the Twenty-Second Annual ACM Symposium on
  Theory of Computing}, STOC '90, page 352–358, New York, NY, USA, 1990.
  Association for Computing Machinery.

\bibitem{Lipton1994OnlineIS}
Richard~J. Lipton and Andrew Tomkins.
\newblock Online interval scheduling.
\newblock In {\em ACM-SIAM Symposium on Discrete Algorithms}, 1994.

\bibitem{lorenz2009optimal}
Julian Lorenz, Konstantinos Panagiotou, and Angelika Steger.
\newblock Optimal algorithms for $k$-search with application in option pricing.
\newblock {\em Algorithmica}, 55(2):311--328, 2009.

\bibitem{ma2024randomizedroundingapproachesonline}
Will Ma.
\newblock Randomized rounding approaches to online allocation, sequencing, and
  matching, 2024.

\bibitem{ma2020algorithms}
Will Ma and David Simchi-Levi.
\newblock Algorithms for online matching, assortment, and pricing with tight
  weight-dependent competitive ratios.
\newblock {\em Operations Research}, 68(6):1787--1803, 2020.

\bibitem{sun2024static}
Bo~Sun, Hossein~Nekouyan Jazi, Xiaoqi Tan, and Raouf Boutaba.
\newblock Static pricing for online selection problem and its variants.
\newblock In {\em The 20th Conference on Web and Internet Economics, WINE 2024,
  Edinburgh, United Kingdom Proceedings ,December 2-5, 2024}. Springer, 2024.

\bibitem{sun2022online}
Bo~Sun, Lin Yang, Mohammad Hajiesmaili, Adam Wierman, John~CS Lui, Don Towsley,
  and Danny~HK Tsang.
\newblock The online knapsack problem with departures.
\newblock {\em Proceedings of the ACM on Measurement and Analysis of Computing
  Systems (SIGMETRICS '22)}, 6(3):1--32, 2022.

\bibitem{tan2020mechanism}
Xiaoqi Tan, Bo~Sun, Alberto Leon-Garcia, Yuan Wu, and Danny~HK Tsang.
\newblock Mechanism design for online resource allocation: A unified approach.
\newblock {\em Proceedings of the ACM on Measurement and Analysis of Computing
  Systems (SIGMETRICS '20)}, 4(2):1--46, 2020.

\bibitem{tan2023threshold}
Xiaoqi Tan, Siyuan Yu, Raouf Boutaba, and Alberto Leon-Garcia.
\newblock Threshold policies with tight guarantees for online selection with
  convex costs.
\newblock In {\em 19th International Conference on Web and Internet Economics
  (WINE '23), Shanghai, China, December 4-8, 2008}. Springer, 2023.

\bibitem{zhou2008budget}
Yunhong Zhou, Deeparnab Chakrabarty, and Rajan Lukose.
\newblock Budget constrained bidding in keyword auctions and online knapsack
  problems.
\newblock In {\em 4th International Conference on Web and Internet Economics
  (WINE '08), Shanghai, China, December 17-20, 2008. Proceedings 4}, pages
  566--576. Springer, 2008.

\end{thebibliography}
\end{document}